%% file: main.tex
\documentclass{article}
\usepackage[utf8]{inputenc}

\input{header.tex}
\input{newcommands.tex}

\title{Splitting-off in Hypergraphs}

\author{
Krist\'{o}f B\'{e}rczi\thanks{MTA-ELTE Matroid Optimization Research Group and HUN-REN-ELTE Egerv\'{a}ry Research Group, Department of Operations Research, E\"{o}tv\"{o}s Lor\'{a}nd
University, Budapest, Hungary, Email: \{kristof.berczi, tamas.kiraly\}@ttk.elte.hu.}
\and Karthekeyan Chandrasekaran\thanks{University of Illinois, Urbana-Champaign, USA, Email: \{karthe, smkulka2\}@illinois.edu. 
Part of this work was done while Karthekeyan and Shubhang were visiting E\"{o}tv\"{o}s Lor\'{a}nd
University.}
\and Tam\'{a}s Kir\'{a}ly\footnotemark[1]
\and Shubhang Kulkarni\footnotemark[2]
}
\date{}
\begin{document}
\maketitle
\input{abstract.tex}

\newpage
\tableofcontents
\newpage

\setcounter{page}{1}
\input{intro-2.tex}

\input{technical-outline-3}
\input{preliminaries-2}
\input{hypergraph-splitting-off-proof-2}
\input{characterization-of-k-connected-hypergraphs}
\input{rooted-steiner-connected-orientations-2}

\input{weak-to-strong-cover}

\input{run-time-analysis-3}
\input{conclusion}

\input{acknowledgement.tex}

\bibliographystyle{abbrv}
\bibliography{references}

\appendix
\input{appendix}

\end{document}

%% file: header.tex
\usepackage[T1]{fontenc}
\usepackage{soul, comment} 
\usepackage{phfqit}
\usepackage{proba} 
\usepackage[sort,nocompress]{cite}
\usepackage{amssymb}
\usepackage{amsmath}
\usepackage{amsthm}
\usepackage{amsfonts}
\usepackage{mathtools}
\usepackage{xspace}
\usepackage{bm}
\usepackage{nicefrac}
\usepackage{commath}
\usepackage{bbm}
\usepackage{boxedminipage}
\usepackage{xparse}
\usepackage{xcolor}
\usepackage{float}
\usepackage{multirow}
\usepackage{makecell}
\usepackage{graphicx}
\usepackage{caption}
\usepackage{subcaption}
\usepackage{ifthen}
\usepackage{thm-restate}
\usepackage{algorithm}
\usepackage[noend]{algpseudocode}
\usepackage{colortbl} 
\usepackage{fancyhdr}   
\usepackage{hyperref}
    \hypersetup{ colorlinks=true, linkcolor=blue, citecolor=magenta, urlcolor=blue,}
\usepackage{fullpage}
\usepackage[utf8]{inputenc}
\usepackage{environ}
\usepackage{mdframed}
\usepackage{cleveref}
\usepackage[short]{optidef} 
\usepackage{enumitem}
\usepackage{tikz}
\usetikzlibrary{shapes}
\usetikzlibrary{positioning}
\usetikzlibrary{fit}
\usetikzlibrary{graphs,graphs.standard}
\usepackage{enumitem}

\newtheorem{lemma}{Lemma}[section]
\newtheorem{remark}{Remark}[section]
\newtheorem{theorem}{Theorem}[section]
\newtheorem{corollary}{Corollary}[section]

\newtheorem{definition}{Definition}[section]

\newtheorem{claim}{Claim}[section]

\usepackage{enumitem}

\NewEnviron{problem}[1]{%
	\begin{center}\fbox{\parbox{6in}{%
				{\centering\scshape #1\par}%
				\parskip=1ex
				\everypar{\hangindent=1em}%
				\BODY
}}\end{center}}



%% file: newcommands.tex
\newcommand{\zeros}{\ensuremath{\mathcal{Z}}\xspace}

\newcommand{\indicator}{\ensuremath{\mathbbm{1}}\xspace}

\newcommand{\functionMaximizationOracle}[1]{\ensuremath{{#1}\text{-}\mathtt{max}\text{-}\mathtt{wc}\text{-}\mathtt{Oracle}}\xspace}
\newcommand{\functionMaximizationOracleStrongCover}[1]{\ensuremath{{#1}\text{-}\mathtt{max}\text{-}\mathtt{sc}\text{-}\mathtt{Oracle}}\xspace}
\newcommand{\functionMaximizationEmptyOracle}[1]
{\ensuremath{{#1}\text{-}\mathtt{max}\text{-}\mathtt{Oracle}}\xspace}

\newcommand{\maximalTightSetFamily}[1]{\ensuremath{\calT_{#1}}\xspace}

\newcommand{\tightSetFamily}[1]{\ensuremath{T_{#1}}\xspace}

\newcommand{\cumulativeMaximalTightSetFamily}[1]{\ensuremath{\calT_{\leq #1}}}

\newcommand{\projCumulativeMaximalTightSetFamily}[2]{\ensuremath{\calT_{\leq #1}|_{V_{#2}}}}

\newcommand{\maximalTightSetFamilyAfterMerge}[1]{\ensuremath{\calT_{#1}^M}\xspace}
\newcommand{\maximalTightSetFamilyAfterReduce}[1]{\ensuremath{\calT_{#1}^R}\xspace}
\newcommand{\tightSetFamilyAfterMerge}[1]{\ensuremath{T_{#1}^M}\xspace}
\newcommand{\tightSetFamilyAfterReduce}[1]{\ensuremath{T_{#1}^R}\xspace}

\newcommand{\cumulativeZeros}[1]{\ensuremath{\calZ_{\leq #1}}\xspace}

\newcommand{\shubhang}[1]{{\color{blue}[{\tiny Shubhang: \bf #1}]\marginpar{\color{blue}*}}}
\newcommand{\knote}[1]{{\color{red}[{\tiny Karthik: \bf #1}]\marginpar{\color{red}*}}}

\newcommand{\functionContract}[2]{\ensuremath{{#1}/_{#2}}\xspace}

\newcommand{\maximalTightSetContainingHyperedge}[1]{\ensuremath{\tau_{#1}}\xspace}
\newcommand{\merge}{\ensuremath{\textsc{Merge}}\xspace}
\newcommand{\reduce}{\ensuremath{\textsc{Reduce}}\xspace}

\newcommand{\head}{\ensuremath{\texttt{head}}\xspace}

%% file: abstract.tex
\begin{abstract}
The splitting-off operation in undirected graphs is a fundamental reduction operation that detaches all edges incident to a given vertex and adds new edges between the neighbors of that vertex while preserving their degrees. Lov\'{a}sz \cite{Lov74, Lovasz-problems-book} and Mader \cite{Mad78} showed the existence of this operation while preserving global and local connectivities respectively in graphs under certain conditions. These results have far-reaching applications in graph algorithms literature \cite{Lov76, Mad78, Frank93:Applications-of-submod-functions, FK02, Kiraly-Lau, Fra92, GB93, Fra94, JFJ95, Frank-book, nagamochi_ibaraki_2008_book, NNI97, HW96, Goe01, Jor03, Kri03, JMS03, CFLY11, BHTP03, Lau07, CS08, NZ20, BN23}. 
In this work, we introduce a splitting-off operation in hypergraphs. We show that there exists a local connectivity preserving complete splitting-off in hypergraphs and give a strongly polynomial-time algorithm to compute it in weighted hypergraphs. We illustrate the usefulness of our splitting-off operation in hypergraphs by showing two applications: (1) we give a constructive characterization of $k$-hyperedge-connected hypergraphs and (2) we give an alternate proof of an approximate min-max relation for max Steiner rooted-connected orientation of graphs and hypergraphs (due to Kir\'{a}ly and Lau \cite{Kiraly-Lau}). Our proof of the approximate min-max relation for graphs circumvents the Nash-Williams' strong orientation theorem and uses tools developed for hypergraphs.
\end{abstract}

%% file: intro-2.tex
\section{Introduction}\label{sec:introduction}
The splitting-off operation in undirected graphs is a simple yet powerful operation in graph theory.  
It is a reduction operation that detaches all edges incident to a given vertex and adds new edges between the neighbors of that vertex while preserving their degrees. 
Equivalently, it enables a vertex to exit the network by informing its neighbors  
how to reconfigure the lost links 
among themselves in order to preserve their degrees.
Lov\'{a}sz \cite{Lov74, Lovasz-problems-book} introduced the splitting-off operation and showed the existence of the operation to preserve global edge-connectivity under certain conditions. Mader \cite{Mad78} showed the existence of the splitting-off operation to preserve local edge-connectivity (i.e., all pairwise edge-connectivities) under certain conditions. 
Both Lov\'{a}sz's and Mader's results also admit strongly polynomial-time algorithms \cite{Frank-book, nagamochi_ibaraki_2008_book, Gab94}. 
Owing to the inductive nature of the splitting-off operation, Lov\'{a}sz's and Mader's results have enabled fundamental results in graph theory as well as efficient algorithms and min-max relations for numerous graph optimization problems. 
In fact, Mader \cite{Mad78} illustrated the power of his local edge-connectivity preserving splitting-off result by deriving Nash-Williams' strong orientation theorem \cite{NW60} (also see Frank's exposition of this derivation \cite{Frank93:Applications-of-submod-functions}).  
Subsequently, the splitting-off operation 
has been used 
to give a constructive characterization of $k$-edge-connected graphs \cite{Frank-book} 
 and 
to address problems in edge-connectivity augmentation \cite{Fra92, Fra94, JFJ95, Frank-book, nagamochi_ibaraki_2008_book}, graph orientation \cite{FK02, Kiraly-Lau}, minimum cuts enumeration \cite{NNI97, HW96, Goe01}, network design \cite{Lov76, GB93, Jor03, CS08, CFLY11}, tree packing \cite{BHTP03, Kri03, JMS03, Lau07}, congruency-constrained cuts \cite{NZ20}, and approximation algorithms for TSP \cite{GB93, BN23}. 
Designing fast algorithms for global/local edge-connectivity preserving splitting-off remains an active area of research (e.g., see recent works \cite{LY13, CLP22-soda, CLP22-stoc, CHLP23}) due to these far-reaching applications. 
In this work, we introduce a splitting-off operation in \emph{hypergraphs}, show the existence of local-connectivity preserving splitting-off operation, design a strongly polynomial-time algorithm to compute it in weighted hypergraphs, and illustrate its usefulness by showing two applications. 

\paragraph{Hypergraphs.}
Hypergraphs offer a richer and more accurate model than graphs for several applications. 
Consequently, hypergraphs have found applications in several modern areas (e.g., see \cite{VeldtBK22,Schlagetal23,LiM17, Ornes21}) and these applications have, in turn, driven exciting recent progress in algorithms for hypergraph optimization problems \cite{KK15, GKP17, CXY19-j, CX18, CC22, CKN20, SY19, BST19, KKTY21-FOCS, KKTY21-STOC, Lee22, JLS22, Quanrud22, GLSTW22, CC23, AGL23, FPZ23}. 
A hypergraph $G=(V, E)$ consists of a finite set $V$ of vertices and a set $E$ of hyperedges, where every hyperedge $e\in E$ is a subset of $V$. We will denote a hypergraph $G=(V, E)$ with hyperedge weights $w: E\rightarrow \Z_+$ by the tuple $(G, w)$ and use $G_w$ to denote the unweighted multi-hypergraph over vertex set $V$ containing $w(e)$ copies of every hyperedge $e\in E$.  
Throughout this work, we will be interested only in hypergraphs with positive integral weights and for algorithmic problems where the input/output is a hypergraph, we will require that the weights are represented in binary. 
If all hyperedges have size at most $2$, then the hyperedges are known as edges and we call such a hypergraph as a graph. 
We emphasize a subtle but important distinction between hypergraphs and graphs: the number of hyperedges in a hypergraph could be exponential in the number of vertices. This is in sharp contrast to graphs where the number of edges is at most the square of the number of vertices. Consequently, in hypergraph network design problems where the goal is to construct a hypergraph with certain properties,  
one needs to be mindful of the number of hyperedges in the hypergraph returned by the algorithm. 
Recent works in hypergraph algorithms literature have 
focused on this 
issue in the context of cut/spectral sparsification of hypergraphs \cite{KK15, CX18, SY19, BST19, CKN20, KKTY21-FOCS, KKTY21-STOC, Lee22, JLS22, Quanrud22}.

\paragraph{Notation.} Let $(G=(V, E), w: E\rightarrow \Z_+)$ be a hypergraph. 
For $X\subseteq V$, let $\delta_G(X)\coloneqq \{e\in E: e\cap X\neq \emptyset, e\setminus X\neq \emptyset\}$.  
We define the cut function $d_{(G, w)}: 2^V\rightarrow \Z_{\ge 0}$ by $d_{(G, w)}(X) \coloneqq \sum_{e\in \delta_G(X)}w(e)$ for every $X\subseteq V$. 
For a vertex $v\in V$, we use $d_{(G,w)}(v)$ to denote $d_{(G,w)}(\{v\})$. 
We define the degree of a vertex $v$ to be the sum of the weights of hyperedges containing $v$---we note that the degree of a vertex is not necessarily equal to $d_{(G, w)}(v)$ since we could have $\{v\}$ itself as a hyperedge (i.e., a singleton hyperedge that contains only the vertex $v$). 
For distinct vertices $u, v\in V$, let $\lambda_{(G, w)}(u, v)\coloneqq \min\{d_{(G, w)}(X): u\in X\subseteq V\setminus \{v\}\}$ -- i.e., $\lambda_{(G, w)}(u, v)$ is the value of a minimum $\{u, v\}$-cut in the hypergraph. If all hyperedge weights are unit, then we drop $w$ from the subscript and simply use $d_G$ and $\lambda_G$. 

\paragraph{Hypergraph Splitting-off.} 
We now introduce our definition of splitting-off in hypergraphs. 
To compare and contrast our definition of splitting-off for hypergraphs with the classical definition of splitting-off for graphs, we include both our definition and the classical definition and distinguish them by identifying them as h-splitting-off and g-splitting-off. 
In the definitions below, we encourage the reader to consider the input hypergraph $(G, w)$ to be a graph while considering g-splitting-off terminology and to be a hypergraph while considering h-splitting-off terminology. We encourage the reader to also assume unit weights during first read. See Figure \ref{figure:hypergraph-splitting-off-example-1} for an example. 

\begin{definition}\label{defn:splitting-off}
    Let $(G=(V, E), w: E\rightarrow \Z_+)$ be a hypergraph and $s\in V$. 
    \begin{enumerate}
            
        \item 
        In \emph{merge almost-disjoint hyperedges}, we pick a pair of hyperedges $e, f\in \delta_G(s)$ such that $e\cap f=\{s\}$, pick a positive integer $\alpha \in \Z_+$ such that $\alpha\le \min\{w(e), w(f)\}$, reduce the weights of hyperedges $e$ and $f$ by $\alpha$, and increase the weight of a hyperedge $g$ by $\alpha$. Here, 
        \begin{enumerate}
            \item if we choose $g\coloneqq e \cup f$, then the associated operation will be called as \emph{h-merge almost-disjoint hyperedges operation}.
            \item if we choose $g\coloneqq (e\cup f)\setminus \{s\}$, then the associated operation will be called as \emph{g-merge almost-disjoint hyperedges operation}.
        \end{enumerate}
        In the above, if $\alpha=w(e)$ (resp. if $\alpha = w(f)$), then we discard the hyperedge $e$ (resp. hyperedge $f$) from the hypergraph obtained after the operation; if the hyperedge $g\not\in E$, then we introduce $g$ as a new hyperedge with weight $w(g) := 0$ before performing the weight increase on the hyperedge $g$.
    
        \item 
        In \emph{trim hyperedge operation}, we pick a hyperedge $e\in \delta_G(s)$, pick a positive integer $\alpha \in \Z_+$,
        reduce the weight of the hyperedge $e$ and 
        increase the weight of the hyperedge $g\coloneqq e\setminus \{s\}$. Here, 
        \begin{enumerate}
            \item if we choose $\alpha \le w(e)$, reduce the weight of the hyperedge $e$ by $\alpha$,  and increase the weight of the hyperedge $g$ by $\alpha$, then the associated operation will be called as \emph{h-trim operation} (if $\alpha=w(e)$, then we discard $e$ from the hypergraph obtained after the operation; if $g\not\in E$, then we add $g$ as a new hyperedge with weight $w(g) := 0$ before performing the weight increase on the hyperedge $g$). 
            
            \item if we choose $\alpha \le w(e)/2$, reduce the weight of the hyperedge $e$ by $2\alpha$, and increase the weight of the hyperedge $g$ by $2\alpha$, then the associated operation will be called as \emph{g-trim operation} (if $\alpha=w(e)/2$, then we discard $e$ from the hypergraph obtained after the operation; if $g\not\in E$, then we add $g$ as a new hyperedge with weight $w(g) := 0$ before performing the weight increase on the hyperedge $g$).  
        \end{enumerate}

\item We say that a hypergraph $(H=(V, E_H), w_H: E_H\rightarrow \Z_+)$ is obtained by applying a 
        \begin{enumerate}
            \item \emph{h-splitting-off operation at $s$ from $(G, w)$} if $(H, w_H)$ is obtained from $(G, w)$ by either the h-merge almost-disjoint hyperedges operation or the h-trim hyperedge operation. 
            \item \emph{g-splitting-off operation at $s$ from $(G, w)$} if $(H, w_H)$ is obtained from $(G, w)$ by either the g-merge almost-disjoint hyperedges operation or the g-trim hyperedge operation. 
        \end{enumerate}

\end{enumerate}
\end{definition}
Certain remarks regarding the definitions are in order. 
Firstly, the trim operation is valuable and unique to hypergraphs. It has been used in the hypergraph literature to obtain small-sized certificates for hypergraph connectivity \cite{CX18} and for certain notions of directed hypergraph connectivity \cite{Frank-book}. 
We note that the trim operation has limited value in graphs---trimming an edge leads to a singleton edge and singleton edges contribute only to the degree but not to the cut value of any set. 
Secondly, all operations mentioned above are degree preserving for vertices $u\in V\setminus \{s\}$: both h-trim and g-trim operations preserve degrees by definition; both h-merge and g-merge \emph{almost-disjoint} hyperedges operations preserve degrees due to the almost-disjoint property of the chosen hyperedges. 
Thirdly, all operations mentioned above 
\emph{do not} increase the 
cut values of subsets $X\subseteq V\setminus\{s\}$. Thus, the relevant goal with these operations is ensuring that the cut values do not decrease too much---i.e., preserving global/local connectivity. 
We will be interested in repeated application of h-splitting-off operations at a vertex from a given hypergraph to isolate that vertex while preserving global/local connectivity. We define these formally next. 
\begin{definition}
    Let $(G=(V, E), w: E\rightarrow \Z_+)$ be a hypergraph and $s\in V$. 
    \begin{enumerate}
        \item We say that a hypergraph $(G^*=(V, E^*), w^*: E^*\rightarrow \Z_+)$ is a 
        \begin{enumerate}
            \item \emph{complete h-splitting-off at $s$ from $(G, w)$} if $d_{(G^*, w^*)}(s) = 0$ and $(G^*, w^*)$ is obtained from $(G, w)$ by repeatedly applying h-splitting-off operations at $s$ from the current hypergraph. 
            \item \emph{complete g-splitting-off at $s$ from $(G, w)$} if $d_{(G^*, w^*)}(s) = 0$ and $(G^*, w^*)$ is obtained from $(G, w)$ by repeatedly applying g-splitting-off operations at $s$ from the current hypergraph. 
        \end{enumerate}

        \item Let $(G^*, w^*)$ be a complete h-splitting-off/g-splitting-off at $s$ from $(G, w)$. We say that $(G^*, w^*)$ 
        \begin{enumerate}
            \item 
            preserves local connectivity if $\lambda_{(G^*, w^*)}(u, v)=\lambda_{(G, w)}(u, v)$ for every distinct $u, v\in V\setminus\{s\}$ and 
            
            \item preserves global connectivity if \\
            $\min\{\lambda_{(G^*, w^*)}(u, v): u, v\in V\setminus \{s\}, u \neq v\}=\min\{\lambda_{(G, w)}(u, v): u, v\in V\setminus \{s\}, u \neq v\}$. 
            
        \end{enumerate}
    \end{enumerate}
\end{definition}

Our first contribution in this work is the definition of h-splitting-off operations at a vertex from a hypergraph. 
To the best of our knowledge, this definition 
has not appeared in the literature before. 
A notion of hypergraph splitting-off motivated by hypergraph connectivity augmentation applications has been studied in the literature before \cite{BJ99, Cosh-thesis, BK12}. 
These works have explored local connectivity preserving complete g-splitting-off at a vertex $s$ from a hypergraph under the assumption that \emph{all hyperedges incident to the vertex $s$ are edges (i.e., have size at most $2$)}. 
In contrast, our focus is on local connectivity preserving complete h-splitting-off at a vertex $s$ from a hypergraph without any assumption on the size of the hyperedges incident to the vertex $s$ (i.e., the vertex $s$ could have arbitrary-sized hyperedges incident to it). We refer the reader to Figure \ref{figure:hypergraph-splitting-off-example-1} for an example of complete h-splitting-off at a vertex from a hypergraph. 

\begin{figure}[thb]
\centering
\includegraphics[width=\textwidth]{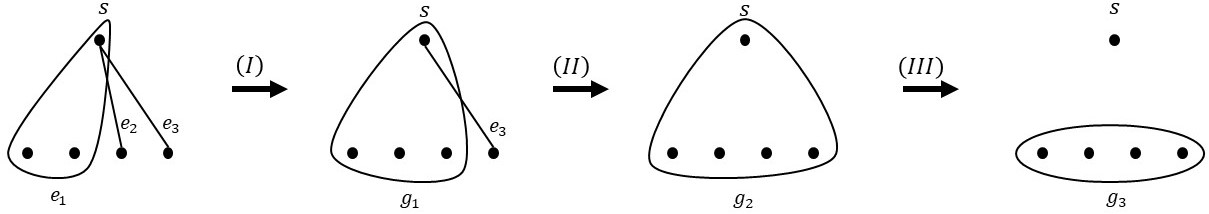}
\caption{Example of (local connectivity preserving) complete h-splitting-off at a vertex $s$ from a hypergraph. Consider the leftmost hypergraph where all hyperedge weights are one and the vertex $s$ is as labeled. Operations (I) and (II) correspond to h-merge almost-disjoint hyperedges operations and Operation (III) corresponds to an h-trim hyperedge operation.}
\label{figure:hypergraph-splitting-off-example-1}
\end{figure}

We will primarily be concerned with complete \emph{h-splitting-off} at a vertex from a \emph{hypergraph} and complete \emph{g-splitting-off} at a vertex from a \emph{graph}.  
Complete g-splitting-off at a vertex from a \emph{graph} is equivalent to the classical and well-studied notion of complete splitting-off at a vertex from a graph (for the definition of the classical notion in graphs, see \cite{Frank-book, nagamochi_ibaraki_2008_book}). 
We cast the results of Lov\'{a}sz \cite{Lov74, Lovasz-problems-book} and Mader \cite{Mad78} in the framework of our definitions now. Let $(G, w)$ be a \emph{graph} and let $s$ be a vertex in $G$. 
Lov\'{a}sz \cite{Lov74, Lovasz-problems-book} showed that if $d_{(G,w)}(\{s\})$ is even and $\min\{\lambda_{(G, w)}(u,v): u, v\in V\setminus \{s\}\}\ge K$ for some $K\ge 2$, then there exists a \emph{global} connectivity preserving complete g-splitting-off at the vertex $s$ from $(G, w)$. 
Mader \cite{Mad78} showed that if $d_{(G,w)}(\{s\})$ is even, there is no cut-edge\footnote{Equivalently, for every edge $e\in \delta_G(s)$ with $w(e)=1$, the removal of that edge does not disconnect the graph.} incident to $s$,  
and $(G, w)$ is connected, then there exists a \emph{local} connectivity preserving complete g-splitting-off at the vertex $s$ from $(G, w)$.

We compare and contrast complete h-splitting-off at a vertex from a hypergraph and complete g-splitting-off at a vertex from a graph.  
Complete h-splitting-off at a vertex $s$ from a hypergraph enables the vertex $s$ to exit the hypergraph by informing its incident hyperedges about how to merge and trim themselves in  order to preserve degrees. In this sense, the definition of complete h-splitting-off at a vertex from a hypergraph serves the same role as complete g-splitting-off at a vertex from a graph. On the other hand, there are important differences between the two notions. Firstly, complete h-splitting-off at a vertex from a \emph{graph} may not necessarily be a graph (owing to the creation of hyperedges of size at least $3$) while it is an easy exercise to show that complete g-splitting-off at a vertex from a \emph{graph} will necessarily be a graph. 
Secondly, local/global connectivity preserving complete g-splitting-off at a vertex from a graph may not exist---see Figure \ref{figure:hypergraph-splitting-off-example-2}. 

\begin{figure}[thb]
\centering
\includegraphics[scale=0.5]{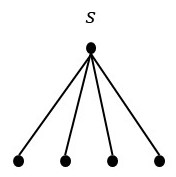}
\caption{An example showing that global connectivity preserving complete g-splitting-off at a vertex from a graph may not exist. All edge weights are one and the vertex $s$ is as labeled.}
\label{figure:hypergraph-splitting-off-example-2}
\end{figure}

As our second main contribution, we show that local connectivity preserving complete h-splitting-off at a vertex from a hypergraph always exists and can be computed in strongly polynomial time (the rightmost hypergraph in Figure \ref{figure:hypergraph-splitting-off-example-1} is a local connectivity preserving complete h-splitting-off at the vertex $s$ from the hypergraph in Figure \ref{figure:hypergraph-splitting-off-example-2}). 

\begin{theorem}\label{thm:complete-splitting-off}
    Given a hypergraph $(G=(V, E), w_G: E\rightarrow \Z_+)$ and a vertex $s\in V$, there exists a strongly polynomial-time algorithm to find a local connectivity preserving complete h-splitting-off at $s$ from $(G, w_G)$. 
\end{theorem}

A distinction between Theorem \ref{thm:complete-splitting-off} and the graph splitting-off results of Lov\'{a}sz and Mader is that Theorem \ref{thm:complete-splitting-off} shows the existence of a local connectivity preserving complete h-splitting-off at a vertex from a hypergraph \emph{without any assumptions on the hypergraph} whereas Lov\'{a}sz's and Mader's results hold only under certain technical assumptions on the graph. In several applications of their results, additional arguments are needed to address cases where those technical assumptions do not hold. For this reason, we believe that Theorem \ref{thm:complete-splitting-off} could be useful in simplifying the arguments involved in some of the applications of Lov\'{a}sz's and Mader's graph splitting-off results (e.g., we will later see that the edge version of Menger's theorem in undirected \emph{graphs} follows in a straightforward fashion from Theorem \ref{thm:complete-splitting-off}). 

A crude run-time of our algorithm that proves Theorem \ref{thm:complete-splitting-off} is $O(|V|^6(|V|+|E|)^3)$.  
We understand that this run-time is impractical for applications. Nevertheless, we mention it here explicitly for the sake of completeness and as a potential starting point for future work: it would be interesting to design a near-linear time algorithm to find a local connectivity preserving complete h-splitting-off at a vertex from a weighted hypergraph.

\begin{remark}\label{remark:exp-sized-conn-preserving-splitting-off}
We note that existence of a local/global connectivity preserving complete h-splitting-off at a vertex from a hypergraph does not necessarily imply a polynomial-time algorithm to find it. 
This is because, 
a local/global connectivity preserving complete h-splitting-off at a vertex from a hypergraph $(G, w_G)$ could contain exponential number of hyperedges although $G$ contains only polynomial number of hyperedges. 
We give an example to illustrate this issue. 
Consider the graph $(G=(V, E), w_G)$ where $G$ is the star graph on $n+1$ vertices with $s$ being the center of the star and all edge weights are $2^{n-1}-1$. 
Consider the hypergraph $(H=(V, E_H), w_H)$ such that 
$E_H:=\{e: e\subseteq V\setminus \{s\} \text{ and } |e|\ge 2\}$ 
with all hyperedge weights being one. 
The hypergraph $(H, w_H)$ is a local connectivity preserving complete h-splitting-off at $s$ from $(G, w_G)$, but $(H, w_H)$ has exponential number of hyperedges although $(G, w_G)$ has only $n$ edges. 
In order to design a polynomial-time algorithm to find a local connectivity preserving complete h-splitting-off at a vertex from a hypergraph, a necessary step is to show the existence of a local connectivity preserving complete h-splitting-off at a vertex from a hypergraph that contains only polynomially many \emph{additional} hyperedges. 
For the star graph $(G,w_G)$ with edge weights $2^{n-1}-1$ mentioned above, the hypergraph $(H'=(V, E_{H'}), w_{H'})$ containing only one hyperedge, namely $E_{H'}:=\{V\setminus \{s\}\}$ with the weight of that hyperedge being $2^{n-1}-1$ is also a local connectivity preserving complete h-splitting-off at $s$ from $(G,w_G)$. 
One of the features of our algorithmic proof of Theorem \ref{thm:complete-splitting-off} is the existence of a local connectivity preserving complete h-splitting-off at a vertex from a hypergraph that contains only polynomially many additional hyperedges. 
\end{remark}

As our third main contribution, 
we present two applications of Theorem \ref{thm:complete-splitting-off}. 

\paragraph{Application 1: Constructive characterization of $k$-hyperedge-connected hypergraphs.}
For the purposes of this application, 
graphs and hypergraphs will refer to multi-graphs and multi-hypergraphs, respectively.
Let $k$ be a positive integer. A graph $G=(V, E)$ is \emph{$k$-edge-connected} if $d_{G}(X)\ge k$ for every non-empty proper subset $X\subsetneq V$. 
Constructive characterization of $k$-edge-connected graphs is a central problem in graph theory. 
It is well-known that a graph is $1$-edge-connected if and only if it admits a spanning tree. 
Robbins' \cite{Rob39} showed that a graph is $2$-edge-connected if and only if it admits an ear decomposition (see \cite{Frank-book} for definition of ear decomposition). 
Generalizing Robbins' result, 
Lov\'{a}sz \cite{Lov74, Lovasz-problems-book} gave a constructive characterization of $k$-edge-connected graphs for \emph{even} $k$ using his result on global connectivity preserving complete g-splitting-off at a vertex from a graph. 
Mader \cite{Mad78} gave a constructive characterization of $k$-edge-connected graphs for \emph{odd} $k$ using his result on local connectivity preserving complete g-splitting-off at a vertex from a graph. 
Motivated by these results, we present a constructive characterization of $k$-hyperedge-connected hypergraphs using our splitting-off result in Theorem \ref{thm:complete-splitting-off}.
A hypergraph $G=(V, E)$ is defined to be \emph{$k$-hyperedge-connected} if $d_{G}(X)\ge k$ for every non-empty proper subset $X\subsetneq V$.

Both Lov\'{a}sz's and Mader's constructive characterizations of $k$-edge-connected graphs are based on a pinching operation in graphs. Our constructive characterization of $k$-hyperedge-connected hypergraphs is also based on a pinching operation, but our pinching operation is defined for hypergraphs. We define this operation now (see \Cref{fig:hypergraph-pinching-example} for an example). 

\begin{definition}\label{def:pinching-operation}
    Let $G = (V, E)$ be a hypergraph and $p,k \in \Z_+$ be positive integers such that $p \leq k$. In \emph{$(k, p)$-pinching} hyperedges of $G$, we obtain a new hypergraph by performing the following sequence of operations:
    \begin{enumerate}
        \item pick $p$ distinct hyperedges $e_1, \ldots, e_p \in E$,
        \item pick $p$ positive integers $t_1, \ldots, t_p \in \Z_+$ such that $\sum_{i=1}^p t_i = k$,
        \item for each $i \in [p]$, choose a partition of the hyperedge $e_i$ into $t_i$ non-empty parts $e_i = \uplus_{j \in [t_i]} f_i^j$,
        \item remove the hyperedges $e_1, \ldots, e_p$ from the hypergraph $G$,
        \item add a new vertex $s$ and hyperedges $\{f_i^j \cup \{s\} : j \in [t_i], i \in [p]\}$ to the hypergraph $G$.
    \end{enumerate}
\end{definition}

\begin{figure}[thb]
    \centering
    \includegraphics[width=0.7\textwidth]{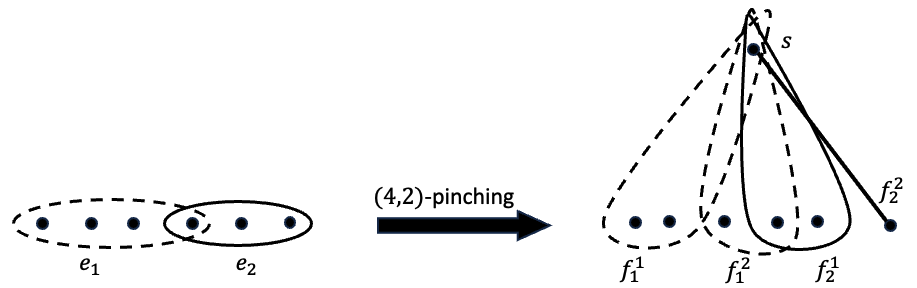}
    \caption{An example of a $(4,2)$-pinching operation. Here, $t_1 = t_2 = 2$.}
    \label{fig:hypergraph-pinching-example}
\end{figure}
With this definition of pinching, we show the following constructive characterization of $k$-hyperedge-connected hypergraphs. 
\begin{restatable}{theorem}{thmConstructiveCharacterization}\label{thm:k-EC-hypergraph-characterization}
    Let $k \in \Z_+$ be a positive integer. A hypergraph $G = (V, E)$ is $k$-hyperedge-connected if and only if $G$ can be obtained by starting from the single vertex hypergraph with no hyperedges and repeatedly applying one of the following two operations:
    \begin{enumerate}
        \item add a new hyperedge over a subset of vertices of the existing hypergraph, and
        \item $(k, p)$-pinching hyperedges of the existing hypergraph for some positive integer $p \leq k$.
    \end{enumerate}
\end{restatable}
Our proof of Theorem \ref{thm:k-EC-hypergraph-characterization} is constructive: i.e., given a $k$-hyperedge-connected hypergraph $G$, our proof gives a polynomial-time algorithm to construct a sequence of hypergraphs $G_0, G_1, G_2, \ldots, G_t$, where $G_0$ is the single vertex hypergraph with no hyperedges, $G_t=G$ and for each $i\in [t]$, the hypergraph $G_i$ is obtained from $G_{i-1}$ by either adding a new hyperedge over a subset of vertices in $G_{i-1}$ or by $(k,p)$-pinching hyperedges in $G_{i-1}$ for some positive integer $p\le k$. 

\begin{remark}
Robbins' constructive characterization of $2$-edge-connected graphs and Lov\'{a}sz's constructive characterization of $2k$-edge-connected graphs find applications in graph orientation problems---e.g., 
Robbins' result leads to an algorithm to find a strongly connected orientation of $2$-edge-connected graphs and 
Lov\'{a}sz's result leads to an algorithm to find a strongly $k$-arc-connected orientation of $2k$-edge-connected graphs. In fact, the latter leads to an alternative proof of Nash-Williams' \emph{weak orientation theorem} \cite{NW60}. 
Along the same vein, we hope that our above characterization of $k$-edge-connected hypergraphs might find applications in hypergraph orientation problems. 
\end{remark}

\paragraph{Application 2.1: Steiner Rooted $k$-arc-connected Orientation of Graphs.}
Orienting a graph to achieve high connectivity is a fundamental area in graph theory, combinatorial optimization, and algorithms. Let $G=(V, E)$ be an undirected graph.  An \emph{orientation} $\overrightarrow{G}$ of $G$ is a directed graph obtained by assigning a direction to each edge of $G$. 
Let $G=(V, E)$ be an undirected graph, $T\subseteq V$ be a set of terminals, $r\in T$ be a root vertex, and $k$ be a positive integer. 
An orientation $\overrightarrow{G}$ of $G$ is defined to be \emph{Steiner rooted $k$-arc-connected} if there exist $k$ arc-disjoint paths in $\overrightarrow{G}$ 
from $t$ to $r$ 
for every terminal $t\in T\setminus \{r\}$. 
In \textsc{Max Steiner Rooted-Connected Orientation} problem, the goal is to find the maximum integer $k$ and an orientation $\overrightarrow{G}$ of $G$ such that $\overrightarrow{G}$ is Steiner rooted $k$-arc-connected. 
\textsc{Max Steiner Rooted-Connected Orientation} generalizes two classic problems in graph theory: The case of $|T|=2$ is the max edge-disjoint $\{r,v\}$-paths problem and is solved via Menger's theorem. The case of $T=V$ is the max edge-disjoint spanning trees problem and is solved via Tutte and Nash-Williams' theorem \cite{Tu61, NW61}. We mention that both these problems are also generalized by the Steiner Tree Packing problem and the associated Kriesell's conjecture \cite{Kri03, JMS03, Lau07}, but we will not focus on that generalization. 

Kir\'{a}ly and Lau \cite{Kiraly-Lau} introduced the \textsc{Max Steiner Rooted-Connected Orientation}, showed that it is NP-hard, and gave a $2$-approximation via an approximate min-max relation. We state their approximate min-max relation now. 
An undirected graph $G$ is \emph{Steiner $k$-edge-connected} if $\lambda_G(u,v)\ge k$ 
for every pair of distinct terminals $u, v\in T$. 
It is clear that if the graph $G$ has a Steiner rooted $k$-arc-connected orientation, then $G$ should be Steiner $k$-edge-connected. However, the converse is not necessarily true. Kir\'{a}ly and Lau observed that if the graph is Steiner $2k$-edge-connected, then it has a Steiner rooted $k$-arc-connected orientation. 

\begin{theorem}[Kir\'{a}ly and Lau \cite{Kiraly-Lau}]\label{thm:rooted-steiner-orientation-in-graphs}
Let $G=(V, E)$ be an undirected graph, $T\subseteq V$ be a subset of terminals, $r\in T$ be the root vertex, and $k$ be a positive integer. If $G$ is Steiner $2k$-edge-connected, then it has a Steiner rooted $k$-arc-connected orientation. 
\end{theorem}

Kir\'{a}ly and Lau observed that Theorem \ref{thm:rooted-steiner-orientation-in-graphs} follows immediately from Nash-Williams' \emph{strong orientation} theorem. 
Nash-Williams' \emph{strong orientation} theorem \cite{NW60} states that every undirected graph $G=(V, E)$ admits an orientation $\overrightarrow{G}$ such that $\lambda_{\overrightarrow{G}}(u,v)\ge \lfloor \lambda_G(u,v)/2\rfloor$ for every distinct $u, v\in V$, where $\lambda_{\overrightarrow{G}}(u,v)$ is the maximum number of arc-disjoint directed paths from $u$ to $v$ in $\overrightarrow{G}$. 
In this work, we give an alternative proof of Theorem \ref{thm:rooted-steiner-orientation-in-graphs} that does not rely on Nash-Williams' strong orientation theorem. Instead, we use Theorem \ref{thm:complete-splitting-off}. 
Our proof strategy is unique since it proves an orientation result for \emph{graphs} using tools developed for \emph{hypergraphs}. 

\begin{remark}
Nash-Williams' proof of the strong orientation theorem \cite{NW60} is a sophisticated inductive argument. Giving a simple and more insightful proof of the strong orientation theorem has been a central topic of interest in graph theory and combinatorial optimization (see \cite{Frank93:Applications-of-submod-functions}). Mader \cite{Mad78} gave a different proof of the strong orientation theorem using his local connectivity preserving splitting-off theorem, but his proof also involved sophisticated technical arguments. Frank \cite{Frank93:Applications-of-submod-functions} condensed the ideas of both Nash-Williams and Mader to present a proof of the strong orientation theorem using Mader's local connectivity preserving splitting-off, but it is still technically complicated. The technical complication in using Mader's local connectivity preserving splitting-off result arises from the assumptions that need to be satisfied by the vertex to be split-off. In contrast, our splitting-off result for hypergraphs (namely, Theorem \ref{thm:complete-splitting-off}) does not need any assumptions on the vertex to be split-off. In light of these considerations, our proof of Theorem \ref{thm:rooted-steiner-orientation-in-graphs} using Theorem \ref{thm:complete-splitting-off} provides hope that Theorem \ref{thm:complete-splitting-off} (or the ideas therein) could be used to give a conceptually simpler proof of Nash-Williams' strong orientation theorem. 
\end{remark}

\paragraph{Application 2.2: Steiner Rooted $k$-hyperarc-connected Orientation of Hypergraphs.}
Orienting \emph{hypergraphs} is also a fundamental area in graph theory and combinatorial optimization (see Frank's book \cite{Frank-book}) 
with far reaching implications. For example, Woodall's conjecture can be reformulated as a hypergraph orientation problem \cite{egres-open-problem-on-head-disjoint-sco}; moreover,  
hypergraph orientation results have recently been used in coding theory \cite{AGL23}.
Kir\'{a}ly and Lau \cite{Kiraly-Lau} showed that the approximate min-max relation in Theorem \ref{thm:rooted-steiner-orientation-in-graphs} also holds for hypergraphs. To state their result, we need some terminology in hypergraph orientations. 

Let $G=(V, E)$ be a hypergraph. An \emph{orientation} $\overrightarrow{G}=(V, E, \head:E\rightarrow V)$ of $G$ is a directed hypergraph obtained by assigning a unique head vertex $\head(e)\in e$ for each $e\in E$. 
A pair $(e, \head(e))$ is denoted as a \emph{hyperarc} with the head of the hyperarc being $\head(e)$ and the tails of the hyperarc being $e\setminus \head(e)$. 
Let $G=(V, E)$ be a hypergraph, $T\subseteq V$ be a set of terminals, $r\in T$ be a root vertex, and $k$ be a positive integer. An orientation $\overrightarrow{G}$ of $G$ is defined to be \emph{Steiner rooted $k$-hyperarc-connected} if there exist $k$ hyperarc-disjoint paths in $\overrightarrow{G}$ 
from $t$ to $r$ 
for every terminal $t\in T\setminus \{r\}$. Here, a path from $t$ to $r$ in a directed hypergraph is an alternating sequence of distinct vertices and hyperarcs $t=v_1, a_1, v_2, a_2, ..., a_{\ell-1}, v_{\ell}=r$ such that 
$v_i$ is a tail of $a_i$ and $v_{i+1}$ is the head of $a_i$ 
for every $i\in [\ell-1]$. We say that a hypergraph $G$ is \emph{Steiner $k$-hyperedge-connected} if 
$\lambda_G(u,v)\ge k$ for 
every pair of distinct terminals $u, v\in T$. 
It is clear that if the hypergraph $G$ has a Steiner rooted $k$-hyperarc-connected orientation, then $G$ should be Steiner $k$-hyperedge-connected. However, the converse is not necessarily true. Kir\'{a}ly and Lau \cite{Kiraly-Lau} showed that if the hypergraph is Steiner $2k$-hyperedge-connected, then it has a Steiner rooted $k$-hyperarc-connected orientation. 

\begin{restatable}[Kir\'{a}ly and Lau \cite{Kiraly-Lau}]{theorem}{thmRootedSteinerArcConnectedOrientationsHypergraphs}\label{thm:rooted-steiner-orientation-in-hypergraphs}
Let $G=(V, E)$ be a hypergraph, $T\subseteq V$ be a subset of terminals, $r\in T$ be the root vertex, and $k$ be a positive integer. If $G$ is Steiner $2k$-hyperedge-connected, then it has a Steiner rooted $k$-hyperarc-connected orientation. 
\end{restatable}
Kir\'{a}ly and Lau's proof of Theorem \ref{thm:rooted-steiner-orientation-in-hypergraphs} was based on careful uncrossing and contractions. In this work, we give an alternative proof of Theorem \ref{thm:rooted-steiner-orientation-in-hypergraphs} using Theorem \ref{thm:complete-splitting-off}. 
Our proof of \Cref{thm:rooted-steiner-orientation-in-hypergraphs} reveals the source of the $2$-factor gap in the approximate min-max relation of Kir\'{a}ly and Lau for \textsc{Max Steiner Rooted-Connected Orientation Problem}: it arises from the $2$-factor gap between connectivity and weak-partition-connectivity of hypergraphs (see Definition \ref{defn:weak-partition-connectivity} for the definition of weak-partition-connectivity and Lemma \ref{lem:2k-conn-implies-k-weak-pc}). 

\begin{remark}
Our proof technique for Theorems \ref{thm:rooted-steiner-orientation-in-graphs} and \ref{thm:rooted-steiner-orientation-in-hypergraphs} using Theorem \ref{thm:complete-splitting-off}---i.e., via the local-connectivity preserving splitting-off operation in hypergraphs---also leads to an alternate proof of Menger's theorem in \emph{undirected} graphs and hypergraphs (edge-disjoint version). For details, we refer the reader to Section \ref{sec:rooted-steiner-connected-orientations} where we discuss a proof of Menger's theorem using \Cref{thm:complete-splitting-off} as a warm-up towards a proof of Theorems \ref{thm:rooted-steiner-orientation-in-graphs} and \ref{thm:rooted-steiner-orientation-in-hypergraphs}. 
\end{remark}

Both Theorems \ref{thm:rooted-steiner-orientation-in-graphs} and \ref{thm:rooted-steiner-orientation-in-hypergraphs} can be extended to weighted graphs/hypergraphs (by considering parallel copies of edges/hyperedges). The weighted version of Theorems \ref{thm:rooted-steiner-orientation-in-graphs} and \ref{thm:rooted-steiner-orientation-in-hypergraphs} can also be shown to admit strongly polynomial-time algorithms using our proof strategy as well as the proof strategy of Kir\'{a}ly and Lau. We avoid stating the weighted versions in the interests of brevity.


%% file: technical-outline-3.tex

\subsection{Proof Technique for Theorem \ref{thm:k-EC-hypergraph-characterization}}
We outline the proof technique for Theorem \ref{thm:k-EC-hypergraph-characterization}. The reverse direction follows by observing that if a hypergraph is $k$-hyperedge-connected, then both operations in the statement of the theorem preserve $k$-hyperedge-connectivity. We sketch a proof of the forward direction. 
The proof is by induction on the number of hyperedges plus the number of vertices. 
First, suppose that there exists a hyperedge $e \in E$ such that $G - e$ is still $k$‐hyperedge‐connected. We note that deleting the hyperedge $e$ is the inverse of operation (1). Consequently, the proof follows by deleting the hyperedge $e$, using the induction hypothesis on the resulting hypergraph $H$, and then noting that the hypergraph $G$ is obtained from $H$ by operation (1). 
Next, suppose that there does not exist a hyperedge $e \in E$ such that $G-e$ is $k$-hyperedge-connected. We call such a hypergraph to be minimally $k$-hyperedge-connected. In \Cref{lem:minimally-k-EC-hypergraph-has-degree-k-vertex}, we show that a minimally $k$-hyperedge-connected hypergraph contains a vertex $u$ with degree exactly $k$. By \Cref{thm:complete-splitting-off}, there exists a global-connectivity preserving complete h-splitting-off at the vertex $u$ from the hypergraph $G$. Let $H$ be a global-connectivity preserving complete h-splitting-off at the vertex $u$ from the hypergraph $G$. 
We note that complete h-splitting-off at $u$ followed by deletion of the vertex $u$ is the inverse of operation (2) at $u$. 
Consequently, the proof follows by using the induction hypothesis on the hypergraph $H-u$ and then noting that the hypergraph $G$ is obtained from $H-u$ by operation (2).


\subsection{Proof Technique for Theorems \ref{thm:rooted-steiner-orientation-in-graphs} and \ref{thm:rooted-steiner-orientation-in-hypergraphs}}
We outline the proof technique for Theorem \ref{thm:rooted-steiner-orientation-in-hypergraphs} and will remark after the proof about how it also implies a proof for Theorem \ref{thm:rooted-steiner-orientation-in-graphs}. 
Our proof of Theorem \ref{thm:rooted-steiner-orientation-in-hypergraphs} will be in three steps. 
Let us denote the set of non-terminals as \emph{Steiner vertices}. 
Our first step is to obtain a hypergraph $H = (T, E_H)$ by applying our local connectivity preserving complete h-splitting-off at each Steiner vertex of $G$ (sequentially, in arbitrary order of the Steiner vertices) and deleting the isolated  vertices. 
We note that deleting the isolated vertices ensures that the vertex set of $H$ is the set of terminals $T$. 
Moreover, our local connectivity preserving complete h-splitting-off ensures that the hypergraph $H$ is $2k$-hyperedge-connected since the hypergraph $G$ is Steiner $2k$-hyperedge-connected. 
Our second step is to show that this hypergraph $H=(T, E_H)$ admits a rooted $k$-hyperarc-connected orientation. 
A known characterization for the existence of a rooted $k$-hyperarc-connected orientation of a hypergraph is that the hypergraph is \emph{$k$-weak-partition-connected} (see Definition \ref{defn:weak-partition-connectivity} for the definition of weak-partition-connectivity and Theorem \ref{thm:frank:weak-pc-iff-rooted-hyparc-conn-orientation} for the characterization). We mention that the notion of weak-partition-connectivity in hypergraphs has been used recently in the  context of coding theory \cite{GLSTW22, AGL23}. 
In order to use the characterization for the existence of a rooted $k$-hyperarc-connected orientation of a hypergraph, we relate the connectivity of a hypergraph to its \emph{weak-partition-connectivity} and conclude that if $H$ is $2k$-hyperedge-connected, then it is $k$-weak-partition-connected (see Lemma \ref{lem:2k-conn-implies-k-weak-pc}). 
Consequently, the hypergraph $H$ admits a rooted $k$-hyperarc-connected orientation. 
We note that such an orientation of $H$ is equivalent to a Steiner rooted $k$-hyperarc-connected orientation of $H$ since the vertex set of $H$ is the set $T$ of terminals. 
Our third step is to use this Steiner rooted $k$-hyperarc-connected orientation of $H$ to obtain a Steiner rooted $k$-hyperarc-connected orientation of the hypergraph $G$: we will see that there is a natural way to extend the orientation of hyperedges while reversing the h-splitting-off operations to preserve Steiner rooted $k$-hyperarc-connected property (see Lemma \ref{lem:lifting-Steiner-orientation-through-splitting-off}). This would complete the proof of \Cref{thm:rooted-steiner-orientation-in-hypergraphs}. 

We note that if the hypergraph $G$ is a graph, then the same proof above obtains the required graph orientation, thus proving \Cref{thm:rooted-steiner-orientation-in-graphs}.
In particular, to prove Theorem \ref{thm:rooted-steiner-orientation-in-graphs}, we start from a graph $G=(V, E)$ that is Steiner $2k$-edge-connected, but our local connectivity preserving complete h-splitting-off operations at Steiner vertices results in a hypergraph $H=(T, E_H)$ that is $2k$-hyperedge-connected; by Lemma \ref{lem:2k-conn-implies-k-weak-pc}, the hypergraph $H$ is $k$-weak-partition-connected; now Theorem \ref{thm:frank:weak-pc-iff-rooted-hyparc-conn-orientation} gives a rooted $k$-hyperarc-connected orientation of the resulting hypergraph. Such an orientation is extended to a  Steiner rooted $k$-hyperarc-connected orientation of the \emph{graph} $G=(V, E)$ using Lemma \ref{lem:lifting-Steiner-orientation-through-splitting-off}. Essentially, the proof starts from the given graph, obtains a related hypergraph, orients that hypergraph, and extends that orientation of the hypergraph back into a desired orientation of the given graph. 

Our proof technique for Theorems \ref{thm:rooted-steiner-orientation-in-graphs} and \ref{thm:rooted-steiner-orientation-in-hypergraphs} also leads to an alternate proof of Menger's theorem in \emph{undirected} graphs and hypergraphs (edge-disjoint version)---see Section \ref{sec:rooted-steiner-connected-orientations}. 

\subsection{Proof Technique for Theorem \ref{thm:complete-splitting-off}}
We prove a more general statement that implies Theorem \ref{thm:complete-splitting-off}. We begin with the definitions needed for the more general statement. 

\begin{definition}
    Let $V$ be a finite set, $p: 2^V\rightarrow \Z$ be a set function,  and $(H=(V, E), w:E\rightarrow \Z_+)$ be a hypergraph. 
\begin{enumerate}
    \item The set function $p$
    \begin{enumerate}
    \item is \emph{symmetric} if $p(X)=p(V-X)$ for every $X\subseteq V$, and 
    \item is \emph{skew-supermodular} if for every $X, Y\subseteq V$, at least one of the following inequalities hold: 
    \begin{enumerate}
    \item $p(X) + p(Y) \leq p(X\cap Y) + p(X\cup Y)$. 
    \item $p(X) + p(Y) \leq p(X - Y) + p(Y - X)$. 
\end{enumerate}
\end{enumerate}
    \item The coverage function $b_{(H, w)}:2^V\rightarrow \Z_{\ge 0}$ is defined by $b_{(H,w)}(X):=\sum_{e\in B_H(X)}w(e)$ for every $X\subseteq V$, where $B_H(X):=\{e\in E: e\cap X\neq \emptyset\}$ for every $X\subseteq V$. 
    \item The hypergraph $(H, w)$ \emph{weakly covers} the function $p$ if $b_{(H, w)}(X)\ge p(X)$ for every $X\subseteq V$. 
    \item The hypergraph $(H, w)$ \emph{strongly covers} the function $p$ if $d_{(H, w)}(X)\ge p(X)$ for every $X\subseteq V$. 
\end{enumerate}
\end{definition}

If a hypergraph $(H, w)$ strongly covers a function $p:2^V\rightarrow \Z$, then it also weakly covers the function $p$. However, the converse is false -- i.e., a weak cover is not necessarily a strong cover\footnote{For example, consider the function $p: 2^V\rightarrow \Z$ defined by $p(X)\coloneqq 1$ for every non-empty proper subset $X\subsetneq V$ and $p(\emptyset)\coloneqq p(V)\coloneqq 0$, and the hypergraph $(H=(V, E\coloneqq \{\{u\}: u\in V\}), w: E\rightarrow \{1\})$.}.  
Bern\'{a}th and Kir\'{a}ly \cite{Bernath-Kiraly} showed 
that a weak cover of a \emph{symmetric skew-supermodular} function can be converted to a strong cover of the same function by repeated \emph{merging of disjoint hyperedges}. 
We recall their definition of the merging operation, discuss their result, and its significance now. 
\begin{restatable}{definition}{defHyperedgeMerge}\label{def:WeakToStrongCover:hyperedge-merge}
Let $(H=(V, E), w: E\rightarrow \Z_+)$ be a hypergraph. 
We use $H_w$ to denote the unweighted multi-hypergraph over vertex set $V$ containing $w(e)$ copies of every hyperedge $e\in E$.  
By \emph{merging two disjoint hyperedges} of $H_w$, we refer to the operation of replacing them by their union in $H_w$. 
We will say that a hypergraph $(G=(V, E_G), c: E_G\rightarrow \Z_+)$ is obtained from $(H, w)$ by \emph{merging hyperedges} if the multi-hypergraph $G_c$ is obtained from the multi-hypergraph $H_w$ by repeatedly merging two disjoint hyperedges in the current hypergraph. 
\end{restatable}
Bern\'{a}th and Kir\'{a}ly showed the following result: 
\begin{theorem}[Bern\'{a}th and Kir\'{a}ly \cite{Bernath-Kiraly}]\label{thm:bernath-kiraly}
    Let $(H=(V, E), w: E\rightarrow \Z_+)$ be a hypergraph and $p:2^V\rightarrow \Z$ be a symmetric skew-supermodular function such that 
$b_{(H, w)}(X) \geq p(X)$ for every $X \subseteq V$. 
Then, there exists a hypergraph $\left(H^* = (V, E^*), w^*:E^* \rightarrow\Z_+\right)$ 
such that
\begin{enumerate}[label=$(\arabic*)$]
    \item $d_{(H^*, w^*)}(X) \geq p(X)$ for every $X\subseteq V$ and 
    \item the hypergraph $(H^*,w^*)$ is obtained by merging hyperedges of the hypergraph $(H, w)$.
\end{enumerate}
\end{theorem}

We observe that Theorem \ref{thm:bernath-kiraly} can be used to prove the existential version of Theorem \ref{thm:complete-splitting-off}: namely, for every hypergraph $(G=(V, E), w_G: E\rightarrow \Z_+)$ and a vertex $s\in V$, \emph{there exists} a local connectivity preserving complete h-splitting-off at $s$ from $(G,w_G)$. This can be shown by setting up the hypergraph $(H,w)$ and the function $p$ suitably based on $(G,w_G)$ and using Theorem \ref{thm:bernath-kiraly} (see the first two paragraphs of the proof of Theorem \ref{thm:hypergraphs-splitting-off-stronger-form} in Section \ref{sec:hypergraph-splitting-off-proof}). 
We emphasize that this conclusion regarding hypergraph splitting-off from Bern\'{a}th and Kir\'{a}ly's result was not known before in the literature and is one of our contributions. 
\begin{remark}
    We were also able to prove the existential version of Theorem \ref{thm:complete-splitting-off} using \emph{element-connectivity preserving reduction operations} \cite{BCCKM-private} (see \cite{CK14} for the definition of element-connectivity and the notion of element-connectivity preserving reduction operations) -- we omit the details of this alternate proof in the interests of brevity. The alternate proof does not seem to be helpful for the purposes of a strongly polynomial time algorithm. In fact, it remains open to design a strongly polynomial-time algorithm to perform \emph{complete} element-connectivity preserving reduction operations in the weighted setting \cite{CK14}. 
\end{remark}

We recall that existence of a local connectivity preserving complete h-splitting-off at a vertex from a hypergraph does not immediately imply a polynomial-time algorithm---see the example in Remark \ref{remark:exp-sized-conn-preserving-splitting-off}. 
However, the above-mentioned proof of 
existence of a local-connectivity preserving complete splitting-off at an arbitrary vertex from a hypergraph (i.e., existential version of Theorem \ref{thm:complete-splitting-off}) 
via Theorem \ref{thm:bernath-kiraly} suggests a natural approach towards designing a strongly polynomial time algorithm to find a local-connectivity preserving complete splitting-off at a given vertex from a given hypergraph: it suffices to prove a constructive version of Theorem \ref{thm:bernath-kiraly} via a strongly polynomial-time algorithm. Towards this end, the example in Remark \ref{remark:exp-sized-conn-preserving-splitting-off} suggests a necessary structural step towards a strongly polynomial-time algorithmic version of Theorem \ref{thm:bernath-kiraly}: we need to show Theorem \ref{thm:bernath-kiraly} with the extra conclusion that the number of additional hyperedges in $H^*$ is polynomial in the number of hyperedges and vertices in $H$.

Bern\'{a}th and Kir\'{a}ly proved Theorem \ref{thm:bernath-kiraly} in the context of a reduction between certain hypergraph connectivity augmentation problems. For that reduction, the existential version of Theorem \ref{thm:bernath-kiraly} is sufficient. However, for the purposes of our application to hypergraph splitting-off, we need an algorithmic version of Theorem \ref{thm:bernath-kiraly}. 
Bern\'{a}th and Kir\'{a}ly's 
proof of Theorem \ref{thm:bernath-kiraly} is in fact algorithmic, but the run-time of the associated algorithm is not necessarily polynomial.  
Their proof implies that the number of additional hyperedges in the hypergraph $(H^*, c^*)$ returned by their algorithm is at most  $\sum_{e\in E}w(e)$ (i.e., $|E^*|-|E|\le \sum_{e\in E}w(e)$) and the run-time of the algorithm is $O(\sum_{e\in E}(|e|+w(e))$. In particular, their run-time is polynomial only if the input weights are given in unary. Moreover, the exponential-sized hypergraph given in the example in Remark \ref{remark:exp-sized-conn-preserving-splitting-off} could indeed arise as a consequence of their algorithm. 

We address both the structural and the algorithmic issues mentioned above by proving a stronger algorithmic 
version of Theorem \ref{thm:bernath-kiraly}. 
In order to phrase an algorithmic version of Theorem \ref{thm:bernath-kiraly}, we need suitable access to the function $p$. Bern\'{a}th and Kir\'{a}ly \cite{Bernath-Kiraly} suggested access to a certain function maximization oracle associated with the function $p$ that we describe below. 

\begin{definition}
    Let $p: 2^V\rightarrow \Z$ be a set function.    $\functionMaximizationOracleStrongCover{p}\left(\left(G_0, c_0\right), S_0, T_0\right)$ takes as input a hypergraph $(G_0=(V, E_0), c_0: E_0\rightarrow \mathbb{Z}_{+})$ and disjoint sets $S_0, T_0\subseteq V$, and returns a tuple $(Z, p(Z))$, where $Z$ is an optimum solution to the following problem:
        \begin{align*}
        \max & \left\{p(Z)-d_{(G_0, c_0)}(Z): S_0\subseteq Z\subseteq V-T_0\right\}. \tag{\functionMaximizationOracleStrongCover{p}}\label{tag:maximization-oracle-strong-cover}
        \end{align*}
\end{definition}
For the purposes of our application (namely local connectivity preserving complete h-splitting-off at a vertex from a hypergraph), the above-mentioned function maximization oracle can be implemented to run in strongly polynomial time (see Lemma \ref{lemma:helper-for-p-max-oracle}). Using the above mentioned oracle, we prove the following algorithmic version of Theorem \ref{thm:bernath-kiraly}.

\begin{restatable}{theorem}{thmWeakToStrongCover}\label{thm:WeakToStrongCover:main}
Let $(H=(V, E), w: E\rightarrow \Z_+)$ be a hypergraph and $p:2^V\rightarrow \Z$ be a symmetric skew-supermodular function such that 
$b_{(H, w)}(X) \geq p(X)$ for every $X \subseteq V$. 
Then, there exists a hypergraph \linebreak $\left(H^* = (V, E^*), w^*:E^* \rightarrow\Z_+\right)$ 
such that
\begin{enumerate}[label=$(\arabic*)$, ref=(\arabic*)]
    \item \label{thm:WeakToStrongCover:main:(1)} $d_{(H^*, w^*)}(X) \geq p(X)$ for every $X\subseteq V$,
    \item \label{thm:WeakToStrongCover:main:(2)} the hypergraph $(H^*,w^*)$ is obtained by merging hyperedges of the hypergraph $(H, w)$, and
    \item \label{thm:WeakToStrongCover:main:(3)} $|E^*| - |E| = O(|V|)$.
\end{enumerate}
Furthermore, given a hypergraph $(H=(V, E), w: E\rightarrow \Z_+)$ and access to \functionMaximizationOracleStrongCover{p} of a symmetric skew-supermodular function $p: 2^V\rightarrow \Z$ where $b_{(H, w)}(X)\ge p(X)$ for every $X\subseteq V$, there exists an algorithm that runs in $O(|V|^3(|V| + |E|)^2)$ time using $O(|V|^3(|V|+|E|))$ queries to \functionMaximizationOracleStrongCover{p} and returns a hypergraph $\left(H^* = (V, E^*), w^*:E^* \rightarrow\Z_+\right)$ satisfying the above three properties. The run-time includes the time to construct the hypergraphs used as input to the queries to \functionMaximizationOracleStrongCover{p}.
Moreover, for each query to \functionMaximizationOracleStrongCover{p}, the hypergraph $(G_0, c_0)$ used as input to the query has $O(|V|)$ vertices and $O(|V|+|E|)$ hyperedges.
\end{restatable}

\Cref{thm:WeakToStrongCover:main} is a strengthening of Theorem \ref{thm:bernath-kiraly} in two ways. Firstly, our theorem shows the existence of a hypergraph that not only satisfies properties \ref{thm:WeakToStrongCover:main:(1)} and \ref{thm:WeakToStrongCover:main:(2)}, but also satisfies property \ref{thm:WeakToStrongCover:main:(3)} -- i.e., the number of additional hyperedges in the returned hypergraph is \emph{linear} in the size of the vertex set. Secondly, our \Cref{thm:WeakToStrongCover:main} shows the existence of a strongly polynomial-time algorithm that returns a hypergraph satisfying the three properties. 
Our main contribution is modifying Bern\'{a}th and Kir\'{a}ly's algorithm and analyzing the modified algorithm to bound the number of additional hyperedges and the run-time. 
We mention that property \ref{thm:WeakToStrongCover:main:(3)} cannot be tightened to guarantee that $|E^*-E|=O(|V|)$ -- we were able to construct an example where 
$|E^*-E|=\Omega(|V|^2)$ (see \Cref{sec:appendix:tight-example-weak-to-strong}). 


\Cref{thm:WeakToStrongCover:main} immediately leads to  a proof of \Cref{thm:complete-splitting-off} (see Theorem \ref{thm:hypergraphs-splitting-off-stronger-form} and its proof in Section \ref{sec:hypergraph-splitting-off-proof}). Instead of using \Cref{thm:WeakToStrongCover:main} as a black-box, if we delve into the proof of it in the context of the proof of \Cref{thm:complete-splitting-off}, we obtain the following theorem: 

\begin{theorem}\label{thm:splitting-off-or-trim-to-preserve-connectivity}
    Let $(G=(V, E), w: E\rightarrow \Z_+)$ be a hypergraph and $s\in V$. Then, there exists a hypergraph $(G'=(V, E'), w': E'\rightarrow \Z_+)$ obtained 
    by applying a h-splitting-off operation at $s$ from $(G, w)$ 
    such that $\lambda_{(G', w')}(u, v)=\lambda_{(G, w)}(u, v)$ for every distinct $u, v\in V\setminus \{s\}$. 
\end{theorem}
We omit the proof of \Cref{thm:splitting-off-or-trim-to-preserve-connectivity} in the interests of brevity. 
\Cref{thm:splitting-off-or-trim-to-preserve-connectivity} closely resembles the \emph{existential edge splitting-off} results of Lov\'{a}sz \cite{Lov74, Lovasz-problems-book} and Mader \cite{Mad78} for graphs. Lov\'{a}sz's and Mader's  existential edge splitting-off results for graphs are important since they have been used to simplify the proofs of fundamental results in graph theory---e.g., Nash-Williams' Strong Orientation Theorem. 
On the other hand, \Cref{thm:splitting-off-or-trim-to-preserve-connectivity} does not immediately imply a strongly polynomial-time algorithm for finding a local connectivity preserving complete h-splitting off at a vertex from a given weighted hypergraph. So, \Cref{thm:complete-splitting-off} may be useful in algorithmic contexts while \Cref{thm:splitting-off-or-trim-to-preserve-connectivity} may be useful in graph-theoretical contexts. 

\subsection{Proof Technique for \Cref{thm:WeakToStrongCover:main}}
In this section, we describe our proof technique for the existential result in \Cref{thm:WeakToStrongCover:main}. 
The algorithmic results in that theorem follow from the existential result using standard algorithmic tools for submodular functions, so we focus only on outlining a proof of the existential result. 
Let $(H=(V, E), w: E\rightarrow \Z_+)$ be a hypergraph and $p: 2^V\rightarrow \Z_+$ be a symmetric skew-supermodular function such that $(H, w)$ weakly covers the function $p$. Our goal is to show that there exists a hypergraph $(H^*=(V, E^*), w^*: E\rightarrow \Z_+)$ such that 
\begin{enumerate}[label=(\arabic*)]
    \item \hypertarget{WeakToStrong(1)}{$(H^*, w^*)$ strongly covers the function $p$},
    \item \hypertarget{WeakToStrong(2)}{$(H^*, w^*)$ is obtained by merging hyperedges of the hypergraph $(H, w)$}, and
    \item \hypertarget{WeakToStrong(3)}{$|E^*|-|E|=O(|V|)$}. 
\end{enumerate}

\paragraph{Preliminaries.} We define a set $X\subseteq V$ to be $(p, H, w)$-tight if $b_{(H, w)}(X) = p(X)$. For a function $p$ and hypergraph $(H, w)$, let $T_{p, H, w}$ denote the family of $(p,H, w)$-tight sets and let $\mathcal{T}_{p, H, w}$ be the family of inclusionwise maximal sets in $T_{p, H, w}$. We will need the following two operations: 
\begin{enumerate}[label=(\roman*)]
    \item For hyperedges $e, f\in E$ and a positive integer  $\alpha \le \min\{w(e), w(f)\}$, the operation \merge$((H, w), e, f, \alpha))$ returns the hypergraph obtained from $(H, w)$ by decreasing the weight of hyperedges $e$ and $f$ by $\alpha$ and increasing the weight of the hyperedge $e\cup f$ by $\alpha$. All hyperedges with zero weight are discarded.
    \item For a hyperedge $e\in E$ and a positive integer $\alpha\le w(e)$, the operation \reduce$((H, w), e, \alpha)$ returns the hypergraph obtained by decreasing the weight of the hyperedge $e$ by $\alpha$. All hyperedges with zero weight are discarded.
\end{enumerate}
 


\paragraph{Algorithm of \cite{Bernath-Kiraly}.} Our proof of the existential result builds on the techniques of Bern\'{a}th and Kir\'{a}ly \cite{Bernath-Kiraly} who proved the existence of a hypergraph $(H^*=(V, E^*), w^*: E^*\rightarrow \Z_+)$ satisfying properties \hyperlink{WeakToStrong(1)}{(1)} and \hyperlink{WeakToStrong(2)}{(2)}, so we briefly recall their techniques. 
We present the algorithmic version of their proof since it will be useful for our purposes. 

The proof in \cite{Bernath-Kiraly} is inductive, and consequently, the algorithm implicit in their proof is recursive. The algorithm takes as input a hypergraph $((H=(V, E), w: E\rightarrow \Z_+)$ and a symmetric skew-supermodular function $p: 2^V\rightarrow \Z$ such that the hypergraph $(H,w)$ weakly covers the function $p$. 
If $w(E) = 0$, then the algorithm is in its base case and returns the empty hypergraph. 
Otherwise, $w(E) > 0$; the algorithm chooses an arbitrary hyperedge $e \in E$ and defines hypergraphs $(H_0, w_0)$ and $(H', w')$ and the function $p'$ by considering two cases.
First, suppose that the hyperedge $e$ is not contained in any set of the family $\maximalTightSetFamily{p, H, w}$. 
In this case, the algorithm defines $(H_0, w_0)$ to be the hypergraph on vertex set $V$ consisting of a single hyperedge $e$ with $w_0(e) = 1$, constructs the hypergraph $(H', w'):=\reduce((H, w), e, 1)$, and defines the function $p' := p - d_{(H_0, w_0)}$. Second, suppose that the hyperedge $e$ is contained in some set $X \in \mathcal{T}_{p, H, w}$. It can be shown that there exists a hyperedge $f \in E$ such that $f \subseteq V - X$. In this case, the algorithm defines $(H_0, w_0)$ to be the empty hypergraph on vertex set $V$, constructs the hypergraph $(H', w'):=\merge((H, w), e, f, 1)$, and defines the function $p' := p$. In both cases, the algorithm recurses on the inputs $(H',w')$ and $p'$ to obtain a hypergraph $(H^*_0, w^*_0)$ and returns the hypergraph $(H^*, w^*) = (H^*_0 + H_0, w^*_0 + w_0)$. 
Here, the hyperedges of $H^*$ are the union of the hyperedges of $H^*_0$ and $H_0$ with the weight $w^*(e)$ of a hyperedge $e$ being the sum of the weights $w^*_0(e)+w_0(e)$ if the hyperedge $e$ is present in both $H^*_0$ and $H_0$, being $w^*_0(e)$ if the hyperedge $e$ is present only in $H^*_0$, and being $w_0(e)$ if the hyperedge $e$ is present only in $H_0$. 

 We note that $w'(E') = w(E) - 1$. Furthermore, it can be shown that the function $p'$ is symmetric skew-supermodular and the hypergraph $(H', w')$ weakly covers the function $p'$. Consequently, by induction on $w(E)$, the algorithm can be shown to terminate in $w(E)$ recursive calls and returns a hypergraph satisfying properties \hyperlink{WeakToStrong(1)}{(1)} and \hyperlink{WeakToStrong(2)}{(2)}. Moreover, the number of additional hyperedges in the returned hypergraph is at most the
number of recursive calls where the \merge operation is performed, which is also at most $w(E)$. 
Thus, in order
to reduce the number of additional hyperedges and to design a strongly polynomial-time algorithm, our goal is to
reduce the recursion depth of the algorithm. We emphasize that the recursion depth of Bern\'{a}th and Kir\'{a}ly's algorithm could indeed be exponential (the exponential sized example mentioned in Remark \ref{remark:exp-sized-conn-preserving-splitting-off} could arise in the execution of their algorithm), so we need to necessarily modify their algorithm. 

\paragraph{Preprocessing for Additional Structure.}
Similar to Bern\'{a}th and Kir\'{a}ly's algorithm, our algorithm also takes as input a hypergraph $(H=(V, E), w: E\rightarrow \Z_+)$ and a symmetric skew-supermodular function $p: 2^V\rightarrow \Z$ such that the hypergraph $(H,w)$ weakly covers the function $p$. 
However, unlike Bern\'{a}th and Kir\'{a}ly's algorithm, our algorithm performs a preprocessing step so that the inputs $(H,w)$ and $p$ satisfy the following two additional conditions: 
\begin{enumerate}[label=(\alph*)]
    \item \hypertarget{WeakToStrong(a)}{every hyperedge $e\in E$ is contained in some set of $\mathcal{T}_{p, H, w}$} and 
    \item \hypertarget{WeakToStrong(b)}{the degree of every vertex in $(H, w)$ is non-zero.}
\end{enumerate}
As a consequence of these additional conditions, the family $\maximalTightSetFamily{p, H, w}$ will be a \emph{disjoint family},  
a property that we leverage heavily throughout our analysis.
Furthermore, we modify Bern\'{a}th and Kir\'{a}ly's algorithm to ensure that these conditions hold during every recursive call.

\paragraph{Our Algorithm.}
We now describe our modification of the above-mentioned Bern\'{a}th and Kir\'{a}ly's algorithm to reduce the recursion depth. Our algorithm is also recursive and its 
base case is the same as that of Bern\'{a}th and Kir\'{a}ly's algorithm (i.e., $w(E)=0$). During recursive cases (i.e., if $w(E) > 0$), instead of performing one of the two (i.e., \merge or  \reduce) operations, our algorithm performs both operations in a sequential fashion. In particular, we find a pair of disjoint hyperedges $e, f\in E$ contained in distinct sets of $\maximalTightSetFamily{p, H, w}$ (such a pair exists by condition \hyperlink{WeakToStrong(a)}{(a)} and the arguments of Bern\'{a}th and Kir\'{a}ly mentioned above). Next, 
instead of performing one \merge operation (as was done by Bern\'{a}th and Kir\'{a}ly's algorithm), we perform as many  \merge operations as possible 
using the hyperedges $e$ and $f$. 
Formally, let $$\alpha^M := \max\left\{\alpha \in \Z_+ : \text{hypergraph returned by } \merge((H, w), e, f, \alpha) \text{ weakly covers } p\right\}.$$
We let $(H^M, w^M):=\textsc{Merge}((H, w), e, f, \alpha^M)$ and $p^M:=p$.
Next, instead of recursing on $((H^M, w^M), p^M)$ (as was done by Bern\'{a}th and Kir\'{a}ly's algorithm), we perform as many \reduce operations as possible 
on the newly created hyperedge $e\cup f$.
Formally, let $$\alpha^R := \max\left\{\alpha \in \Z_{\ge 0}: \text{hypergraph returned by } \reduce((H^M, w^M), e \cup f, \alpha) \text{ weakly covers } (p^M - d_{(H_0^\alpha, w_0^\alpha)})\right\},$$
where $(H_0^{\alpha}, w_0^{\alpha})$ denotes
the hypergraph on vertex set $V$ consisting of a single hyperedge $e\cup f$ with $w_0^{\alpha}(e\cup f) = \alpha$. 
We construct the hypergraph 
$(H_0, w_0):=(H_0^{\alpha^R}, w_0^{\alpha^R})$, the hypergraph $(H^R, w^R):=\reduce((H^M, w^M), e\cup f, \alpha^R)$ and define the function $p^R:=p^M - d_{(H_0^{\alpha^R}, w_0^{\alpha^R})}$.
This immediate reduce step ensures that the hypergraph $(H^R, w^R)$ and the function $p^R$ satisfy condition \hyperlink{WeakToStrong(a)}{(a)} -- i.e., every hyperedge in $(H^R, w^R)$ is contained in some set of $\mathcal{T}_{p^R, H^R, w^R}$. 
Finally, we compute sets $\zeros:=\{u\in V: b_{(H^R, w^R)}(u)=0\}$ and $V':=V-\zeros$, hypergraph $(H':=(V', E':=E^R), w':=w^R)$, and define the function $p':2^{V'}\rightarrow \Z$ by $p'(X):=\max\{p(X\cup R): R\subseteq \zeros\}$ for every $X\subseteq V'$ -- this final step can be viewed as a clean up step since it gets rid of vertices that are not incident to any hyperedges (and revises the function $p$ appropriately). 
This clean up step ensures that the hypergraph $(H', w')$ satisfies condition \hyperlink{WeakToStrong(b)}{(b)} -- i.e., the degree of every vertex in $(H', w')$ is non-zero. 
It can be shown that the function $p'$ is symmetric skew-supermodular and the hypergraph $(H', w')$ weakly covers the function $p'$. Furthermore, the function $p'$ and hypergraph $(H', w')$ satisfy conditions \hyperlink{WeakToStrong(a)}{(a)} and \hyperlink{WeakToStrong(b)}{(b)}. 
We recursively call the algorithm on input $((H', w'), p')$ to obtain a hypergraph $(H_0^*, w_0^*)$. We obtain the hypergraph $(G, c)$ from $(H_0^*, w_0^*)$ by adding the vertices $\zeros$ and return the hypergraph $(G+H_0, c+w_0)$.  
By induction on the total weight of hyperedges in the input hypergraph, it can be shown that our algorithm returns a hypergraph satisfying properties \hyperlink{WeakToStrong(1)}{(1)} and \hyperlink{WeakToStrong(2)}{(2)} and also terminates within a finite number of recursive calls. 

\paragraph{Recursion Depth and Potential Functions.} We now sketch our proof to show that the recursion depth of our algorithm is $|E|+O(|V|)$. We note that this also bounds the number of additional hyperedges in the hypergraph returned by the algorithm, and consequently this hypergraph also satisfies property \hyperlink{WeakToStrong(3)}{(3)}. Let $\ell$ be the number of recursive calls made by the algorithm on the input instance $((H, w), p)$. 
We partition 
the set $[\ell]$ of recursive calls into two parts: let $P_1\subseteq [\ell]$ be the set of recursive calls during which 
the merged hyperedge $e\cup f$ survives in the hypergraph $(H', w')$ that is input to the subsequent recursive call and 
$P_2\subseteq [\ell]$ be the recursive calls during which 
the merged hyperedge $e\cup f$ does not survive in the hypergraph $(H', w')$ that is input to the subsequent recursive call. 
We note that $[\ell]=P_1\uplus P_2$. We bound $|P_1|$ and $|P_2|$ separately using certain carefully designed potential functions.

First, we show that $|P_1|=O(|V|)$ as follows: for $i\in [\ell]$, consider the maximal tight set family $\mathcal{T}_i=\mathcal{T}_{p_i, H_i, w_i}$ where $((H_i, w_i), p_i)$ is the input to the $i^{th}$ recursive call. Also, let $\mathcal{T}_{\le 1}:=\mathcal{T}_1$ and $\mathcal{T}_{\le i}:=\mathcal{T}_{i}\cup \{X\cap V_i: X\in \mathcal{T}_{\le i-1}\}$ for integers $i$ where $2 \le i\le \ell$ and $V_i$ is the ground set of the input to the $i^{th}$ recursive call. Thus, $\calT_{\leq i}$ is the \emph{projection} of all the maximal tight sets encountered in the first $i$ recursive calls of the algorithm onto the ground set of the inputs at the $i^{th}$ recursive call. We show that $\mathcal{T}_{\le i}$ is laminar for every $i\in [\ell]$ (\Cref{lem:WeakToStrong:properties-of-cumulative-tight-set-families}). 
However, $|\mathcal{T}_{\le i}|$ is not necessarily non-decreasing with $i$ since projection of a set family to a subset could result in the loss of sets from the family. Consequently, $|\mathcal{T}_{\le i}|$ is not suitable as a potential function to measure progress. Instead, we use the potential function $\phi(i):=|\mathcal{T}_{\le i}|+3|\zeros_{\le i-1}|$, where $\zeros_{\le i}$ is the union of the sets $\zeros$ computed up to the $i^{th}$ recursive call. We show that $\phi(i)$ is non-decreasing and strictly increases if $i\in P_1$ (\Cref{claim:WeakToStrong:hyperedge-support-size:alpha^R=weight-of-added-hyperedge} in \Cref{lem:WeakToStrongCover:NumberOfHyperedges}). Consequently, $|P_1|=O(|V|)$.

Secondly, we bound $|P_2|$ as follows. We use a lookahead-potential function: let $\Phi_1(i)$ be the number of recursive calls between $i$ and $\ell$
during which the merged hyperedge $e\cup f$ survives in the hypergraph $(H', w')$ that is input to the subsequent recursive call
and let $\Phi(i):=|E_i|+\Phi_1(i)$, where $E_i$ is the set of hyperedges in the hypergraph $(H_i, w_i)$ input to the $i^{th}$ recursive call. We show that $\Phi(i)$ is non-increasing and strictly decreases if 
$i\in P_2$ 
(\Cref{claim:WeakToStrong:hyperedge-support-size:alpha^R<weight-of-added-hyperedge} in \Cref{lem:WeakToStrongCover:NumberOfHyperedges}). Hence, $|P_2|\le \Phi(1)-\Phi(\ell)\le \Phi(1) \le |E_1|+|P_1|=|E|+O(|V))$, where the last equality is because of the bound on $|P_1|$ from the previous paragraph.  

\begin{remark}
    Our key technical contributions are twofold. Our first key technical contribution is identifying conditions \hyperlink{WeakToStrong(a)}{(a)} and \hyperlink{WeakToStrong(b)}{(b)} under which $\maximalTightSetFamily{p, H, w}$ becomes a disjoint family. We ensure that conditions \hyperlink{WeakToStrong(a)}{(a)} and \hyperlink{WeakToStrong(b)}{(b)} hold in every recursive call by performing \emph{immediate} reduction and clean-up steps in the algorithm. 
    Our second key technical contribution is identifying appropriate potential functions to measure progress of the algorithm. The disjointness of $\maximalTightSetFamily{p, H, w}$ was crucial for identifying the laminar structure of the family of projected maximal tight sets across recursive calls, which was subsequently helpful in bounding the number of recursive calls corresponding to $P_1$. 
    Moreover, we bound the number of recursive calls corresponding to $P_2$ using a lookahead-potential function that relates $|P_2|$ to $|P_1|$. As discussed above, the additive $|E|$ in the recursion depth comes from the bound on $|P_2|$.
\end{remark}

%% file: preliminaries-2.tex
\section{Preliminaries}\label{sec:prelims}


We denote the sets of 
\emph{integers}, \emph{non-negative integers}, and \emph{positive integers} by 
$\mathbb{Z}$, $\mathbb{Z}_{\geq 0}$, and $\mathbb{Z}_+$ respectively. For a positive integer $n$, we use $[n]$ to denote the set $\{1, 2, \ldots, n\}$. 
Let $V$ be a finite set. We use $2^{V}$ to denote all subsets of $V$. 
Let $X, Y\subseteq V$. We will denote $X\setminus Y$ by $X-Y$ and $X\cup Y$ by $X+Y$. If $Y$ consists of a single element $y$, then $X-\{y\}$ and $X+\{y\}$ are abbreviated as $X-y$ and $X+y$, respectively. 
Let $\calF\subseteq 2^V$ be a family of subsets of $V$. 
The family $\calF$ is 
a \emph{disjoint family} if $X\cap Y=\emptyset$ for every $X, Y\in \calF$, 
a \emph{chain family} if either $X\subseteq Y$ or $Y\subseteq X$ for every $X, Y\in \calF$, and a \emph{laminar family} if either $X\cap Y=\emptyset$ or $X\subseteq Y$ or $Y\subseteq X$ for every $X, Y\in \calF$. 
We recall that the size of a laminar family $\mathcal{F}\subseteq 2^V$ is at most $2|V|$. 
For a subset $U\subseteq V$, we define the projection of $\mathcal{F}$ to $U$ as   $\mathcal{F}|_U:=\{X\cap U: X\in \mathcal{F}\}-\{\emptyset\}$ (we will use the convention that $\mathcal{F}|_U$ is a set and not a multi-set). 
The next lemma shows that the projection of a laminar family is also laminar and the size of the projection is comparable to that of the original family. We give a proof in \Cref{appendix:sec:uncrossing-properties:projection-of-laminar-families}.  

\begin{restatable}{lemma}{projectionLaminarFamily}\label{lem:CoveringAlgorithm:projection-laminar-family}
    Let $\calL \subseteq 2^V$ be a laminar family and $\calZ \subseteq V$ be a subset of elements. 
    Let $\calL':=\calL|_{V-\zeros}$. 
    Then, the family $\calL'$ is a laminar family and $|\calL| \leq |\calL'| + 3|\calZ|$.
\end{restatable}


Let $p: 2^{V} \rightarrow \mathbb{R}$ be a set function over a finite ground set $V$. 
The function $p$ is \emph{monotone} if  $p(X)\le p(Y)$ for every $X\subseteq Y\subseteq V$. 
If $p(X)+p(Y)\le p(X\cap Y)+p(X\cup Y)$
for every $X,Y\subseteq V$, then the function $p$ is said to be \emph{supermodular}. 
If $-p$ is supermodular, then the function $p$ is said to be \emph{submodular}. 
We recall that for a hypergraph $(H=(V, E), w: E\rightarrow \Z_+)$, the cut function $d_{(H, w)}: 2^V\rightarrow \Z_{\ge 0}$ is symmetric submodular and the coverage function $b_{(H, w)}: 2^V\rightarrow \Z_{\ge 0}$ is monotone submodular \cite{Frank-book}. 
For a subset $Z \subseteq V$, the \emph{contracted function} $\functionContract{p}{Z} : 2^{V - Z}\rightarrow\Z$ is defined as $\functionContract{p}{Z}(X) \coloneqq  \max\{p(X\cup R) : R \subseteq Z\}$. We note that if $p$ is a skew-supermodular function and $Z\subseteq V$, then $\functionContract{p}{Z}$ is skew-supermodular. Moreover, if $p$ is a symmetric function and $Z\subseteq V$, then $\functionContract{p}{Z}$ is symmetric. 

We recall that in our algorithmic problems, we have access to the input set function $p$ via \functionMaximizationOracleStrongCover{p}. 
We describe an additional oracle that will simplify our proofs. We will show that the oracle defined below can be implemented in strongly polynomial time using \functionMaximizationOracleStrongCover{p}. 
\begin{definition}[$p$-maximization oracle]
    Let $p: 2^V\rightarrow \Z$ be a set function. 
    \functionMaximizationOracle{p}$((G_0, c_0), S_0, T_0)$ takes as input a hypergraph $(G_0=(V, E_0), c_0: E_0\rightarrow \mathbb{Z}_{+})$  
    and disjoint sets $S_0, T_0\subseteq V$, and returns a tuple $(Z, p(Z))$, where $Z$ is an optimum solution to the following problem:
    \begin{align*}
        \max & \left\{p(Z)-b_{(G_0, c_0)}(Z): S_0\subseteq Z\subseteq V-T_0\right\}. \tag{\functionMaximizationOracle{p}}\label{tag:oracle-main}
        \end{align*} 
\end{definition}
The evaluation oracle for a function $p$ can be implemented using one query to  
\functionMaximizationOracle{p} where the input hypergraph $(G_0, c_0)$ is the empty hypergraph and $S_0 := Z$, $T_0 := V - Z$. The next lemma shows that \functionMaximizationOracle{p} can be implemented using at most $|V|+1$ queries to \functionMaximizationOracleStrongCover{p} where the hypergraphs used as input to  \functionMaximizationOracleStrongCover{p} have size of the order of the size of the hypergraph input to \functionMaximizationOracle{p}. We give a proof in \Cref{sec:wc-oracle-from-sc-oracle}.  

\begin{restatable}{lemma}{lemmaWCoraclefromSCoracle}\label{lem:Preliminaries:wc-oracle-from-sc-oracle}
Let $p: 2^V\rightarrow \Z$ be a set function, $(G = (V, E), w:E\rightarrow\Z_+)$ be a hypergraph, and $S, T \subseteq V$ be disjoint sets. Then, $\functionMaximizationOracle{p}((G, w), S, T)$ can be implemented to run in $O(|V|(|V|+|E|))$ time using at most $|V|+1$ queries to $\functionMaximizationOracleStrongCover{p}$.
The run-time includes the time to construct the hypergraphs used as input to the queries to \functionMaximizationOracleStrongCover{p}. Moreover, each query to \functionMaximizationOracleStrongCover{p} is on an input hypergraph $(G_0, c_0)$ that has at most $|V|$ vertices and $|E|$ hyperedges.
\end{restatable}

For two hypergraphs $(G=(V, E_G), c_G:E_G\rightarrow \Z_+)$ and $(H=(V, E_H), w_H: E_H\rightarrow \Z_+)$ on the same vertex set $V$, we recall that the hypergraph $(G+H=(V, E_{G+H}), c_G+w_H)$ is the hypergraph with vertex set $V$ and hyperedge set $E_{G+H}\coloneqq E_G\cup E_H$ with the weight of every hyperedge $e\in E_G\cap E_H$ being $c_G(e) + w_H(e)$, the weight of every hyperedge $e\in E_G\setminus E_H$ being $c_G(e)$,  and the weight of every hyperedge $e\in E_H\setminus E_G$ being $w_H(e)$. 

%% file: hypergraph-splitting-off-proof-2.tex
\section{Local Connectivity Preserving Complete h-Splitting-Off}
\label{sec:hypergraph-splitting-off-proof}
We prove Theorem \ref{thm:complete-splitting-off} using Theorem \ref{thm:WeakToStrongCover:main} in this section. 
We need the following lemma showing that a certain set function $p: 2^V\rightarrow \Z$ is symmetric skew-supermodular and admits an efficient algorithm to implement \functionMaximizationOracleStrongCover{p}. This lemma has been observed in the literature before. We refer the reader to \Cref{sec:helper-for-p-max-oracle} for its proof. 

\begin{restatable}{lemma}{lempmaxoracle}\label{lemma:helper-for-p-max-oracle}
    Let $(G=(V, E), c: E\rightarrow \Z_+)$ be a hypergraph. Let $r:\binom{V}{2}\rightarrow \Z_{\ge 0}$ be a function on pairs of elements of $V$ and $R, p_{(G, c, r)}:2^V\rightarrow \Z$ be functions defined by $R(X)\coloneqq \max\{r(\{u, v\}): u\in X, v\in V\setminus X\}$ for every $X\subseteq V$, $R(\emptyset)\coloneqq 0$, $R(V)\coloneqq 0$, and $p_{(G, c, r)}(X)\coloneqq R(X)-d_{(G, c)}(X)$ for every $X\subseteq V$. 
    \begin{enumerate}
        \item The function $p_{(G, c, r)}$ is symmetric skew-supermodular. 
        \item 
        Given hypergraph $(G, c)$, function $r$, hypergraph $(G_0=(V, E_0), c_0: E_0\rightarrow \Z_{+})$, and disjoint sets $S_0, T_0\subseteq V$, the oracle   \functionMaximizationOracleStrongCover{p_{(G, c, r)}}$((G_0, c_0), S_0, T_0)$ can be computed in $O(|V|^3(|V|+|E_0|+|E|)(|E_0|+|E|))$ time.
    \end{enumerate}
\end{restatable}

We now prove the following stronger form of Theorem \ref{thm:complete-splitting-off}.
\begin{theorem}\label{thm:hypergraphs-splitting-off-stronger-form}
    Let $(G=(V, E), w_G: E\rightarrow \Z_+)$ be a hypergraph and $s\in V$. Then, there exists a hypergraph $(G^*=(V, E^*), c^*)$ such that 
    \begin{enumerate}
        \item $(G^*, c^*)$ is a local connectivity preserving complete h-splitting-off at $s$ from $(G, w_G)$ and 
        \item $|E^*|-|E|=O(|V|)$. 
    \end{enumerate}
    Moreover, given the hypergraph $(G, w_G)$, a hypergraph $(G^*, c^*)$ satisfying the above two properties can be obtained in $O(|V|^6(|V|+|E|)^3)$ time. 
\end{theorem}
\begin{proof}
    We begin by proving the existence of a local connectivity preserving complete h-splitting-off at $s$ from $(G=(V, E), w_G: E\rightarrow \Z_+)$ such that the number of additional hyperedges is $O(|V|)$. 
    We will use \Cref{thm:WeakToStrongCover:main} with an appropriate setting of the symmetric skew-supermodular function $p$ and the hypergraph $(H, w)$. 
    Let $V'\coloneqq V\setminus \{s\}$. 
    Let the hypergraph $(H:=(V', E_H), w_H: E_H\rightarrow \Z_+)$ be defined by $E_H\coloneqq \{e-\{s\}: e\in \delta_G(s)\}$ with $w_H(e-\{s\})\coloneqq w_G(e)$ for every $e\in \delta_G(s)$. 
    In order to define the function $p: 2^{V'}\rightarrow \Z$, we consider another hypergraph $(J:=({V'}, E_J), w_J: E_J\rightarrow \Z_+)$ defined by $E_J\coloneqq E\setminus B_G(s)$ with $w_J(e)\coloneqq w_G(e)$ for every $e\in E\setminus B_G(s)$. 
    We consider the function $R:2^{V'}\rightarrow \Z_{\ge 0}$ defined by $R(X)\coloneqq \max\{\lambda_{(G, w_G)}(u, v): u\in X, v\in V'\setminus X\}$ for every non-empty proper subset $X\subsetneq  V'$, $R(\emptyset)\coloneqq 0$, and $R(V')\coloneqq 0$. 
    Now, consider the function $p: 2^{V'}\rightarrow \Z$ defined by $p(X)\coloneqq R(X)-d_{(J, w_J)}(X)$ for every $X\subseteq {V'}$. 
    The function $p$ is symmetric skew-supermodular by \Cref{lemma:helper-for-p-max-oracle}. 
    We also note that $b_{(H, w_H)}(X)\ge p(X)$ for every $X\subseteq {V'}$: for an arbitrary $X\subseteq {V'}$, we have that $b_{(H, w_H)}(X)+d_{(J, w_J)}(X)=d_{(G, w_G)}(X)\ge R(X)$ where the last inequality is by definition of the function $R$ and hence, $b_{(H, w_H)}(X)\ge R(X)-d_{(J, w_J)}(X)=p(X)$. 
    Hence, $(H,w_H)$ weakly covers the symmetric skew-supermodular function $p$. 
    Applying \Cref{thm:WeakToStrongCover:main} to the hypergraph $(H, w_H)$ and the function $p$, we obtain that there exists a hypergraph $(H^* = (V', E_{H^*}), w_{H^*}: E_{H^*}\rightarrow \Z_+)$ satisfying the following three properties: (i) $d_{(H^*, w_{H^*})}\ge p(X)$ for every $X\subseteq V'$,  (ii) the hypergraph $(H^*, w_{H^*})$ is obtained by merging hyperedges of the hypergraph $(H, w_H)$, and (iii) $|E_{H^*}|-|E_H|=O(|V|)$. 
    Consider the hypergraph $(G^*=(V, E^*), c^*)$ obtained from the hypergraph $(J+H^*, w_J+w_{H^*})$ by adding the vertex $s$. 
    We have that 
    \begin{align*}
    |E^*|
    &=|E_J\cup E_{H^*}| &\text{(by definition of $E^*$)}\\
    &\le |E_J|+|E_{H^*}| &\\
    &=|E_J|+|E_H|+O(|V|) & \text{(by property (iii))}\\
    &=|E_J\cup E_H|+O(|V|) & \text{(by definitions of $E_J$ and $E_H$)}\\
    &=|E|+O(|V|), & \text{(by definitions of $E_J$ and $E_H$)}
    \end{align*}
    and hence, $|E^*|-|E|=O(|V|)$. 
    
    We now show that $(G^*, c^*)$ is a complete h-splitting-off at $s$ from $(G, w_G)$ and $(G^*, c^*)$ preserves local connectivity. 
    We note that property (ii) implies that $(G^*, c^*)$ is a complete h-splitting-off at $s$ from $(G, c)$: this is because, 
    hyperedges of $E_{H^*}\setminus E_H$ that are obtained by merging hyperedges of $E_H$ are equivalent to being obtained by a sequence of h-merge almost-disjoint hyperedges operation starting from the hypergraph $(G, w_G)$ and a final h-trim operation; 
    moreover, hyperedges of $E_{H^*}\cap E_H$ are equivalent to being obtained by the h-trim operation from the hypergraph $(G, w_G)$;  
    finally, we have that $d_{(G^*, c^*)}(s)=0$ by definition of $(G^*, c^*)$. 
    Next, we note that property (i) implies that $(G^*, c^*)$ preserves local connectivity: consider an arbitrary pair of distinct vertices $u, v\in V'$. We recall that h-merge almost-disjoint hyperedges as well as h-trim operations do not increase the cut value of any subset $X\subseteq V'$ and hence, $\lambda_{(G^*, c^*)}(u, v) \le d_{(G^*, c^*)}(X)\le d_{(G, w_G)}(X)$ for every $X\subseteq V'$ such that $u\in X, v\in V'\setminus X$ and in particular, $\lambda_{(G^*, c^*)}(u, v) \le \lambda_{(G, w_G)}(u, v)$. Hence, it suffices to show that $d_{(G^*, c^*)}(X)\ge \lambda_{(G, w_G)}(u, v)$ for every $X\subseteq V'$ such that $u\in X, v\in V'\setminus X$. Let $X\subseteq V'$ such that $u\in X, v\in V'\setminus X$. We have that 
    \begin{align*}
    d_{(G^*, c^*)}(X)
    &=d_{(J, w_J)}(X) + d_{(H^*, w_{H^*})}(X) & \text{(by definition of $(G^*, c^*)$)}\\
    &\ge d_{(J, w_J)}(X) +p(X) & \text{(by property (i) of $(H^*, w_{H^*})$)}\\
    &=d_{(J, w_J)}(X) +R(X)-d_{(J, w_J)}(X) & \text{(by definition of the function $p$)}\\
    &=R(X) \\
    &\ge \lambda_{(G, w_G)}(u,v). & \text{(by definition of the function $R$)}
    \end{align*}
    This completes the proof that $(G^*, c^*)$ is a complete splitting off at $s$ from $(G, w_G)$ and preserves local connectivity. 
    
   We now analyze the run-time of the algorithm to obtain such a hypergraph $(G^*, c^*)$. Given $(G=(V, E), w_G: E\rightarrow \Z_+)$ as input, we can construct $(H, w_H)$ and $(J, w_G)$ as above in time $O(|V| + |E|)$. We note that both $(H, w_H)$ and $(J, w_G)$ have vertex set $V-\{s\}$ and at most $|E|$ hyperedges. For the function $p$ defined as above, we consider the time to implement \functionMaximizationOracleStrongCover{p}: by Lemma \ref{lemma:helper-for-p-max-oracle}, for a given hypergraph $(G_0=(V-\{s\}, E_0), c_0: E_0\rightarrow \Z_+)$, sets $S_0, T_0\subseteq V$, and a vector $y_0\in \R^V$, the oracle \functionMaximizationOracleStrongCover{p}$((G_0, c_0), S_0, T_0, y_0)$ can be implemented to run in $O(|V|^3(|V|+|E_0|+|E|)(|E_0|+|E|))$ time. By the algorithmic conclusion of \Cref{thm:WeakToStrongCover:main}, it follows that the hypergraph $(H^*, w_{H^*})$ mentioned in the previous paragraph can be obtained in $O(|V|^6(|V|+|E|)^3)$ time. Finally, we can construct the hypergraph $(G^*, c^*)$ using the hypergraph $(H^*, w_{H^*})$ and the hypergraph $(J, w_J)$ in $O(|V|+|E_{H^*}|+|E|)=O(|V|+|E|)$ time. Thus, the overall run-time is $O(|V|^6(|V|+|E|)^3)$.
\end{proof}

%% file: characterization-of-k-connected-hypergraphs.tex
\section{Constructive Characterization of \(k\)-Hyperedge-Connected Hypergraphs}
\label{sec:characterizing-k-edge-connected-hypergraphs}
In this section, we prove Theorem \ref{thm:k-EC-hypergraph-characterization} using Theorem \ref{thm:complete-splitting-off}. 
Throughout this section, we use the term \emph{hypergraph} to refer to a \emph{multi-hypergraph}.
We define a hypergraph $H = (V, E)$ to be \emph{minimally} $k$-hyperedge-connected if $H-e$ is not $k$-hyperedge-connected for every $e\in E$. 
We need the following observation about minimally $k$-hyperedge-connected hypergraphs. 
\begin{lemma}\label{lem:minimally-k-EC-hypergraph-has-degree-k-vertex}
    Let $k \in \Z_+$ be a positive integer and $G = (V, E)$ be a minimally $k$-hyperedge-connected hypergraph with $|V| \geq 2$. Then, there exists a vertex $u \in V$ such that $d_G(u) = k$.
\end{lemma}
\begin{proof}
    Since $G$ is minimally $k$-hyperedge-connected, there exists a set $X \subseteq V$ such that $d_G(X) = k$. Let $X\subseteq V$ be an inclusion-wise minimal set such that $d_G(X)=k$. We note that if $|X| = 1$, then the claim holds. Suppose that $|X| \geq 2$.

    We first show that there exists a hyperedge $e\in E$ contained in the set $X$. By way of contradiction, suppose this is false. Let $u \in X$ be an arbitrary vertex. We note $X - \{u\} \not = \emptyset$ since $|X| \geq 2$. Furthermore, every hyperedge incident to $u$ is not contained in the set $X$, i.e. $\delta_G(u) \subseteq \delta_G(X)$. Thus, we have that $k \leq d_G(u) \leq d_G(X) = k$, where the first inequality is because $G$ is $k$-hyperedge-connected. Thus, all inequalities are equations, and we have that $d_G(u) = k$, contradicting minimality of the set $X$.

    Let $e\in E$ be a hyperedge contained in the set $X$. 
    Next, we show that 
    there exists a set $Y \subseteq V$ such that $d_G(Y) = k$ and $e \in \delta_G(Y)$. 
    By way of contradiction, suppose this is false. Then, for every $Y \subseteq V$ such that $e \in \delta_G(Y)$, we have that $d_G(Y) > k$. Thus, the hypergraph $G - e$ is $k$-hyperedge-connected, contradicting the minimally $k$-hyperedge-connected property of $G$.

    We now complete the proof of the lemma. By minimality of the set $X$, we have that $Y-X\neq \emptyset$ and $(V-Y)-X\neq \emptyset$. 
    Furthermore, since $e \subseteq X$ but $e \in \delta_G(Y)$, we have that $X-Y \not = \emptyset$ 
    and $X \cap Y \not = \emptyset$. 
    Thus, all four sets $Y-X, X-Y, X\cap Y, V-(X\cup Y)$ are non-empty.
    Hence, we have the following: $2k = d_{G}(X) + d_G(Y) \geq d_G(X\cap Y) + d_G(X\cup Y) \geq 2k$, where the first inequality is by submodularity of the cut function and the second inequality is because $G$ is $k$-hyperedge-connected. 
    Thus, all inequalities are equations and we have that $d_{G}(X\cap Y) = k$, contradicting minimality of the set $X$.
\end{proof}

We now restate and prove \Cref{thm:k-EC-hypergraph-characterization}. 

\thmConstructiveCharacterization*
\begin{proof}
    If $H$ is a $k$-hyperedge-connected hypergraph, then performing both operations mentioned in the theorem gives a $k$-hyperedge-connected hypergraph. Consequently, the reverse direction follows by induction on $|V| + |E|$. We prove the forward direction.  By way of contradiction, suppose the forward direction of the theorem is false. Let $G = (V, E_G)$ be a counter-example that minimizes $|V| + |E_G|$, i.e. $G$ is $k$-hyperedge-connected but cannot be obtained by applying operations (1) and (2). We note that $G$ is minimally $k$-hyperedge-connected since otherwise, there exists a hyperedge $e \in E_G$ such that $G - e$ is $k$-hyperedge-connected, and consequently, the hypergraph $G - e$ is also a counter-example to the forward direction of the theorem contradicting minimality of the counter-example $G$. Thus, by \Cref{lem:minimally-k-EC-hypergraph-has-degree-k-vertex}, there exists a vertex $u \in V$ such that $d_G(u) = k$. Let $H = (V - u, E_H)$ be the hypergraph obtained as a local connectivity preserving complete h-splitting-off at the vertex $u$ from $G$ as guaranteed to exist by \Cref{thm:complete-splitting-off}. Then, the hypergraph $H$ is $k$-hyperedge-connected. Moreover, $|V-u| + |E_H| < |V| + |E_G|$ because complete h-splitting-off at a vertex from a multi-hypergraph does not increase the total number of hyperedges. 
    By minimality of the counter-example $G$, the hypergraph $H$ can be obtained by starting from the single vertex hypergraph with no hyperedges and repeatedly applying operations (1) and (2). 
    We note that $G$ can be obtained from $H$ by applying operation (2). This is because $H$ was obtained by complete h-splitting-off at vertex $u$ from $G$ where $d_G(u)=k$ and operation (2) is the inverse of a complete h-splitting-off operation at a vertex with degree exactly $k$. 
    Thus, $G$ is not a counter-example to the forward direction of the theorem, a contradiction.
\end{proof}

We remark that our proof of Theorem \ref{thm:k-EC-hypergraph-characterization} is constructive since it is equivalent to a proof by induction on $|V|+|E|$. 
Using the polynomial-time algorithm in Theorem \ref{thm:complete-splitting-off}, we get the following conclusion: given a $k$-hyperedge-connected hypergraph $H$, our proof gives a polynomial-time algorithm to construct a sequence of hypergraphs $H_0, H_1, H_2, \ldots, H_t$, where $H_0$ is the single vertex hypergraph with no hyperedges, $H_t=H$ and for each $i\in [t]$, the hypergraph $H_i$ is obtained from $H_{i-1}$ by either adding a new hyperedge over a subset of vertices in $H_{i-1}$ or by $(k,p)$-pinching hyperedges in $H_{i-1}$ for some positive integer $p\le k$. 

%% file: rooted-steiner-connected-orientations-2.tex
\section{Steiner Rooted Connected Orientation of Graphs and Hypergraphs}
\label{sec:rooted-steiner-connected-orientations}
In this section, we use  \Cref{thm:complete-splitting-off} to
prove Theorems \ref{thm:rooted-steiner-orientation-in-graphs} and \ref{thm:rooted-steiner-orientation-in-hypergraphs}. 
We recall the \textsc{Max Steiner Rooted-Connected Orientation} problem: the input here is a hypergraph $G=(V, E)$, a subset $T\subseteq V$ of terminals, and a root vertex $r\in T$. The goal is to find a maximum integer $k$ and an orientation $\overrightarrow{G}$ of $G$ such that $\overrightarrow{G}$ is Steiner rooted $k$-hyperarc-connected. 
Although \textsc{Max Steiner Rooted-Connected Orientation} is NP-hard, two extreme cases of the problem are polynomial-time solvable: the number of terminals $|T|=2$ and the number of terminals $|T|=|V|$. We discuss these polynomial-time solvable cases as warm-ups towards Theorems \ref{thm:rooted-steiner-orientation-in-graphs} and \ref{thm:rooted-steiner-orientation-in-hypergraphs}. The ideas underlying these cases will be useful in proving Theorems \ref{thm:rooted-steiner-orientation-in-graphs} and \ref{thm:rooted-steiner-orientation-in-hypergraphs}. 
We recall that non-terminals are denoted as Steiner vertices. 

\paragraph{The case of $|T|=2$.} \textsc{Max Steiner Rooted-Connected Orientation} problem where we have exactly $2$ terminals is polynomial-time solvable owing to Menger's theorem in hypergraphs. 
A path between $u$ and $v$ in a hypergraph is an alternative sequence of distinct vertices and hyperedges $v_1=u, e_1, v_2, e_2, ..., e_{\ell-1}, v_{\ell}=v$ such that $v_i, v_{i+1}\in e_i$ for every $i\in [\ell-1]$. 
We recall Menger's theorem: for a hypergraph $G=(V, E)$ and a pair of distinct vertices $r, v\in V$, there exist $k$ hyperedge-disjoint paths between $r$ and $v$ in $G$ if and only if $\lambda_G(r,v)\ge k$. We present a proof of this fundamental graph-theoretical result using \Cref{thm:complete-splitting-off}. Our proof also holds for undirected graphs. 
Parts of our proof will be useful in the proof of Theorems \ref{thm:rooted-steiner-orientation-in-graphs} and \ref{thm:rooted-steiner-orientation-in-hypergraphs}. 

Let $G=(V, E)$ be a hypergraph, $r,v\in V$ be distinct vertices, and $k$ be a positive integer. It is easy to see that if there exist $k$ hyperedge-disjoint paths between $r$ and $v$ in $G$, then $\lambda_G(r,v)\ge k$. We prove the converse. Suppose that $\lambda_G(r, v)\ge k$. We would like to show that there exist $k$ hyperedge-disjoint paths between $r$ and $v$ in $G$. 
Consider performing local connectivity preserving complete h-splitting-off operations at Steiner vertices in the current hypergraph and then deleting them. This will result in a $2$-vertex hypergraph $H=(\{r,v\}, E_H)$ such that $\lambda_H(r,v)\ge k$. Since $H$ contains only $2$ vertices, it immediately follows that all hyperedges of $H$ are in fact edges, and they can be oriented to obtain $k$ hyperarc-disjoint paths 
from $v$ to $r$. 
Let $\overrightarrow{H}$ be such an orientation. 
We now extend this orientation $\overrightarrow{H}$ to an orientation $\overrightarrow{G}$ of the original hypergraph $G$ that still contains $k$ hyperarc-disjoint paths 
from $v$ to $r$. 
We use Lemma \ref{lem:lifting-Steiner-orientation-through-splitting-off} below for this extension. 
Lemma \ref{lem:lifting-Steiner-orientation-through-splitting-off} below proves a more general statement for arbitrary number of terminals and for a single h-splitting-off operation. We recall that a complete h-splitting-off at a vertex from a hypergraph is a sequence of h-splitting-off operations. 
Lemma \ref{lem:lifting-Steiner-orientation-through-splitting-off} can be applied inductively starting from $\overrightarrow{H}$ to arrive at an orientation $\overrightarrow{G}$ of the hypergraph $G$ such that there exist $k$ hyperarc-disjoint paths 
from $v$ to $r$ 
in $\overrightarrow{G}$. Such an orientation immediately gives the $k$ hyperedge-disjoint paths between $r$ and $v$ in $G$. This concludes our proof of Menger's theorem in hypergraphs and graphs (undirected). 

\begin{lemma}\label{lem:lifting-Steiner-orientation-through-splitting-off}
    Let $G=(V, E_G)$ be a hypergraph, $T\subseteq V$ be a set of terminals, $r\in T$ be a specified root vertex, and $s\in V-T$. 
    Let $H = (V, E_H)$ be a hypergraph obtained by applying a h-splitting-off operation at the vertex $s$ from $G$. 
    Suppose there exists an orientation $\overrightarrow{H}$  
    of $H$ such that $\overrightarrow{H}$ is Steiner rooted $k$-hyperarc-connected. 
    Then, there exists an orientation 
    $\overrightarrow{G}$ 
    of $G$ such that $\overrightarrow{G}$ is Steiner rooted $k$-hyperarc-connected. 
\end{lemma}
\begin{proof}
 Let $\overrightarrow{H}= (V, E_H, \head_H:E_H\rightarrow V)$ be an orientation of $H$ that is Steiner rooted $k$-hyperarc-connected.  
  We will use the orientation $\overrightarrow{H}$ to define an orientation $\overrightarrow{G}=(V, E_G, \head_G: E_G\rightarrow V)$ for the hypergraph $G$. 
  We recall that $H$ is obtained by applying a h-splitting-off operation at $s$ from $G$. We have two cases: 
\begin{enumerate}
    \item Suppose that $H$ is obtained 
    from $G$ by a h-merge almost disjoint hyperedges operation. Let $e, f \in \delta_{G}(s)$ be the hyperedges with $e\cap f = \{s\}$ that are h-merged during this operation. For notational convenience, we denote $g:= e\cup f$ as the hyperedge formed after the h-merge operation. Then, we define the orientation $\head_{G}:E_G\rightarrow V$ as follows: for each $e' \in E_G$,
$$\head_{G}(e') := \begin{cases}
    s& \text{ if } e' \in \{e, f\} \text{ and } \head_{H}(g) \not \in e'\\
    
    \head_{H}(g) &\text{ if } e' \in \{e, f\} \text{ and } \head_{H}(g) \in e'\\
    
    \head_{H}(e')& \text{ otherwise (i.e., $e'\in E_G-\{e, f\}$).}
\end{cases}$$

    \item Suppose that $H$ is obtained from $G$ 
    by a h-trim hyperedge operation. Let $e \in \delta_{G}(s)$ be the hyperedge that is h-trimmed during this operation. For notational convenience, we denote $g:= e - \{s\}$ as the hyperedge formed after the h-trim operation. We define an orientation $\head_{G}:E_G\rightarrow V$ as follows: for each $e' \in E_G$,
$$\head_{G}(e') := \begin{cases}
    \head_{H}(g)& \text{ if } e' = e\\
    \head_{H}(e')& \text{ otherwise (i.e., $e'\in E_G-\{e\}$).}
\end{cases}$$  
\end{enumerate}
        We now show that the orientation $\overrightarrow{G} = (V, E_G, \head_G)$ is Steiner rooted $k$-hyperarc-connected. We will prove this only for the first case above, i.e. $H$ is obtained from $G$ by a h-merge almost disjoint hyperedges operation, and remark that the proof for the second case is along similar lines. Let $v \in T - r$ be a terminal. Since $\overrightarrow{H}$ is a Steiner rooted $k$-hyperarc-connected orientation of $H$, there exist $k$ hyperarc-disjoint paths $\calP := \{P_1, \ldots, P_k\}$ 
        from $v$ to $r$ 
        in $\overrightarrow{H}$. We will use the set of paths $\calP$ to show that there are also $k$ hyperarc-disjoint paths 
        from $v$ to $r$ 
        in $\overrightarrow{G}$. We recall that $e, f \in \delta_G(s)$ are the hyperedges that are h-merged and $g = e \cup f $ is the hyperedge formed during the splitting-off almost disjoint hyperedges operation. First, suppose that the hyperarc $(g, \head_H(g))$ is not present in any of the paths of $\calP$. Then, $\calP$ is a set of $k$ hyperarc-disjoint paths 
        from $v$ to $r$ 
        in $\overrightarrow{G}$ and the claim holds. Next, suppose that the hyperarc $(g, \head_H(g))$ is present in the path $P_j \in \calP$ for some $j \in [k]$. We note that since the paths in $\calP$ are hyperarc-disjoint, $P_j$ is the unique path containing the hyperarc $(g, \head_H(g))$ and the paths in $\calP - \{P_j\}$ are hyperarc-disjoint in $\overrightarrow{G}$. We now define a new path $\Tilde{P_j}$ 
        from $v$ to $r$ 
        in $\overrightarrow{G}$ that is hyperarc-disjoint from the paths in $\calP - \{P_j\}$. Let $\ell \in \Z_+$ denote the length of the path $P_j$ and let the vertices and hyperarcs along the path $P_j$ be 
        $$P_j = \left(v, \left(e_0, u_1\right), u_1, \left(e_1, u_2\right), u_2,  \ldots, u_{i-1}, \left(e_{i-1}, u_{i}\right), u_i, \left(g, u_{i+1}\right), u_{i+1}, \left(e_{i+1}, u_{i+2}\right), u_{i+1}, \ldots, u_{\ell}, \left(e_\ell, r\right), r\right).$$ 
        If the vertices $u_i, u_{i+1} \in e$, then we define the path $\Tilde{P}_j$ by replacing the hyperarc $(g, \head_H(g))=(g, u_{i+1})$ with the hyperarc $(e, \head_G(e))=(e, u_{i+1})$ in $P_j$ (the case where $u_i, u_{i+1} \in f$ can be shown similarly). 
        Suppose that $u_i \in e-f$ and $u_{i+1} \in f-e$: then, we define the path $\Tilde{P}_j$ by replacing the hyperarc $(g, \head_H(g))=(g,u_{i+1})$ with 
        the sub-path $(e, s), s, (f, u_{i+1})$ in $P_j$; we note that $\head_G(e)=s$, $s$ is a tail of $f$, and $\head_G(f)=u_{i+1}$ (the case where $u_i \in f-e$ and $u_{i+1} \in e-f$ can be shown similarly). 
        In both cases, the path $\Tilde{P}_j$ is a $v$ to $r$ path in the directed hypergraph $\overrightarrow{H}$. Moreover, the $k$ paths in $\calP - \{P_j\} + \{\Tilde{P}_j\}$ are hyperarc-disjoint. Consequently, $\overrightarrow{G}$ is a Steiner rooted $k$-hyperarc-connected orientation of $G$.
\end{proof}

\paragraph{The case of $|T|=|V|$. } \textsc{Max Steiner Rooted-Connected Orientation} problem where all vertices are terminals is polynomial-time solvable owing to a characterization for the existence of rooted $k$-hyperarc-connected orientation via weak-partition-connectivity. We define weak-partition-connectivity now. 
\begin{definition}\label{defn:weak-partition-connectivity}
    Let $k \in \Z_{\geq 0}$ be a non-negative integer. A hypergraph $H = (V, E)$ is \emph{$k$-weakly-partition-connected} if for every partition $\calP$ of $V$, we have that $$\sum_{e \in E} \left(\calP\left(e\right) - 1\right) \geq k (|\calP| - 1),$$
    where for each $e \in E$, $\calP(e)$ denotes the number of parts of $\calP$ that have non-empty intersection with the hyperedge $e$. 
\end{definition}

For a hypergraph $H=(V, E)$ and a root vertex $r\in V$, we recall that an orientation $\overrightarrow{H}=(V, E, \head: E\rightarrow V)$ of $H$ is rooted $k$-hyperarc-connected if for each $u\in V$, there exist $k$ hyperarc-disjoint paths 
from the vertex $u$ to the root vertex $r$. 
The next theorem is a consequence of Theorem 9.4.13 in \cite{Frank-book} and shows that weak-partition-connectivity characterizes the existence of rooted hyperarc-connected orientations. 

\begin{theorem}[Theorem 9.4.13 in \cite{Frank-book}]\label{thm:frank:weak-pc-iff-rooted-hyparc-conn-orientation}
    Let $H = (V, E)$ be an undirected hypergraph and $r \in V$ be a root vertex. There exists a rooted $k$-hyperarc-connected orientation of $H$ if and only if $H$ is $k$-weakly-partition-connected.
\end{theorem}
Theorem \ref{thm:frank:weak-pc-iff-rooted-hyparc-conn-orientation} leads to a polynomial-time algorithm to solve \textsc{Max Steiner Rooted-Connected Orientation} problem where all vertices are terminals. 

\paragraph{Proof of Theorems \ref{thm:rooted-steiner-orientation-in-hypergraphs} and \ref{thm:rooted-steiner-orientation-in-graphs}.}
We will now use Lemma \ref{lem:lifting-Steiner-orientation-through-splitting-off} and Theorems \ref{thm:frank:weak-pc-iff-rooted-hyparc-conn-orientation} and \ref{thm:complete-splitting-off} to prove Theorems \ref{thm:rooted-steiner-orientation-in-graphs} and \ref{thm:rooted-steiner-orientation-in-hypergraphs}. 
We will prove Theorem \ref{thm:rooted-steiner-orientation-in-hypergraphs} and remark after the proof about how our proof also implies Theorem \ref{thm:rooted-steiner-orientation-in-graphs}. 
To prove Theorem \ref{thm:rooted-steiner-orientation-in-hypergraphs}, we relate connectivity and weak-partition-connectivity. We recall that a hypergraph $H=(V, E)$ is defined to be $k$-hyperedge-connected if $d_H(X)\ge k$ for every non-empty proper subset $X\subseteq V$. 
It is easy to see that if a hypergraph is $k$-weakly-partition-connected, then it is also $k$-hyperedge-connected, but the converse is not necessarily true (counterexamples exist even for graphs). We show an approximate converse: if a hypergraph is $2k$-hyperedge-connected, then it is $k$-weak-partition-connected.
\begin{lemma}\label{lem:2k-conn-implies-k-weak-pc}
    Let $H = (V, E)$ be an undirected hypergraph that is $2k$-hyperedge-connected for some non-negative integer $k$. Then, $H$ is $k$-weak-partition-connected.
\end{lemma}
\begin{proof}
    Let $\calP = (P_1, P_2, \ldots, P_t)$ 
    be a partition of $V$ such that $t > 1$.
    Let $H' = (V', E')$ be the hypergraph obtained after contracting the parts of $\calP$ and discarding the singleton hyperedges. We note that $|V'| = t$ and $|e| \geq 2$ for each $e \in E'$. 
    The following claim gives a useful intermediate step.
    \begin{claim} \label{claim:weak-partition-conn-vs-min-cut}
        $\frac{\sum_{e \in E'}|e|}{t} \leq \frac{2\sum_{e \in E'}(|e| - 1)}{t - 1}.$
    \end{claim}
    \begin{proof}
    First, we observe that $2t|E'| \leq \sum_{e\in E'}|e|$ because $|e| \geq 2$ for each $e \in E'$.  
    Adding $(t-1)\sum_{e \in E'}|e| - 2t|E'|$ to the LHS of the previous equation, and $(2t-1)\sum_{e \in E'}|e| - 2t|E'|$ to the RHS of the previous equation, we obtain the following inequality: $(t-1)\sum_{e \in E'}|e| \leq 2t\sum_{e \in E'}|e| - 2t|E'|$.
    This inequality implies the claim since $t>1$. 
    \end{proof}
    With the above claim, we have the following:
$$2k \leq \min_{i \in [t]}\{d_{H}(P_i)\} 
            = \min_{u \in V'}\{d_{H'}(u)\}
            \leq \frac{\sum_{u \in V'}d_{H'}(u)}{t}
            = \frac{\sum_{e \in E'}|e|}{t}
            \leq \frac{2\sum_{e \in E'}(|e| - 1)}{t - 1}
            =  \frac{2\sum_{e \in E'}(\calP(e) - 1)}{|P| - 1}.$$
    Here, the first inequality is because $H$ is $2k$-hyperedge-connected and the third inequality is by \Cref{claim:weak-partition-conn-vs-min-cut}. Thus, $H$ is $k$-weak-partition-connected.
\end{proof}

We now restate and prove \Cref{thm:rooted-steiner-orientation-in-hypergraphs}.

\thmRootedSteinerArcConnectedOrientationsHypergraphs*
\begin{proof}
    Let $u_1, \ldots, u_p$ denote an arbitrary ordering of the Steiner vertices. By repeated application of \Cref{thm:complete-splitting-off}, there exists a hypergraph $H = (V, E_H)$ that can be obtained by complete h-splitting-off vertices $u_1, \ldots, u_p$ in order while preserving local connectivity. Let $H' = (T, E_H)$ be the hypergraph obtained from $H$ by deleting the isolated vertices. We note that $H$ is Steiner $2k$-hyperedge-connected, and consequently, the hypergraph $H' = (T, E_H)$ obtained from $H$ by deleting the isolated Steiner vertices is $2k$-hyperedge-connected. By \Cref{lem:2k-conn-implies-k-weak-pc}, the hypergraph $H'$ is $k$-weak-partition-connected. By \Cref{thm:frank:weak-pc-iff-rooted-hyparc-conn-orientation}, there exists a rooted $k$-hyperarc-connected orientation $\overrightarrow{H'} = (T, E_{H}, \head_{H'}:E_H \rightarrow T)$ of the hypergraph $H'$. Consequently, the orientation $\overrightarrow{H} = (V, E_H, \head_H:E_H \rightarrow V)$, where $\head_H(e) := \head_{H'}(e)$ for each $e \in E_H$, is a Steiner rooted $k$-hyperarc-connected orientation of $H$. By \Cref{lem:lifting-Steiner-orientation-through-splitting-off} and induction on the number of h-splitting-off operations needed to obtain $H$ from $G$, we have that $G$ has a Steiner rooted $k$-hyperarc-connected orientation.
\end{proof}

\begin{remark}
Certain remarks regarding our proof of Theorem \ref{thm:rooted-steiner-orientation-in-hypergraphs} are in order.
\begin{enumerate}
\item Our proof of \Cref{thm:rooted-steiner-orientation-in-hypergraphs} is constructive and leads to a polynomial-time algorithm to find a Steiner rooted $k$-hyperarc-connected orientation of a $2k$-hyperedge-connected-hypergraph. 
\item Our proof of \Cref{thm:rooted-steiner-orientation-in-hypergraphs} also implies a proof of Theorem \ref{thm:rooted-steiner-orientation-in-graphs}: to prove Theorem \ref{thm:rooted-steiner-orientation-in-graphs}, we start from a graph $G=(V, E)$, but our local connectivity preserving complete h-splitting-off operations at Steiner vertices results in a hypergraph $H=(T, E_H)$ that is $2k$-hyperedge-connected; by Lemma \ref{lem:2k-conn-implies-k-weak-pc}, the hypergraph $H$ is $k$-weak-partition-connected;   
now Theorem \ref{thm:frank:weak-pc-iff-rooted-hyparc-conn-orientation} gives a rooted $k$-hyperarc-connected orientation of the resulting hypergraph. Such an orientation is extended to a Steiner rooted $k$-arc-connected orientation of the graph $G=(V, E)$ using Lemma \ref{lem:lifting-Steiner-orientation-through-splitting-off}. 
\item Our proof of \Cref{thm:rooted-steiner-orientation-in-hypergraphs} reveals the source of the $2$-factor gap in the approximate min-max relation of Kir\'{a}ly and Lau for \textsc{Max Steiner Rooted-Connected Orientation Problem}: it arises from the $2$-factor gap between connectivity and weak-partition-connectivity of hypergraphs. This insight immediately leads to a minimal example that exhibits tightness of the $2$-factor in Theorem 
\ref{thm:rooted-steiner-orientation-in-graphs}: consider the cycle graph on $3$ vertices with all vertices being terminals and an arbitrary vertex chosen as the root. This graph is $2$-edge-connected, $1$-weak-partition-connected but not $2$-weak-partition-connected, and admits a Steiner rooted $2$-arc-connected orientation but does not admit a Steiner rooted $2$-arc-connected orientation. 
\end{enumerate}
\end{remark}

%% file: weak-to-strong-cover.tex
\section{Weak to Strong Cover in Strongly Polynomial Time}\label{sec:WeakToStrongCover}
In this section, we 
prove Theorem \ref{thm:WeakToStrongCover:main}. We restate it below. 
\thmWeakToStrongCover*

Our proof of \Cref{thm:WeakToStrongCover:main} will require the following definitions. 
\begin{definition}\label{def:(pHw)-tight-sets}
    Let $p:2^V\rightarrow\Z$ be a set function and $(H=(V,E), w)$ be a hypergraph. We will say that the hypergraph $H$ is \emph{non-empty} if $E\neq \emptyset$. 
    A set $X \subseteq V$ is said to be \emph{$(p, H, w)$-tight} if $b_{(H,w)}(X)=p(X)$. We let $\tightSetFamily{p,H,w}$ denote the family of $(p, H, w)$-tight sets and $\maximalTightSetFamily{p, H, w}$  denote the family of \emph{inclusion-wise maximal} sets of $\tightSetFamily{p,H,w}$. 
    For a hyperedge $e\in E$, we let $\maximalTightSetContainingHyperedge{p, H, w}(e)$ denote an arbitrary set in $\tightSetFamily{p, H, w}$ that contains $e$; if there is no such set, then we will use the convention that $\maximalTightSetContainingHyperedge{p, H, w}(e)$ is undefined.  
    Additionally, we define two operations as follows:
    \begin{enumerate}
        \item For hyperedges $e, f\in E$ and an integer  $\alpha\le \min\{w(e), w(f)\}$, the operation $\merge((H, w), e, f, \alpha)$ returns the hypergraph obtained from $(H, w)$ by decreasing the weight of hyperedges $e$ and $f$ by $\alpha$ and increasing the weight of the hyperedge $e\cup f$ by $\alpha$ (if $\alpha=w(e)$, then we discard $e$; if $\alpha=w(f)$, then we discard $f$; if $e\cup f\not\in E$, then we add the hyperedge $e\cup f$ and set its weight to be $\alpha$). 
        \item For a hyperedge $e\in E$ and an integer $\alpha \le w(e)$, the operation $\reduce((H, w), e, \alpha)$ returns the hypergraph obtained by decreasing the weight of the hyperedge $e$ by $\alpha$ (if $\alpha=w(e)$, then we discard $e$). 
    \end{enumerate}
\end{definition}

We first describe how to prove \Cref{thm:WeakToStrongCover:main} under the assumptions that the input function $p$ and hypergraph $\left(H = (V,E),w\right)$ satisfy the following two technical conditions:
\begin{enumerate}[label=(\alph*), ref=(\alph*)]
    \item \label{cond:WeakToStrong:(a)} 
    every hyperedge $e \in E$ is contained in some $(p, H, w)$-tight set, and
    \item \label{cond:WeakToStrong:(b)} $b_{(H,w)}(u) \not = 0$ for every $u \in V$, i.e. every vertex is contained in some hyperedge.
\end{enumerate}  
Under these assumptions, we will prove \Cref{thm:WeakToStrongCover:main} using the algorithm that we describe in \Cref{sec:WeakToStrongCover:Algorithm} (see \Cref{alg:WeakToStrongCover}). In \Cref{sec:WeakToStrongCover:UncrossingProperties}, we show certain properties of set families that will be used in the description and analysis of our algorithm. In \Cref{sec:WeakToStrongCover:PartialCorrectness}, we show that our algorithm terminates, and the hypergraph returned by the algorithm satisfies properties \ref{thm:WeakToStrongCover:main:(1)} and \ref{thm:WeakToStrongCover:main:(2)} (see \Cref{lem:weak-to-strong:correctness:main:strong-cover} and \Cref{lem:weak-to-strong:correctness:main:hyperedge-merge}). In \Cref{sec:WeakToStrongCover:NumberOfHyperedges}, we show that the hypergraph returned by the algorithm additionally satisfies property \ref{thm:WeakToStrongCover:main:(3)} (see \Cref{lem:WeakToStrongCover:NumberOfHyperedges}). In \Cref{sec:WeakToStrongCover:strongly-poly-runtime:main}, we show that our algorithm can be implemented to run in the claimed strongly polynomial time bound given oracle access to \functionMaximizationOracleStrongCover{p} (see \Cref{lem:WeakToStrong:strongly-polytime:main}).
\Cref{alg:WeakToStrongCover}, \Cref{lem:weak-to-strong:correctness:main:strong-cover}, \Cref{lem:weak-to-strong:correctness:main:hyperedge-merge}, \Cref{lem:WeakToStrongCover:NumberOfHyperedges}, and \Cref{lem:WeakToStrong:strongly-polytime:main} together complete the proof of \Cref{thm:WeakToStrongCover:main} under the assumptions that conditions \ref{cond:WeakToStrong:(a)} and \ref{cond:WeakToStrong:(b)} hold.

We now describe how to circumvent assumptions \ref{cond:WeakToStrong:(a)} and \ref{cond:WeakToStrong:(b)} on the input function $p$ and hypergraph $(H,w)$. Let $e \in E$ be a hyperedge not contained in any $(p, H, w)$-tight set. Let $\alpha_e^R \coloneqq  \min\{w(e), \beta^R(e)\}$, where $$\beta^R(e) \coloneqq  \min\{b_{(H,w)}(X) - p(X) : e \subseteq X \subseteq V\}.$$
We construct the hypergraph $(H_0^e:=(V, E_0^e:=\{e\}), w_0^e:E_0\rightarrow \{\alpha^R_e\})$. Now, consider the function $p' \coloneqq  p - d_{(H_0^e, w_0^e)}$ and the  hypergraph $(H', w'):=\reduce((H, w), e, \alpha^R_e)$; we term this operation as the reduce operation with respect to the hyperedge $e$. 
Iteratively performing this reduce operation with respect to every hyperedge $e \in E$ results in a function $p_1$ and a hypergraph $(H_1 = (V, E_1), w_1)$ such that $b_{(H_1, w_1)}(X)\ge p_1(X)$ for every $X\subseteq V$ and also satisfies condition \ref{cond:WeakToStrong:(a)}. Furthermore, suppose that $\zeros \coloneqq  \{u \in V : b_{H_1, w_1}(u) = 0\}$. Then, the function $p_2 \coloneqq  \functionContract{p_1}{\zeros}$ and hypergraph $(H_2 = (V - \zeros, E_1), w_1)$ is an instance with $b_{(H_2, w_1)}(X)\ge p_2(X)$ for every $X\subseteq V$ and also satisfies conditions \ref{cond:WeakToStrong:(a)} and \ref{cond:WeakToStrong:(b)}. Finally, we observe that if $(H^*, w^*)$ is a hypergraph that strongly covers the function $p_2$, then the hypergraph $(H^*+\sum_{e \in E}H_0^e, w^*+w_0^e)$ strongly covers the original function $p$. Moreover,  if $(H^*, w^*)$ is obtained by merging hyperedges of $(H_2, w_1)$, then the  hypergraph $(H^*+\sum_{e \in E}H_0^e, w^*+w_0^e)$ is obtained by merging hyperedges of $(H, w)$. Thus, our final algorithm is to first obtain the instance $(H_2, w_1)$ and an oracle for \functionMaximizationOracleStrongCover{p_2} using \functionMaximizationOracleStrongCover{p} such that the instance $(p_2, H_2, w_1)$ is such that $(H_2, w_1)$ weakly covers $p_2$ and satisfies the technical conditions \ref{cond:WeakToStrong:(a)} and \ref{cond:WeakToStrong:(b)}, obtain a hypergraph $(H^*, w^*)$ by the procedure described in the previous paragraph, and finally return $(H^*+\sum_{e \in E}H_0^e, w^*+w_0^e)$. 

\subsection{Weak to Strong Cover Algorithm}\label{sec:WeakToStrongCover:Algorithm}
Our algorithm is recursive and takes two inputs: (1) a hypergraph $(H = (V, E), w)$, and (2) a symmetric skew-supermodular function $p:2^V \rightarrow\Z$. We note that in contrast to the statement of \Cref{thm:WeakToStrongCover:main}, the inputs $(H,w)$ and $p$ to the algorithm are required to satisfy the two technical conditions 
\ref{cond:WeakToStrong:(a)} and \ref{cond:WeakToStrong:(b)} specified in the previous section. The algorithm returns a hypergraph $\left(H^*= \left(V, E^*\right), w^*\right)$. 

We now give an informal description of the algorithm (refer to \Cref{alg:WeakToStrongCover} for a formal description). If $E=\emptyset$, then the algorithm returns the empty hypergraph. Otherwise, the algorithm chooses two hyperedges $e, f \in E$  such that $\maximalTightSetContainingHyperedge{p, H, w}(e)\neq \maximalTightSetContainingHyperedge{p, H, w}(f)$ ($\maximalTightSetContainingHyperedge{p, H, w}(e)$ is well-defined for every hyperedge $e\in E$ owing to condition \ref{cond:WeakToStrong:(a)}; in the next section, we show that $\maximalTightSetContainingHyperedge{p, H, w}(e)$ is unique for every hyperedge $e\in E$ under condition \ref{cond:WeakToStrong:(a)} and also that if $(H, w)$ is not a strong cover of $p$, then there exists a pair of hyperedges $e, f\in E$ such that $\maximalTightSetContainingHyperedge{p, H, w}(e)\neq \maximalTightSetContainingHyperedge{p, H, w}(f)$). 
Next, 
the algorithm computes the value 
 $$\beta^M \coloneqq  \min\left\{b_{(H,w)}(X) - p(X) : X \subseteq V \text{ and } e \cap X, f \cap X \not = \emptyset\right\},$$ sets $\alpha^M \coloneqq  \min\{\beta^M, w(e), w(f)\}$, and defines the hypergraph $(H^M, w^M) := \merge((H, w), e, f, \alpha^M)$.
Next, the algorithm computes the value $$\beta^R \coloneqq  \min\left\{b_{(H^M, w^M)}\left(X\right) - p\left(X\right) : e \cup f \subseteq X \subseteq V\right\},$$ 
    sets $\alpha^R \coloneqq  \min\{\beta^R, w^M(e\cup f)\}$, and defines the hypergraph $(H^R, w^R) := \reduce((H^M, w^M), e \cup f, \alpha^R)$.
It computes hypergraph $(H_0=(V, E_0:=\{e\cup f\}), w_0: E_0\rightarrow \{\alpha^R\})$ and defines $p^R:=p-d_{(H_0, w_0)}$. 
Next, the algorithm computes the set $\zeros \coloneqq  \{u \in V : b_{(H^R, w^R)}(u) = 0\}$, and recursively calls itself with the inputs 
    $(H' = (V':=V-\zeros, E':=E^R),w':=w^R) $ and $p' \coloneqq  \functionContract{p^R}{\zeros}$ to obtain a hypergraph $\left(H_0^*, w_0^*\right)$. Finally, the algorithm obtains the hypergraph $(G, c)$ from $(H_0^*, w_0^*)$ by adding the set $\zeros$ of vertices and returns the hypergraph $(G+H_0, c + w_0)$.

\begin{algorithm}[!h]
\caption{Weak to Strong Cover}\label{alg:WeakToStrongCover}
$\mathbf{\textsc{Algorithm}}\left(H = (V, E), w, p\right):$
\begin{algorithmic}[1]
    \If{$E = \emptyset$} \textsc{Base Case:}
    \State{\Return{$(H, w)$.}}
    \Else:
    \State{Pick a pair of hyperedges $e, f \in E $ s.t. $\maximalTightSetContainingHyperedge{p, H, w}(e) \not = \maximalTightSetContainingHyperedge{p, H, w}(f)$}\label{alg:WeakToStrongCover:(4)}
    \State{$\alpha^M \coloneqq  \min\begin{cases}
    \beta^M \coloneqq  \min\left\{b_{(H,w)}(X) - p(X) : X \subseteq V \text{ and } e \cap X, f \cap X \not = \emptyset\right\} & \\
    \min\left\{w(e), w(f)\right\}&
    \end{cases}
    $}\label{alg:WeakToStrongCover:(5)}
    \State{$(H^M, w^M):=\merge\left(\left(H, w\right), e, f, \alpha^M\right)$}\label{alg:WeakToStrongCover:(6)}
    \State{$\alpha^R \coloneqq  \min \begin{cases}
        \beta^R \coloneqq  \min\left\{b_{(H, w^M)}\left(X\right) - p\left(X\right) : e \cup f \subseteq X \subseteq V\right\} &\\
        w^M(e \cup f) &
    \end{cases}$}\label{alg:WeakToStrongCover:(7)}
    \State{$(H^R, w^R) \coloneqq \reduce\left(\left(H^M, w^M\right), e\cup f, \alpha^R\right)$}\label{alg:WeakToStrongCover:(8)}
    \State{Construct $\left(H_0:=\left(V, E_0:=\left\{e\cup f\right\}\right), w_0: E_0\rightarrow \left\{\alpha^R\right\}\right)$}\label{alg:WeakToStrongCover:(9)}
    \State{$p^R\coloneqq p - d_{(H_0, w_0)}$}\label{alg:WeakToStrongCover:(10)}
    
    \State{$\zeros \coloneqq  \left\{v \in V : b_{(H^R, w^R)}(v) = 0\right\}$}\label{alg:WeakToStrongCover:(11)}
    
    \State{$V' \coloneqq  V - \zeros$}\label{alg:WeakToStrongCover:(12)}
    \State{$p' \coloneqq  \functionContract{p^R}{\zeros}$}\label{alg:WeakToStrongCover:(13)}
    
    \State{$(H_0^*, w_0^*):=\textsc{Algorithm}\left(H' = \left(V', E' := E^R\right), w' := w^R , p'\right)$}\label{alg:WeakToStrongCover:(14)}
    
    \State{Obtain $(G, c)$ from $(H_0^*, w_0^*)$ by adding vertices \zeros}\label{alg:WeakToStrongCover:(15)}
    
    \State{\Return{$(G+H_0, c+w_0 )$}}\label{alg:WeakToStrongCover:(16)}
    \EndIf
\end{algorithmic}
\end{algorithm}

\subsection{Properties of Tight Sets}\label{sec:WeakToStrongCover:UncrossingProperties}
In this section, we show properties of the set families $\tightSetFamily{p, H, w}$ and $\maximalTightSetFamily{p, H, w}$ which will be useful to show correctness and convergence of \Cref{alg:WeakToStrongCover} in subsequent sections.
In particular, \Cref{lem:weak-to-strong:uncrossing:uncrossing-lemma} below shows an important property of the family of tight sets and \Cref{cor:weak-to-strong:uncrossing:maximal-tight-set-containing-hyperedge-unique} shows that for every hyperedge $e\in E$, there exists a unique maximal $(p, H, w)$-tight set containing $e$ under condition \ref{cond:WeakToStrong:(a)}. Furthermore, \Cref{lem:weak-to-strong:uncrossing:maximal-tight-set-family-disjoint} shows that the family $\maximalTightSetFamily{p, H, w}$ is disjoint under condition \ref{cond:WeakToStrong:(b)}, and consequently, \Cref{cor:weak-to-strong:uncrossing:exist-two-hyperedges-contained-in-disjoint-maximal-tight-sets} implies  that \Cref{alg:WeakToStrongCover} \Cref{alg:WeakToStrongCover:(4)} is well-defined.

\begin{lemma}\label{lem:weak-to-strong:uncrossing:uncrossing-lemma}
Let $p:2^V\rightarrow\Z$ be a skew-supermodular function and let $\left(H = (V, E), w:E \rightarrow\Z_+\right)$ be a hypergraph such that $b_{(H,w)}(Z) \geq p(Z)$ for every $Z \subseteq V$. Let $X, Y \in \tightSetFamily{p, H, w}$ be sets such that $X\cap Y \not = \emptyset$.
If there exists a hyperedge $e \in E$ such that (1) $e$ intersects the set $X\cap Y$ and 
(2) $e$ intersects at most one of the sets $X - Y$ and $Y - X$, then $X\cup Y, X \cap Y \in \tightSetFamily{p, H, w}$.
\end{lemma}
\begin{proof}
We note that the claim holds if $X\subseteq Y$ or $Y \subseteq X$. Suppose that $X - Y, Y - X \not= \emptyset$. First, suppose that the function $p$ satisfies $p(X) + p(Y) \leq p(X - Y) + p(Y - X)$. Then, the following gives us a contradiction:
\begin{align*}
    b_{(H,w)}(X) + b_{(H,w)}(Y) 
    &= p(X) + p(Y) &\\
    & \leq p(X - Y) + p(Y - X) & \\
    & \leq b_{(H,w)}(X - Y) + b_{(H,w)}(Y - X) & \\
    & < b_{(H,w)}(X) + b_{(H,w)}(Y).
\end{align*}
Here, the first equality is because $X, Y \in \tightSetFamily{p, H, w}$. The second inequality is because $b_{(H,w)}(Z) \geq p(Z)$ for every $Z\subseteq V$. The final inequality is because the hyperedge $e$ intersects at most one of the sets $X - Y$ and $Y - X$.
Thus, the function $p$ satisfies $p(X) + p(Y) \leq p(X\cap Y) + p(X \cup Y)$. Consequently, we have that
\begin{align*}
    b_{(H,w)}(X) + b_{(H,w)}(Y) 
    &= p(X) + p(Y) &\\
    & \leq p(X \cup Y) + p(X \cap Y) &\\
    & \leq b_{(H,w)}(X\cup Y) + b_{(H,w)}(X\cap Y) &\\
    & \leq b_{(H,w)}(X) + b_{(H,w)}(Y).&
\end{align*}
Here, the first equality is because $X, Y \in \tightSetFamily{p, H, w}$. The second inequality is because $b_{(H,w)}(Z) \geq p(Z)$ for every $Z\subseteq V$. The final inequality is because the coverage function $b_{(H, w)}$ is submodular. Thus, all inequalities in the above sequence are in fact equalities, and consequently we have that $X\cup Y, X\cap Y \in \tightSetFamily{p, H, w}$.
\end{proof}

\begin{corollary}\label{cor:weak-to-strong:uncrossing:maximal-tight-set-containing-hyperedge-unique}
Let $p:2^V\rightarrow\Z$ be a skew-supermodular function and let $\left(H = (V, E), w:E \rightarrow\Z_+\right)$ be a hypergraph such that  
every hyperedge $e \in E$ is contained in some $(p, H, w)$-tight set
and $b_{(H,w)}(Z) \geq p(Z)$ for every $Z \subseteq V$. 
Then, for every hyperedge $e \in E$, there is a unique set from the family $\maximalTightSetFamily{p, H, w}$ containing the hyperedge $e$.
\end{corollary}
\begin{proof}
    Let $e \in E$ be a hyperedge. We note that there exists a $(p, H, w)$-tight set containing the hyperedge $e$. Consequently, there exists a set in the family \maximalTightSetFamily{p, H, w} containing the hyperedge $e$. By way of contradiction, suppose that $X, Y \in \maximalTightSetFamily{p, H, w}$ are distinct sets containing the hyperedge $e$. We note that $X - Y, Y - X \not = \emptyset$ by the maximality of the sets $X, Y \in \maximalTightSetFamily{p, H, w}$.
    Also, $X\cap Y\neq \emptyset$ since $e$ is contained in $X\cap Y$. Moreover, $e$ is disjoint from both $X-Y$ and $Y-X$. Hence, by \Cref{lem:weak-to-strong:uncrossing:uncrossing-lemma}, we have that $X\cup Y \in \tightSetFamily{p, H, w}$, contradicting maximality of the sets $X, Y \in \maximalTightSetFamily{p, H, w}$.
\end{proof}

The next lemma shows that for a \emph{skew-supermodular} function $p$ and hypergraph $(H, w)$ satisfying conditions  \ref{cond:WeakToStrong:(a)} and \ref{cond:WeakToStrong:(b)}, the family $\maximalTightSetFamily{p, H, w}$ is disjoint. The corollary following the lemma says that if the function $p$ is also \emph{symmetric} and the hypergraph $H$ is non-empty, then there exist two hyperedges that are each contained in distinct maximal $(p, H, w)$-tight sets w.r.t. $p$ and $(H, w)$.  We note that this makes \Cref{alg:WeakToStrongCover} \Cref{alg:WeakToStrongCover:(4)} well-defined.

\begin{lemma}\label{lem:weak-to-strong:uncrossing:maximal-tight-set-family-disjoint}
    Let $p:2^V\rightarrow\Z$ be a skew-supermodular function and let $\left(H = (V, E), w:E \rightarrow\Z_+\right)$ be a hypergraph 
    such that every hyperedge $e \in E$ is contained in some $(p, H, w)$-tight set, $b_{(H,w)}(u) \not = 0$ for every $u \in V$, and $b_{(H,w)}(X) \geq p(X)$ for every $X \subseteq V$. 
    Then, the family $\maximalTightSetFamily{p, H, w}$ is a disjoint family.
\end{lemma}
\begin{proof}
By way of contradiction, let $X, Y \in \maximalTightSetFamily{p, H, w}$ be distinct sets such that $X\cap Y \not = \emptyset$. By maximality of the sets $X, Y$, we also have that $X - Y, Y - X \not = \emptyset$. Let $u \in X\cap Y$. Since $b_{(H,w)}(u) > 0$, there exists a hyperedge $e \in E$ such that $u \in e$. We note that if either $e\cap (X - Y) = \emptyset$ or $e\cap (Y - X) = \emptyset$, then the set $X \cup Y \in \tightSetFamily{p, H, w}$ by \Cref{lem:weak-to-strong:uncrossing:uncrossing-lemma}, contradicting maximality of $X, Y \in \maximalTightSetFamily{p, H, w}$. Thus, we have that $e \cap (X - Y), e \cap (Y - X) \not= \emptyset$. Furthermore, since the hyperedge $e$ is contained in some $(p, H, w)$-tight set, we have that  $\maximalTightSetContainingHyperedge{p, H, w}(e) \not = \emptyset$. Consequently, we have that $\maximalTightSetContainingHyperedge{p, H, w}(e)\cap (X - Y), \maximalTightSetContainingHyperedge{p, H, w}(e) \cap (Y - X) \not = \emptyset$. However, $e \subseteq \maximalTightSetContainingHyperedge{p, H, w}(e)$ and thus, $e \cap (X - \maximalTightSetContainingHyperedge{p, H, w}(e)) = \emptyset$. By \Cref{lem:weak-to-strong:uncrossing:uncrossing-lemma}, we have that the set $X \cup \maximalTightSetContainingHyperedge{p, H, w}(e) \in \tightSetFamily{p, H, w}$, contradicting maximality of the set $X \in \maximalTightSetFamily{p, H, w}$.
\end{proof}

We recall that a hypergraph is non-empty if it contains at least one hyperedge. 
\begin{corollary}[\hspace{-1sp}\cite{Bernath-Kiraly}]
\label{cor:weak-to-strong:uncrossing:exist-two-hyperedges-contained-in-disjoint-maximal-tight-sets}
Let $p:2^V\rightarrow\Z$ be a symmetric skew-supermodular function and $\left(H = (V, E), w:E \rightarrow\Z_+\right)$ be a non-empty hypergraph 
such that every hyperedge $e \in E$ is contained in some $(p, H, w)$-tight set, $b_{(H,w)}(u) \not = 0$ for every $u \in V$, and $b_{(H,w)}(X) \geq p(X)$ for every $X \subseteq V$. 
Then, there exist distinct hyperedges $e,f \in E$ such that $\maximalTightSetContainingHyperedge{p, H, w}(e) \cap \maximalTightSetContainingHyperedge{p, H, w}(f) = \emptyset$.
\end{corollary}
\begin{proof}
Let $e \in E$ be a hyperedge.
We will show that that there exists a hyperedge $f \in E$ such that $f \subseteq V - \maximalTightSetContainingHyperedge{p, H, w}(e)$. The claim will then follow by \Cref{lem:weak-to-strong:uncrossing:maximal-tight-set-family-disjoint} since $\maximalTightSetContainingHyperedge{p, H, w}(f) \not = \emptyset$ (because the hyperedge $f$ is contained in some $(p, H, w)$-tight set) and $\maximalTightSetContainingHyperedge{p, H, w}(e), \maximalTightSetContainingHyperedge{p, H , w}(f) \in \maximalTightSetFamily{p, H, w}$.
By way of contradiction, suppose that there is no hyperedge contained in the set $V - \maximalTightSetContainingHyperedge{p, H, w}(e)$. Then, we have that $p(V - \maximalTightSetContainingHyperedge{p, H, w}(e)) = p(\maximalTightSetContainingHyperedge{p, H, w}(e)) = b_{(H,w)}(\maximalTightSetContainingHyperedge{p, H, w}(e)) > b_{(H,w)}(V - \maximalTightSetContainingHyperedge{p, H, w}(e))$, a contradiction. Here, the first equality is because the function $p$ is symmetric and the final inequality is because every hyperedge that intersects the set $V - \maximalTightSetContainingHyperedge{p, H, w}(e)$ also intersects the set $\maximalTightSetContainingHyperedge{p, H, w}(e)$ by our contradiction assumption.
\end{proof}

\subsection{Termination and Partial Correctness of the Algorithm}
\label{sec:WeakToStrongCover:PartialCorrectness}
In this section, we show that every recursive call of \Cref{alg:WeakToStrongCover} is well-defined. 
Furthermore, we show that \Cref{alg:WeakToStrongCover} terminates and that the hypergraph returned by the algorithm satisfies properties \ref{thm:WeakToStrongCover:main:(1)} and \ref{thm:WeakToStrongCover:main:(2)} of \Cref{thm:WeakToStrongCover:main}. The main results of this section are \Cref{lem:weak-to-strong:correctness:main:strong-cover} and \Cref{lem:weak-to-strong:correctness:main:hyperedge-merge}. We prove certain additional lemmas in this section which will be useful in subsequent sections. 
 
    Suppose that the input to \Cref{alg:WeakToStrongCover} is a non-empty hypergraph $(H = (V, E), w)$ and a symmetric skew-supermodular function $p:2^V\rightarrow \Z$ 
such that every hyperedge $e \in E$ is contained in some $(p, H, w)$-tight set, $b_{(H,w)}(u) \not = 0$ for every $u \in V$, and $b_{(H,w)}(X) \geq p(X)$ for every $X \subseteq V$. 
  Then, by \Cref{cor:weak-to-strong:uncrossing:exist-two-hyperedges-contained-in-disjoint-maximal-tight-sets}, there exist hyperedges $e , f \in E$ such that $\maximalTightSetContainingHyperedge{p, H, w}(e) \not = \maximalTightSetContainingHyperedge{p, H, w}(f)$. In particular, \Cref{alg:WeakToStrongCover} \Cref{alg:WeakToStrongCover:(4)} is well-defined. 
  The next two lemmas establish several useful properties of the intermediate quantities used in the algorithm.

\begin{lemma}\label{lem:WeakToStrong:properties-of-alpha^M-and-alpha^R}
    Suppose that the input to \Cref{alg:WeakToStrongCover} is a non-empty hypergraph $(H = (V, E), w)$ and a symmetric skew-supermodular function $p:2^V\rightarrow \Z$ 
such that every hyperedge $\Tilde{e} \in E$ is contained in some $(p, H, w)$-tight set, $b_{(H,w)}(u) \not = 0$ for every $u \in V$, and $b_{(H,w)}(X) \geq p(X)$ for every $X \subseteq V$. 
    Let $e, f, \alpha^M, (H^M, w^M), \alpha^R$, $(H^R, w^R),(H_0, w_0), p^R$ and $(H' = (V', E'), w')$ be as defined by \Cref{alg:WeakToStrongCover}. 
    Then, we have the following:
    \begin{enumerate}[label=(\arabic*)]
        \item $\alpha^M \in \Z_+$,
        \item 
        $E^M=E-\{e\}\cdot \indicator_{\alpha^M= w(e)}- \{f\}\cdot \indicator_{\alpha^M = w(f)} + \{e\cup f\}$ with the weights of the hyperedges in $E^M$ being  
        \[
        w^M(\Tilde{e})=
        \begin{cases}
            w(\Tilde{e}) &\text { if } \Tilde{e}\in E-\{e, f, e\cup f\},\\
            w(e\cup f) + \alpha^M &\text { if } \Tilde{e}=e\cup f \text{ and } e\cup f \in E, \\
            \alpha^M &\text { if } \Tilde{e}=e\cup f \text{ and } e\cup f \not \in E, \\
            w(e) - \alpha^M &\text { if } \Tilde{e}=e \text{ and } \alpha^M<w(e), \\
            w(f) - \alpha^M &\text { if } \Tilde{e}=f \text{ and } \alpha^M<w(f).
        \end{cases}
        \]
        \item $b_{(H^M, w^M)}(X) \geq p(X)$ for every $X \subseteq V$,
        \item $\alpha^R \in \Z_{\geq 0}$, 
        \item 
        $E' = E^R = E^M - \{e\cup f\}\cdot \indicator_{\alpha^R=w^M(e\cup f)}$ with the weights of the hyperedges in $E' = E^R$ being 
        \[
        w'(\Tilde{e}) = w^R(\Tilde{e})=
        \begin{cases}
            w^M(\Tilde{e}) &\text { if } \Tilde{e}\in E^M-\{e\cup f\},\\
            w^M(e\cup f) - \alpha^R &\text { if } \Tilde{e}=e\cup f \text{ and } \alpha^R< w^M(e\cup f), \text{ and}  \\
        \end{cases}
        \]
        \item $b_{(H^R, w^R)}(X) \geq p^R(X)$ for every $X \subseteq V$. 
    \end{enumerate}
\end{lemma}
\begin{proof} We prove each claim separately below.
\begin{enumerate}
    \item We recall that the weight function $w:E\rightarrow\Z_+$ is a positive integer-valued function, and the function $p$ is an integer-valued function. Thus, it suffices to show that $\beta^M > 0$. We note that $\beta^M \geq 0$ since $b_{(H,w)}(X) \ge p(X)$ for every $X \subseteq V$. By way of contradiction, suppose that $\beta^M = 0$. Let $X \subseteq V$ witness $\beta^M$, i.e. $e \cap X, f\cap X \not = \emptyset$ and $b_{(H,w)}(X) -p(X)= 0$.
    The latter property implies that $X \in \tightSetFamily{p, H, w}$ and thus, there exists $Y \in \maximalTightSetFamily{p, H, w}$ such that $X \subseteq Y$. 
    We recall that by the algorithm's choice of hyperedges $e, f \in E$, we have that $\maximalTightSetContainingHyperedge{p, H, w}(e) \not = \maximalTightSetContainingHyperedge{p, H, w}(f)$. Since, $\maximalTightSetContainingHyperedge{p, H, w}(e), \maximalTightSetContainingHyperedge{p, H, w}(f) \in \maximalTightSetFamily{p, H, w}$, we have that $\maximalTightSetContainingHyperedge{p, H, w}(e) \cap \maximalTightSetContainingHyperedge{p, H, w}(f) = \emptyset$ by \Cref{lem:weak-to-strong:uncrossing:maximal-tight-set-family-disjoint}. In particular, we have that $\maximalTightSetContainingHyperedge{p, H, w}(e) \not = Y$ since
    $f \subseteq V - \maximalTightSetContainingHyperedge{p, H, w}(e)$ but $f \cap Y \not = \emptyset$.
    Thus, $Y\cap \maximalTightSetContainingHyperedge{p, H, w}(e) = \emptyset$ by \Cref{lem:weak-to-strong:uncrossing:maximal-tight-set-family-disjoint}, contradicting $\emptyset\neq e\cap X\subseteq e \cap Y \subseteq \maximalTightSetContainingHyperedge{p, H, w}(e)\cap Y$. 

    \item Follows from the definition of \merge operation (see \Cref{def:(pHw)-tight-sets}).
    
    \item 
    Let $X \subseteq V$. If at most one of the sets $X\cap e$ or $X\cap f$ is non-empty, then by part (2) of the current lemma (shown above) we have that $b_{(H^M, w^M)}(X) = b_{(H,w)}(X) \geq p(X)$. Suppose that $e\cap X, f\cap X \not = \emptyset$. Then, we have that $\alpha^M \leq \beta^M \leq b_{(H,w)}(X) - p(X)$.
    Thus,  
    $b_{(H^M, w^M)}(X) = b_{(H,w)}(X) - \alpha^M \geq b_{(H,w)}(X) - \beta^M \geq p(X)$, where the equality is by part (2) of the current lemma (shown above).

    \item We note that $w^M(e\cup f) = \alpha^M \in \Z_+$ by parts (1) and (2) of the current lemma (shown above). Furthermore, $b_{H^M,w^M}(X) \geq p(X)$ for every $X\subseteq V$ by part (3) of the current lemma (shown above). Thus, $\beta^R \in \Z_{\geq 0}$ since $p$ is an integer-valued function.

    \item Follows from the definition of \reduce operation (see \Cref{def:(pHw)-tight-sets}).

    \item Let $X\subseteq V$. First, suppose that the hyperedge $e\cup f$ is not contained in the set $X$. 
    Then, we have that $b_{(H^R, w^R)}(X) = b_{(H^M, w^M)}(X) - d_{(H_0, w_0)}(X) \geq p(X) - d_{(H_0, w_0)}(X) = p^R(X)$. Here, the first equality is by part (5) of the current lemma and the inequality is by part (3) of the current lemma (proofs of both parts shown above).  Next, suppose that $e\cup f \subseteq X$. Then, we have that $b_{(H^R, w^R)}(X) = b_{(H^M, w^M)}(X) - \alpha^R \geq b_{H^M,w^M}(X) - \beta^R \geq b_{(H^M, w^M)}(X) - (b_{(H^M, w^M)}(X) - p(X)) = p(X) = p^R(X)$, where the first equality is by part (5) of the current lemma (shown above) and because $e\cup f \subseteq X$, and the final equality is because $e\cup f \subseteq X$.
\end{enumerate}
\end{proof}

With the previous properties established, we now show that the function $p'$ and hypergraph $(H', w')$ constructed by \Cref{alg:WeakToStrongCover}  as input to the subsequent recursive call of the algorithm satisfy the two hypothesis conditions of \Cref{thm:WeakToStrongCover:main}  --  the function $p'$ is symmetric skew-supermodular, and the hypergraph $(H', w')$ weakly covers the function $p'$. This is a necessary step in order to give an inductive proof of \Cref{thm:WeakToStrongCover:main}. 

\begin{lemma}\label{lem:weak-to-strong-cover:next-recursive-call-hypothesis-satification:symm-skew-supermod-strong-cover}
Suppose that the input to \Cref{alg:WeakToStrongCover} is a non-empty hypergraph $(H = (V, E), w)$ and a symmetric skew-supermodular function $p:2^V\rightarrow \Z$ 
such that every hyperedge $\Tilde{e} \in E$ is contained in some $(p, H, w)$-tight set, $b_{(H,w)}(u) \not = 0$ for every $u \in V$, and $b_{(H,w)}(X) \geq p(X)$ for every $X \subseteq V$. 
Let $p'$ and $(H' = (V', E'), w')$ be as defined by \Cref{alg:WeakToStrongCover}. Then, we have the following:
\begin{enumerate}
    \item the function $p'$ is symmetric skew-supermodular, and
    \item $b_{(H', w')}(X) \geq p'(X)$ for every $X \subseteq V'$.
\end{enumerate}
\end{lemma}
\begin{proof} We prove each claim separately below. Let $(H_0, w_0), p^R, \zeros,$ and $(H^R, w^R)$ be as defined by \Cref{alg:WeakToStrongCover}.
\begin{enumerate}
    \item We recall that $p^R =  p - d_{(H_0, w_0)}$. We first observe that the function $p^R$ is symmetric skew-supermodular because the function $p$ is symmetric skew-supermodular and the cut function $d_{(H_0, w_0)}$ is symmetric submodular. Next, we observe that the function $p' = \functionContract{p^R}{\zeros}$. The claim then follows because the contraction operation preserves skew-supermodularity and symmetry.
    \item Let $X \subseteq V'$. Let $R \subseteq \zeros$ be such that $p'(X) = \functionContract{p^R}{\zeros}(X) = p^R(X \cup R)$. Then, we have the following: $b_{(H', w')}(X) = b_{(H^R, w^R)}(X \cup R) \geq p^R(X \cup R) = p'(X)$. Here, the first equality is because $w' = w^R$ and $b_{(H^R, w^R)}(u) = 0$ for every $u \in R$ by the definition of $\zeros$. The inequality is by \Cref{lem:WeakToStrong:properties-of-alpha^M-and-alpha^R}(6).  
\end{enumerate}
\end{proof}

In order to show that the function $p'$ and $(H',w')$ constructed by \Cref{alg:WeakToStrongCover} satisfy conditions \ref{cond:WeakToStrong:(a)} and \ref{cond:WeakToStrong:(b)}, we examine how the families of $(p, H, w)$-tight sets change during a recursive call of the algorithm. 
The next lemma shows that $(p, H, w)$-tight sets persist through the merge and reduce steps of a recursive call. Furthermore, the lemma also shows that the projection of a $(p, H, w)$-tight set onto the ground set of the subsequent recursive call is  $(p', H', w')$-tight.

\begin{lemma}\label{lem:WeakToStrongCover:TightSetFamiliesDuringAlgorithm}
Suppose that the input to \Cref{alg:WeakToStrongCover} is a non-empty hypergraph $(H = (V, E), w)$ and a symmetric skew-supermodular function $p:2^V\rightarrow \Z$ 
such that every hyperedge $\Tilde{e} \in E$ is contained in some $(p, H, w)$-tight set, $b_{(H,w)}(u) \not = 0$ for every $u \in V$, and $b_{(H,w)}(X) \geq p(X)$ for every $X \subseteq V$. 
Let $(H^M, w^M), (H^R, w^R), p'$ 
and $(H' = (V', E'), w')$ be as defined by \Cref{alg:WeakToStrongCover}.
Then, we have the following:
     \begin{enumerate}
        \item $\tightSetFamily{p, H, w} \subseteq \tightSetFamily{p, H^M, w^M} \subseteq \tightSetFamily{p^R, H^R, w^R}$, and 
        \item $\tightSetFamily{p^R, H^R, w^R}|_{V'} \subseteq \tightSetFamily{p', H', w'}$.
    \end{enumerate}
\end{lemma}
\begin{proof} We prove the claims separately below. Let $e, f, \alpha^M, \zeros$ be as defined by \Cref{alg:WeakToStrongCover}.
\begin{enumerate}
    \item We will first show that $\tightSetFamily{p, H, w} \subseteq \tightSetFamily{p, H^M, w^M}$. Let $X \in \tightSetFamily{p, H, w}$. 
    Suppose $e\cap X\neq \emptyset$ and $f\cap X\neq \emptyset$, then 
        $b_{(H^M, w^M)}(X) = b_{(H,w)}(X) - \alpha^M = p(X) - \alpha^M < p(X)$, which is a contradiction. Here, the first equality is because $e\cap X\neq \emptyset$ and $f\cap X\neq \emptyset$ and the final inequality is by \Cref{lem:WeakToStrong:properties-of-alpha^M-and-alpha^R}(1). 
        So, at most one of $e\cap X$ and $f\cap X$ is non-empty and consequently, $b_{(H,w)}(X)=b_{(H^M, w^M)}(X)$. 
        Thus, we have that $p(X) = b_{(H,w)}(X) = b_{(H^M, w^M)}(X)$, i.e. $X \in \tightSetFamily{p, H^M, w^M}$. 

        We will now show the second inclusion in the claim, i.e. $\tightSetFamily{p, H^M, w^M} \subseteq \tightSetFamily{p^R, H^R, w^R}$. Let $X \in \tightSetFamily{p, H^M, w^M} $. First, suppose that $\alpha^R = 0$. Then, we have that $b_{(H^R, w^R)}(X) = b_{(H^M, w^M)}(X) = p(X) = p^R(X)$, where the first and final equalities are because $\alpha^R = 0$. Thus, $X \in \tightSetFamily{p^R, H^R, w^R}$. Next, suppose that $\alpha^R > 0$. We note that if $(e\cup f) \cap X = \emptyset$, then $b_{(H^R,w^R)}(X) = b_{(H^M,w^M)}(X) = p(X) = p^R(X)$ and so $X \in \tightSetFamily{p^R, H^R, w^R}$. Consider the case where $(e\cup f) \cap X \not = \emptyset$. Here, we may assume that 
        $e\cup f\in \delta(X)$ since otherwise, 
        we obtain the following contradiction  $b_{(H^R, w^R)}(X) = b_{(H^M, w^M)}(X) - \alpha^R = p(X) - \alpha^R < p(X) = p^R(X)$, where the final equality is because $d_{(H_0, w_0)}(X) = 0$ since $e\cup f \not \in \delta_H(X)$. 
        Consequently, we have that $b_{(H^R, w^R)}(X) = b_{(H^M, w^M)}(X) - \alpha^R = p(X) - \alpha^R = p^R(X)$, and the claim holds. 
        
        \item Let $X \in \tightSetFamily{p^R, H^R, w^R}$.  Let $\zeros$ be as defined by \Cref{alg:WeakToStrongCover}. Then, we have the following:
        $$p'(X - \zeros) = \functionContract{p^R}{\zeros}(X - \zeros) \geq p^R(X) = b_{(H^R, w^R)}(X) = b_{(H^R, w^R)}(X - \zeros) = b_{(H', w')}(X - \zeros),$$
        where the penultimate equality is because $\zeros = \{u \in V : b_{(H^R, w^R)}(u) = 0\}$. However, by \Cref{lem:weak-to-strong-cover:next-recursive-call-hypothesis-satification:symm-skew-supermod-strong-cover}, we have that $b_{H',w'}(X - \zeros) \geq p'(X - \zeros)$. Thus, the inequality in the above sequence is an equation and we have that $X - \zeros \in \tightSetFamily{p', H', w'}$.
\end{enumerate}
\end{proof}
We now show another useful property of \Cref{alg:WeakToStrongCover} which we will need while reasoning about the convergence of \Cref{alg:WeakToStrongCover} in the next section.
\begin{lemma}\label{lem:WeakToStrong:alpha^R=0-iff-alpha^M=beta^M}
Suppose that the input to \Cref{alg:WeakToStrongCover} is a non-empty hypergraph $(H = (V, E), w)$ and a symmetric skew-supermodular function $p:2^V\rightarrow \Z$
such that every hyperedge $\Tilde{e} \in E$ is contained in some $(p, H, w)$-tight set, $b_{(H,w)}(u) \not = 0$ for every $u \in V$, and $b_{(H,w)}(X) \geq p(X)$ for every $X \subseteq V$. 
 Let $\alpha^M, \alpha^R$ and $\beta^M$ be as defined by \Cref{alg:WeakToStrongCover}.  
 Then, $\alpha^R = 0$ if and only if $\alpha^M = \beta^M$.
\end{lemma}
\begin{proof}
    We prove the claim by showing the implication in both directions. Let $e, f, (H^M, w^M), \beta^R$ be as defined by \Cref{alg:WeakToStrongCover}. 
    First, suppose that $\alpha^M = \beta^M$. Let $X \subseteq V$ be such that $\beta^M = b_{(H,w)}(X) - p(X)$ and $e\cap X, f\cap X \not = \emptyset$ (such a set exists by definition of $\beta^M$). Then, we have that
    $b_{(H^M,w^M)}(X) = b_{(H,w)}(X) - \alpha^M = p(X)$, where the first equality is by \Cref{lem:WeakToStrong:properties-of-alpha^M-and-alpha^R}(2). Thus, $X \in \tightSetFamily{p, H^M, w^M}$. Let $Y \in \maximalTightSetFamily{p, H^M, w^M}$ such that $X \subseteq Y$. \Cref{claim:WeakToStrong:new-tight-set-containing-e-and-f} below says that the set $Y$ contains the hyperedge $e \cup f$. With this claim, we complete the proof of the forward direction as follows: By \Cref{claim:WeakToStrong:new-tight-set-containing-e-and-f}, we have that $\alpha^R \leq \beta^R \leq b_{H^M,w^M}(Y) - p(Y) = 0$, where the final equality is because $Y \in \tightSetFamily{p, H^M, w^M}$. Consequently, $\alpha^R = 0$ because $\alpha^R \geq 0$ by \Cref{lem:WeakToStrong:properties-of-alpha^M-and-alpha^R}(4). 
\begin{claim}\label{claim:WeakToStrong:new-tight-set-containing-e-and-f}
       $e \cup f \subseteq Y$.
    \end{claim}
    \begin{proof}
    By way of contradiction, say $(e\cup f) - Y \not = \emptyset$. Thus, $e - Y \not = \emptyset$ or $f - Y \not = \emptyset$ (or both). We consider the case where $e - Y \not = \emptyset$ and remark that the proof of the case where $f - Y \not = \emptyset$ is along the same lines.
        We recall that the hyperedge $e$ is contained in some $(p, H, w)$-tight set, and thus the set $\maximalTightSetContainingHyperedge{p, H, w}(e)$ is well-defined and non-empty. Furthermore, by \Cref{lem:WeakToStrongCover:TightSetFamiliesDuringAlgorithm}(1), we have that the set $\maximalTightSetContainingHyperedge{p, H, w}(e) \in \tightSetFamily{p, H^M, w^M}$. Consequently, there exists a set $Z \in \maximalTightSetFamily{p, H^M, w^M}$ such that $e \subseteq \maximalTightSetContainingHyperedge{p, H, w}(e) \subseteq Z$. Since $e - Y \not = \emptyset$ but $e \subseteq Z$, we have that $Y \not = Z$.  We recall that $e \cap X \not = \emptyset$ by our choice of the set $X$, and hence, $e \cap Y \not = \emptyset$. Furthermore, by \Cref{lem:WeakToStrong:properties-of-alpha^M-and-alpha^R}(3), we have that $b_{(H^M, w^M)}(A) \geq p(A)$ for every $A \subseteq V$. Thus, by \Cref{lem:weak-to-strong:uncrossing:uncrossing-lemma} applied to the function $p$, hypergraph $(H^M, w^M)$, sets $Y$ and $Z$, and the hyperedge $e$, we have that the set $Y \cup Z \in \maximalTightSetFamily{p, H^M, w^M}$, a contradiction to the maximality of the sets $Y, Z \in \maximalTightSetFamily{p, H^M, w^M}$.
    \end{proof}
    
    Next, we prove the reverse direction of the claim. Suppose that $\alpha^R = 0$. 
    We note that $w^M(e\cup f) \geq \alpha^M > 0$ by Lemma \ref{lem:WeakToStrong:properties-of-alpha^M-and-alpha^R}(2) and \ref{lem:WeakToStrong:properties-of-alpha^M-and-alpha^R}(1). Thus, we have that $\alpha^R < w^M(e\cup f)$ and so $\alpha^R = \beta^R$. It follows that $\beta^R=0$. Consequently, there exists a set $X \subseteq V$ such that $e\cup f \subseteq X$ and $b_{(H^M,w^M)}(X) = p(X)$ (by definition of $\beta^R$). 
     Then, we have the following: $$p(X) = b_{(H^M, w^M)}(X) = b_{(H,w)}(X) - \alpha^M \geq b_{(H,w)}(X) - \beta^M \geq b_{(H,w)}(X) - (b_{(H,w)}(X) - p(X)) = p(X).$$
    Thus, all inequalities are equations, and we have that $\alpha^M = \beta^M$. Here, the second equality is by \Cref{lem:WeakToStrong:properties-of-alpha^M-and-alpha^R}(2) and $e\cup f \subseteq X$. The second inequality is also because $e\cup f \subseteq X$.
\end{proof}

 We now show that if the inputs to \Cref{alg:WeakToStrongCover} satisfy conditions \ref{cond:WeakToStrong:(a)} and \ref{cond:WeakToStrong:(b)}, then the inputs to every subsequent recursive call of \Cref{alg:WeakToStrongCover} also satisfy the two conditions.

\begin{lemma}\label{lem:weak-to-strong-cover:next-recursive-call-hypothesis-satification:hyperedge-in-tight-set-positive-coverage}
Suppose that the input to \Cref{alg:WeakToStrongCover} is a non-empty hypergraph $(H = (V, E), w)$ and a symmetric skew-supermodular function $p:2^V\rightarrow \Z$ 
such that every hyperedge $\Tilde{e} \in E$ is contained in some $(p, H, w)$-tight set, $b_{(H,w)}(u) \not = 0$ for every $u \in V$, and $b_{(H,w)}(X) \geq p(X)$ for every $X \subseteq V$. 
Let $(H' = (V', E'), w')$ and $p'$ be as defined by \Cref{alg:WeakToStrongCover}. Then, we have the following:
\begin{enumerate}
    \item every hyperedge $\Tilde{e} \in E'$ is contained in some $(p', H', w')$-tight set, and
    \item  $b_{H',w'}(u) > 0$ for every vertex $u \in V'$.
\end{enumerate}
\end{lemma}
\begin{proof} We prove the claims separately below. Let $e, f, \zeros, \alpha^R, \beta^R, p', (H^M, w^M), (H^R, w^R)$ be as defined by \Cref{alg:WeakToStrongCover}.
\begin{enumerate}
    \item We consider two cases. First, suppose that $\Tilde{e} \in E \cap E'$. Since the hyperedge $\Tilde{e}$ is contained in some $(p, H, w)$-tight set, we have that $\maximalTightSetContainingHyperedge{p, H, w}(\Tilde{e})$ is well-defined and non-empty. Furthermore, by \Cref{lem:WeakToStrongCover:TightSetFamiliesDuringAlgorithm}(2), we have that $\maximalTightSetContainingHyperedge{p, H, w}(\Tilde{e})|_{\zeros} \in \tightSetFamily{p', H', w'}$. Consequently, $\Tilde{e} \subseteq \maximalTightSetContainingHyperedge{p, H, w}(\Tilde{e})|_{\zeros} \in \tightSetFamily{p', H', w'}$. Next, suppose that $\Tilde{e} \in E' - E$. In particular, this implies that $E' - E \not= \emptyset$. 
    We recall that $e, f \in E$ are the hyperedges chosen by  \Cref{alg:WeakToStrongCover}  \Cref{alg:WeakToStrongCover:(4)}. 
    By Lemma \ref{lem:WeakToStrong:properties-of-alpha^M-and-alpha^R}(5) and \ref{lem:WeakToStrong:properties-of-alpha^M-and-alpha^R}(2), we have that $\Tilde{e} = e\cup f$ and $\alpha^R = \beta^R < w^M(e \cup f)$. Consequently, $\zeros = \emptyset$ and $V' = V$. Let $Z \subseteq V$ be a set such that 
    $\alpha^R = b_{(H^M, w^M)}(Z) - p(Z)$ and $e\cup f \subseteq Z$ (such a set exists by definition of $\beta^R$).
    Then, we have the following: 
    $$b_{H',w'}(Z) = b_{(H^R, w^R)}(Z) = b_{(H^M, w^M)}(Z) - \alpha^R = p(Z) = p^R(Z) = p'(Z),$$
    and consequently the set $Z \in \tightSetFamily{p', H', w'}$.
    Here, the first equality is by definitions of the weight function $w^R$ and the set $\zeros$. The second and the fourth equalities are because $e \cup f \subseteq Z$. The third equality is because the set $Z$ is a witness set for $\alpha^R = \beta^R$. 
    The fifth equality is because $\zeros = \emptyset$.  

    \item We recall that $V' = V - \zeros = V - \{u\in V : b_{(H^R, w^R)}(u) = 0\}$. The claim then follows because $w' = w^R$. 
    
\end{enumerate}
\end{proof}

We now show the first main result of the section which says that \Cref{alg:WeakToStrongCover} terminates, and the hypergraph returned by the algorithm indeed strongly covers the input symmetric skew-supermodular function, i.e. satisfies property \ref{thm:WeakToStrongCover:main:(1)} of \Cref{thm:WeakToStrongCover:main}.

\begin{lemma}\label{lem:weak-to-strong:correctness:main:strong-cover}
    Suppose that the input to \Cref{alg:WeakToStrongCover} is a  hypergraph $(H = (V, E), w)$ and a symmetric skew-supermodular function $p:2^V\rightarrow \Z$ 
such that every hyperedge $\Tilde{e} \in E$ is contained in some $(p, H, w)$-tight set, $b_{(H,w)}(u) \not = 0$ for every $u \in V$, and $b_{(H,w)}(X) \geq p(X)$ for every $X \subseteq V$.  Then, \Cref{alg:WeakToStrongCover} terminates in finite number of recursive calls. Furthermore, let $(H^*, w^*)$ be the hypergraph returned by the algorithm. Then, $d_{(H^*, w^*)}(X) \geq p(X)$ for every $X\subseteq V$.
\end{lemma}
\begin{proof}
By Lemmas \ref{lem:weak-to-strong-cover:next-recursive-call-hypothesis-satification:symm-skew-supermod-strong-cover} and \ref{lem:weak-to-strong-cover:next-recursive-call-hypothesis-satification:hyperedge-in-tight-set-positive-coverage}, the input to every recursive call (except the last one) of the execution satisfies the hypothesis of \Cref{cor:weak-to-strong:uncrossing:exist-two-hyperedges-contained-in-disjoint-maximal-tight-sets}, and consequently, during every recursive call of the execution (except the last one), there exist two hyperedges such that \Cref{alg:WeakToStrongCover} \Cref{alg:WeakToStrongCover:(4)} is well-defined.
We first show that 
algorithm terminates in a finite number of recursive calls. 
Consider the potential function $\phi(H=(V, E), w) \coloneqq  \sum_{e \in E}w(e)$. We note that $\phi(H, w) \in \Z_{\geq 0}$ since the weight function $w:E\rightarrow\Z_+$ is a positive integral function. Furthermore, if $\phi(H, w) = 0$, then we have that $E = \emptyset$ and the algorithm terminates. We recall that $p'$ and $(H', w')$ denote the inputs to the subsequent recursive call of \Cref{alg:WeakToStrongCover}. 
By Lemmas \ref{lem:weak-to-strong-cover:next-recursive-call-hypothesis-satification:symm-skew-supermod-strong-cover} and \ref{lem:weak-to-strong-cover:next-recursive-call-hypothesis-satification:hyperedge-in-tight-set-positive-coverage}, we have that $p'$ is a symmetric skew-supermodular function and $(H' = (V', E'), w')$ is a non-empty hypergraph such that every hyperedge $\Tilde{e} \in E'$ is contained in some $(p', H', w')$-tight set, $b_{(H', w')}(u) \not = 0$ for every $u \in V'$, and $b_{(H', w')}(X) \geq p'(X)$ for every $X \subseteq V'$.
Then, by Lemma \ref{lem:WeakToStrong:properties-of-alpha^M-and-alpha^R}(1), \ref{lem:WeakToStrong:properties-of-alpha^M-and-alpha^R}(2) and \ref{lem:WeakToStrong:properties-of-alpha^M-and-alpha^R}(5), we have that $\phi(H, w) > \phi(H', w')$, i.e. the value of the potential function strictly decreases between subsequent recursive calls. Consequently, the algorithm terminates in a finite number of recursive calls.

We now show that the hypergraph $(H^*,w^*)$ returned by \Cref{alg:WeakToStrongCover} strongly covers the function $p$, i.e. $d_{(H^*, w^*)}(X) \geq p(X)$ for every $X\subseteq V$. We show this by induction on $w(E)$. For the base case, consider $w(E) = 0$. Since the weight function $w$ is positive, we have that $E = \emptyset$.
Let $X \subseteq V$. We have that $d_{(H, w)}(X) = 0 = b_{(H,w)}(X) \geq p(X)$ and so the claim holds. For the inductive case, consider $w(E) > 0$. Let $X\subseteq V$ be an arbitrary subset. 
Let $(H_0^* = (V', E_0^*), w_0^*)$ be the hypergraph returned by the subsequent recursive call on input hypergraph $(H' = (V', E'), w')$ and function $p'$. 
We recall that $w'(E') < w(E)$ and so by the inductive hypothesis, we have that $d_{(H_0^*, w_0^*)}(X) \geq p'(X)$ for every $X\subseteq V'$. We also recall that the hypergraph $(H^*, w^*) = (G+H_0, c + w_0)$, where the hypergraph $(G, c)$ is obtained from $(H_0^*, w_0^*)$ by adding the vertices $\zeros$. Then, we have the following:
\begin{align*}
    d_{(H^*, w^*)}(X)& = d_{G,c}(X) + d_{(H_0, w_0)}(X)&\\
    & = d_{H_0^*, w_0^*}(X) + \alpha^R\indicator_{e\cup f\in \delta(X)}&\\
    &\geq p'(X) + \alpha^R\indicator_{e\cup f\in \delta(X)}& \text{(by inductive hypothesis)}\\
    & = \functionContract{\left(p - \alpha^R\cdot \indicator_{e\cup f\in \delta(X)}\right)}{\zeros}(X) + \alpha^R\indicator_{e\cup f\in \delta(X)} &\\
    &\geq \left(p - \alpha^R\cdot  \indicator_{e\cup f\in \delta(X)}\right)(X) + \alpha^R\indicator_{e\cup f\in \delta(X)}&\\
    &= p(X).&
\end{align*}
\end{proof}

We now show that the hypergraph returned by \Cref{alg:WeakToStrongCover} is obtained by merging hyperedges of the input hypergraph (see \Cref{def:WeakToStrongCover:hyperedge-merge} for the definition of merging hyperedges). 
This will imply that the hypergraph returned by the algorithm satisfies property \ref{thm:WeakToStrongCover:main:(2)} of \Cref{thm:WeakToStrongCover:main}. To show this property, we will need the following lemma about the choice of hyperedges $e$ and $f$ being merged by the algorithm.

\begin{lemma}\label{lem:WeakToStrongCover:properties-of-e-and-f-chosen-during-alg}
    Suppose that the input \Cref{alg:WeakToStrongCover} call is a non-empty hypergraph $(H = (V, E), w)$ and a symmetric skew-supermodular function $p:2^V\rightarrow \Z$ 
such that every hyperedge $\Tilde{e} \in E$ is contained in some $(p, H, w)$-tight set, $b_{(H,w)}(u) \not = 0$ for every $u \in V$, and $b_{(H,w)}(X) \geq p(X)$ for every $X \subseteq V$.  Let $e, f$ be as defined by \Cref{alg:WeakToStrongCover}. Then,
\begin{enumerate}
    \item $e \cap f = \emptyset$, and
    \item $e\cup f \not \in E$.
\end{enumerate}
\end{lemma}
\begin{proof} We prove the claims separately below.
    \begin{enumerate}
        \item By the choice of hyperedges $e$ and $f$, we have that $\maximalTightSetContainingHyperedge{p, H, w}(e) \not = \maximalTightSetContainingHyperedge{p, H, w}(f)$. We observe that by maximality of the sets $\maximalTightSetContainingHyperedge{p, H, w}(e)$ and $\maximalTightSetContainingHyperedge{p, H, w}(f)$, we have that $\maximalTightSetContainingHyperedge{p, H, w}(e), \maximalTightSetContainingHyperedge{p, H, w}(f) \in \maximalTightSetFamily{p, H, w}$. By \Cref{lem:weak-to-strong:uncrossing:maximal-tight-set-family-disjoint}, we have that $\maximalTightSetContainingHyperedge{p, H, w}(e) \cap  \maximalTightSetContainingHyperedge{p, H, w}(f) = \emptyset$. Consequently, $e \cap f = \emptyset$.
        \item By way of contradiction, suppose that $e\cup f \in E$. Thus, $\maximalTightSetContainingHyperedge{p, H, w}(e\cup f) \not = \emptyset$. Since $e, f \subseteq \maximalTightSetContainingHyperedge{p, H, w}(e\cup f)$, we have that $\maximalTightSetContainingHyperedge{p, H, w}(e) = \maximalTightSetContainingHyperedge{p, H, w}(f) = \maximalTightSetContainingHyperedge{p, H, w}(e\cup f)$ by the maximality of the three sets, contradicting the choice of hyperedges $e, f$ made by the algorithm.
    \end{enumerate}
\end{proof}

We now conclude the section by showing that the hypergraph returned by \Cref{alg:WeakToStrongCover} indeed satisfies property \ref{thm:WeakToStrongCover:main:(2)} of \Cref{thm:WeakToStrongCover:main}.

\begin{lemma}\label{lem:weak-to-strong:correctness:main:hyperedge-merge}
    Suppose that the input to \Cref{alg:WeakToStrongCover} is a hypergraph $(H = (V, E), w)$ and a symmetric skew-supermodular function $p:2^V\rightarrow \Z$ 
such that every hyperedge $\Tilde{e} \in E$ is contained in some $(p, H, w)$-tight set, $b_{(H,w)}(u) \not = 0$ for every $u \in V$, and $b_{(H,w)}(X) \geq p(X)$ for every $X \subseteq V$. 
    Then, the hypergraph $(H^* = (V, E^*), w^*)$ returned by the algorithm is obtained by merging hyperedges of the hypergraph $(H, w)$.
\end{lemma}
\begin{proof}
     We show the claim by induction on $w(E)$. For the base case, consider $w(E) = 0$. Since the weight function $w$ is positive, we have that $E = \emptyset$ and hence $(H, w)$ is the empty hypergraph. In this case, the hypergraph $(H^*, w^*)$ returned by the algorithm is also the empty hypergraph and so the claim holds. 

    For the inductive case, suppose that $w(E) > 0$. Let $(H^M, w^M), \alpha^M$, $(H' = (V', E'), w'), p',(H^R, w^R)$, $(H_0 = (V, E_0), w_0), \zeros, (H^*_0, w^*_0), (G, c)$ be as defined by \Cref{alg:WeakToStrongCover}. By Lemmas \ref{lem:weak-to-strong-cover:next-recursive-call-hypothesis-satification:symm-skew-supermod-strong-cover} and \ref{lem:weak-to-strong-cover:next-recursive-call-hypothesis-satification:hyperedge-in-tight-set-positive-coverage}, the inputs $(H',w')$ and $p'$ satisfy the hypothesis of \Cref{lem:weak-to-strong:correctness:main:strong-cover}. 
    Thus, the execution of the algorithm from this recursive call terminates in a finite number of recursive calls and returns the hypergraph $(H_0^* = (V', E_0^*), w_0^*)$. 
    Then, by Claims \ref{claim:WeakToStrong:MergeProperty:H^*w^*-from-H^Mw^M} and \ref{claim:WeakToStrong:MergeProperty:H^Mw^M-from-Hw} below, we have that the hypergraph $(H^*, w^*)$ is obtained by merging hyperedges of $(H, w)$.
    \end{proof}
    \begin{claim}\label{claim:WeakToStrong:MergeProperty:H^*w^*-from-H^Mw^M}
        $(H^*, w^*)$ is obtained by merging hyperedges of $(H^M, w^M)$.
    \end{claim}
    \begin{proof}
        We first show that the hypergraph $(G, c)$ is obtained by merging hyperedges of $(H^R, w^R)$.
         By Lemma \ref{lem:WeakToStrong:properties-of-alpha^M-and-alpha^R}(1), (2), and (5), we have that $w'(E') < w(E)$. Thus, by the inductive hypothesis, we have that the hypergraph $(H_0^*, w_0^*)$ is obtained by merging hyperedges of $(H', w')$. Consequently, $(G, c)$ is formed by merging hyperedges of $(H^R, w^R)$ because of the following two observations: $(H', w')$ is obtained from $(H^R, w^R)$ by removing the vertices $\zeros$ and $(G, c)$ is obtained by adding vertices $\zeros$ to $(H_0^*, w_0^*)$.
        
        Next, we recall that $(H^*, w^*) = (G + H_0, c + w_0)$. Thus, by the conclusion of the previous paragraph, we have that $(H^*, w^*)$ is obtained by merging hyperedges of $(H^R + H_0, w^R + w_0)$. The claim then follows because $(H^R + H_0, w^R + w_0) = (H^M, w^M)$ by \Cref{lem:WeakToStrong:properties-of-alpha^M-and-alpha^R}(5).
    \end{proof}

    \begin{claim}\label{claim:WeakToStrong:MergeProperty:H^Mw^M-from-Hw}
        $(H^M,w^M)$ is obtained by merging hyperedges of $(H, w)$.
    \end{claim}
    \begin{proof}
        We recall that $(H^M, w^M) = \merge((H,w), e, f, \alpha^M)$ (see \Cref{alg:WeakToStrongCover} \Cref{alg:WeakToStrongCover:(6)}).
        Thus, by definition of the \merge operation, the hypergraph $(H^M, w^M)$ is obtained by decreasing the weight of hyperedges $e$ and $f$ by $\alpha^M$ and increasing the weight of the hyperedge $e\cup f$ by $\alpha^M$. Consequently, the multi-hypergraph $H^M_{w^M}$ can be obtained from the multi-hypergraph $H_w$ by repeatedly merging hyperedges $e$ and $f$ exactly $\alpha^M$ times. We note that the hyperedges $e$ and $f$ are disjoint by \Cref{lem:WeakToStrongCover:properties-of-e-and-f-chosen-during-alg}(1). Thus, the hypergraph $(H^M,w^M)$ is obtained by merging hyperedges of $(H, w)$ by definition. 
    \end{proof}

\subsection{Number of Recursive Calls and Hypergraph Support Size}\label{sec:WeakToStrongCover:NumberOfHyperedges}

We first set up some notation that will be used in the remainder of this section. 
By \Cref{lem:weak-to-strong:correctness:main:strong-cover}, the number of recursive calls made by \Cref{alg:WeakToStrongCover} is finite.
We will use $\ell$ to denote the depth of recursion of the algorithm. We will refer to the recursive call at depth $i$ as \emph{recursive call $i$} or the \emph{$i^{th}$ recursive call}. Let $i \in [\ell]$ be a recursive call. We let the hypergraph $(H_i = (V_i, E_i), w_i:E_i \rightarrow \Z_+)$ and the function $p_i : 2^{V_i} \rightarrow \Z$ denote the inputs to the $i^{th}$ recursive call. Furthermore, we let $e_i, f_i, \beta_i^M, \alpha_i^M, (H_i^M, w_i^M), \beta_i^R, \alpha_i^R, (H_i^R, w_i^R), p_i^R, (H_0^{i}, w_0^{i}), E_0^i, \zeros_i, V_i', p'_i$ denote the quantities $e, f, \beta^M, \alpha^M, (H^M, w^M), \beta^R, \alpha^R, (H^R, w^R), p^R, (H_0, w_0), E_0, \zeros, V', p'$ as defined by \Cref{alg:WeakToStrongCover} during the $i^{th}$ recursive call.

We now define certain set families of that will be crucial in our analysis. We let $\tightSetFamily{i}, \tightSetFamilyAfterMerge{i}, \tightSetFamilyAfterReduce{i}$ denote the families $\tightSetFamily{p_i, H_i, w_i}, \tightSetFamily{p_i, H_i^M, w_i^M}$ and $ \tightSetFamily{p_i^R, H_i^R, w_i^R}$ respectively, and $\maximalTightSetFamily{i}, \maximalTightSetFamilyAfterMerge{i}, \maximalTightSetFamilyAfterReduce{i}$ denote the families $\maximalTightSetFamily{p_i, H_i, w_i}, \maximalTightSetFamily{p_i, H_i^M, w_i^M}$ and $ \maximalTightSetFamily{p_i^R, H_i^R, w_i^R}$ respectively. We define the family of \emph{cumulative projected maximal $(p_i, H_i, w_i)$-tight sets} as:
$$\cumulativeMaximalTightSetFamily{i} \coloneqq  \begin{cases}
    \maximalTightSetFamily{1}& \text{ if $i = 1$,}\\
    
   \cumulativeMaximalTightSetFamily{i-1}|_{V_i} \cup \maximalTightSetFamily{i}& \text{ if $i \geq 2$}.
\end{cases}$$
The next lemma states two important properties of these set families.
\begin{lemma}\label{lem:WeakToStrong:properties-of-cumulative-tight-set-families}
Suppose that the input to \Cref{alg:WeakToStrongCover} is a non-empty hypergraph $(H = (V, E), w)$ and a symmetric skew-supermodular function $p:2^V\rightarrow \Z$ 
such that every hyperedge $\Tilde{e} \in E$ is contained in some $(p, H, w)$-tight set, $b_{(H,w)}(u) \not = 0$ for every $u \in V$, and $b_{(H,w)}(X) \geq p(X)$ for every $X \subseteq V$. Furthermore, suppose that the execution of \Cref{alg:WeakToStrongCover} on the input instance terminates using $\ell \in \Z_+$ recursive calls. Then, for every $i \in [\ell - 1]$, we have the following:
     \begin{enumerate}
         \item $\projCumulativeMaximalTightSetFamily{i}{i+1} \subseteq \tightSetFamily{i+1}$, and
         \item the family $\cumulativeMaximalTightSetFamily{i}$ is laminar.
     \end{enumerate}
\end{lemma}
\begin{proof} 
Let $i\in [\ell - 1]$ be a recursive call. We recall that $p_i:2^{V_i}\rightarrow\Z$ and $(H_i, w_i)$ are inputs to the $i^{th}$ recursive call, where $p_1 := p$ and $(H_1, w_1) := (H, w)$.
By Lemmas \ref{lem:weak-to-strong-cover:next-recursive-call-hypothesis-satification:symm-skew-supermod-strong-cover} and \ref{lem:weak-to-strong-cover:next-recursive-call-hypothesis-satification:hyperedge-in-tight-set-positive-coverage} and induction on $i$, we have that $p_i$ is a symmetric skew-supermodular function and $(H_i = (V_i, E_i), w_i)$ is a non-empty hypergraph such that every hyperedge $\Tilde{e} \in E_i$ is contained in some $(p_i, H_i, w_i)$-tight set, $b_{(H_i, w_i)}(u) \not = 0$ for every $u \in V_i$, and $b_{(H_i, w_i)}(X) \geq p_i(X)$ for every $X \subseteq V_i$.
We now prove both claims separately below.
    \begin{enumerate}
        \item By way of contradiction, let $i \in [\ell - 1]$ be the least index such that $\cumulativeMaximalTightSetFamily{i}|_{V_{i+1}} - T_{i+1} \not = \emptyset$. We note that $i \geq 2$ since $\cumulativeMaximalTightSetFamily{1} = \maximalTightSetFamily{1}\subseteq T_1$ and $T_1|_{V_2}\subseteq T_2$ by Lemma  \ref{lem:WeakToStrongCover:TightSetFamiliesDuringAlgorithm}(1) and \ref{lem:WeakToStrongCover:TightSetFamiliesDuringAlgorithm}(2), and hence, $\cumulativeMaximalTightSetFamily{1}|_{V_2}\subseteq T_2$. Thus, by our choice of index $i$, we have that $\cumulativeMaximalTightSetFamily{i-1}|_{V_i} \subseteq \tightSetFamily{i}$, and so it follows that $\cumulativeMaximalTightSetFamily{i} = \cumulativeMaximalTightSetFamily{i-1}|_{V_i} \cup \maximalTightSetFamily{i} \subseteq \tightSetFamily{i}$. Moreover, $T_i|_{V_{i+1}}\subseteq T_{i+1}$ by Lemma  \ref{lem:WeakToStrongCover:TightSetFamiliesDuringAlgorithm}(1) and \ref{lem:WeakToStrongCover:TightSetFamiliesDuringAlgorithm}(2). Consequently we have that $\cumulativeMaximalTightSetFamily{i}|_{V_{i+1}} \subseteq \tightSetFamily{i}|_{V_{i+1}} \subseteq \tightSetFamily{i+1}$, a contradiction to the choice of the index $i$. 
        \item By way of contradiction, let $i \in [\ell - 1]$ be the least index such that the family $\cumulativeMaximalTightSetFamily{i}$  is not laminar. 
        We note that $i \geq 2$ since the family $\cumulativeMaximalTightSetFamily{1} = \maximalTightSetFamily{1}$ is disjoint by \Cref{lem:weak-to-strong:uncrossing:maximal-tight-set-family-disjoint}. Thus, we have that $\cumulativeMaximalTightSetFamily{i} = \projCumulativeMaximalTightSetFamily{i-1}{i} \cup \maximalTightSetFamily{i}$.
        We note that the family $\maximalTightSetFamily{i}$ is also disjoint by \Cref{lem:weak-to-strong:uncrossing:maximal-tight-set-family-disjoint}. Furthermore, the family $\cumulativeMaximalTightSetFamily{i-1}$ (and consequently, the family  $\cumulativeMaximalTightSetFamily{i-1}|_{V_i}$) is laminar by our choice of $i\in[\ell - 1]$. Since the family $\cumulativeMaximalTightSetFamily{i}$ is not laminar, there exist distinct sets $X \in \projCumulativeMaximalTightSetFamily{i-1}{i}$ and $Y \in \maximalTightSetFamily{i}$ such that $X - Y, Y - X, X \cap Y \not = \emptyset$. By part (1) of the current lemma (shown above), we have that $X \in \tightSetFamily{i}$. Consequently, there exists $Z\in \maximalTightSetFamily{i}$ such that $X \subseteq Z$. We note that $Y \not = Z$ because $X -Z=\emptyset$ and $X - Y \not = \emptyset$. Thus, by \Cref{lem:weak-to-strong:uncrossing:maximal-tight-set-family-disjoint} we have that $Z\cap Y = \emptyset$. Therefore, $X\cap Y = \emptyset$, contradicting our choice of sets $X$ and $Y$.

    \end{enumerate} 
\end{proof}

\begin{lemma}\label{lem:WeakToStrongCover:NumberOfHyperedges}
Suppose that the input to \Cref{alg:WeakToStrongCover} is a non-empty hypergraph $(H = (V, E), w)$ and a symmetric skew-supermodular function $p:2^V\rightarrow \Z$ 
such that every hyperedge $\Tilde{e} \in E$ is contained in some $(p, H, w)$-tight set, $b_{(H,w)}(u) \not = 0$ for every $u \in V$, and $b_{(H,w)}(X) \geq p(X)$ for every $X \subseteq V$. 
Then, we have the following:
    \begin{enumerate}
        \item the recursion depth $\ell \in \Z_+$ of the algorithm is at most $|E| + 10|V| - 1$, and
        \item $|E^*| \leq |E| + 10|V| - 2$, where $(H^* = (V, E^*), w^*)$ denotes the hypergraph returned by the algorithm.
    \end{enumerate}
\end{lemma}

\begin{proof}
Let $i \in [\ell - 1]$ be a recursive call. We recall that $(H_i = (V_i, E_i), w_i:E_i \rightarrow \Z_+)$ and $p_i : 2^{V_i} \rightarrow \Z$ denote the inputs to the $i^{th}$ recursive call, and $(H_i^* = (V, E_i^*), w_i^*)$ denotes the hypergraph returned by the $i^{th}$ recursive call. Furthermore, $e_i, f_i, \beta_i^M, \alpha_i^M, (H_i^M, w_i^M), \beta_i^R, \alpha_i^R, (H_i^R, w_i^R)$, $p_i^R$, $(H_0^{i}, w_0^{i})$, $E_0^i, \zeros_i$ denote the quantities  as defined by \Cref{alg:WeakToStrongCover} during the $i^{th}$ recursive call.
    We prove the claims separately below.
    \begin{enumerate}
        \item We have the following:
        $$\ell \leq \left|\left\{i \in [\ell - 1] : \alpha_i^R = \beta_i^R\right\}\right|  +  \left|\left\{i \in [\ell- 1] : \alpha_i^R = w_i^M(e_i \cup f_i)\right\}\right| + 1 \leq |E| + 10|V| - 1,$$
        where the final inequality is by \Cref{claim:WeakToStrong:hyperedge-support-size:alpha^R=weight-of-added-hyperedge} and \Cref{claim:WeakToStrong:hyperedge-support-size:alpha^R<weight-of-added-hyperedge} below.
        \item 
        We have the following:
        $$|E_1^*| \leq |E_2^*| + |E_0^2| =   |E_{2}^*| + \indicator_{\alpha_1^R > 0} = |E_{\ell}^*| + \sum_{i \in [\ell - 1]}\indicator_{\alpha_i^R > 0} \leq \ell - 1\leq |E_1| + 10|V_1| - 2.$$
        Here, the second inequality is because $E_{\ell}^* = \emptyset$ by definition of the \Cref{alg:WeakToStrongCover} and the final inequality is by part (1) of the claim (shown above).
    \end{enumerate}
\end{proof}

\begin{claim}\label{claim:WeakToStrong:hyperedge-support-size:alpha^R=weight-of-added-hyperedge}
$\left|\left\{i \in [\ell-1] : \alpha_i^R = \beta_i^R\right\}\right| \leq 5|V| - 1$.
\end{claim}
\begin{proof}
We define a potential function $\phi:[\ell - 1] \rightarrow\Z_{\geq 0}$ as follows: for every $i \in [\ell]$,
    $$\phi(i) \coloneqq  \left|\cumulativeMaximalTightSetFamily{i}\right| + 3|\cumulativeZeros{i-1}|,$$
    where we define the set $\cumulativeZeros{0} \coloneqq  \emptyset$.
We first show that the function $\phi$ is monotone non-decreasing.
Let $i \in [\ell -1]$ be a recursive call of \Cref{alg:WeakToStrongCover}. If $\zeros_i = \emptyset$, then $|\cumulativeMaximalTightSetFamily{i}| \leq |\cumulativeMaximalTightSetFamily{i+1}|$ since $\cumulativeMaximalTightSetFamily{i+1} = \cumulativeMaximalTightSetFamily{i} |_{V_{i+1}}\cup \maximalTightSetFamily{i+1} = \cumulativeMaximalTightSetFamily{i}  \cup \maximalTightSetFamily{i+1},$   where the second equality is because $V_{i+1} = V_i - \zeros_i = V_i$. Then, we have that $\phi(i+1) \geq \phi(i)$ since $\cumulativeZeros{i-1} \subseteq \cumulativeZeros{i}$ by definition. Thus, we suppose that $\zeros_i \not = \emptyset$. We have the following: 
\begin{align*}
    \phi(i+1) - \phi(i)& = |\cumulativeMaximalTightSetFamily{i+1}| + 3 |\cumulativeZeros{i}| - |\cumulativeMaximalTightSetFamily{i}| - 3|\cumulativeZeros{i-1}|&\\
    & = \left|\cumulativeMaximalTightSetFamily{i}|_{V_{i+1}}\cup \maximalTightSetFamily{i+1}  \right| - |\cumulativeMaximalTightSetFamily{i}| + 3|\zeros_{i}|&\\
    &\geq \left|\cumulativeMaximalTightSetFamily{i}|_{V_{i+1}}\right| - |\cumulativeMaximalTightSetFamily{i}| + 3|\zeros_{i}|&\\
    & \geq  0.&
\end{align*}
Here, the final inequality is because of the following: we recall that the family $\cumulativeMaximalTightSetFamily{i}$ is laminar by \Cref{lem:WeakToStrong:properties-of-cumulative-tight-set-families}. Furthermore, the family $\projCumulativeMaximalTightSetFamily{i}{i+1}$ is the projection of the laminar family $\cumulativeMaximalTightSetFamily{i}$ to the ground set $V_{i+1} = V_i - \zeros_i$. The inequality then follows by \Cref{lem:CoveringAlgorithm:projection-laminar-family}  --  here, we observe that  \Cref{lem:CoveringAlgorithm:projection-laminar-family} says that this inequality is strict because $\zeros_i \not = \emptyset$. We will refer back to this observation in the next part of the proof.

Consider a recursive call $i \in \left\{i \in [\ell - 1] : \alpha_i^R = \beta_i^R\right\}$. We now show that $\phi(i+1) > \phi(i)$, i.e. the potential value strictly increases. 
We suppose that $\zeros_i = \emptyset$ as otherwise $\phi(i+1) > \phi(i)$ by the observation in the final line of the previous paragraph.
Thus, $p_{i+1} = p_i^R$ and $(H', w') = (H^R, w^R)$. 
We note that the hypergraph $(H_i, w_i)$ is non-empty since $i < \ell$.
We consider a set $X\subseteq V_i$ such that $e_i\cup f_i\subseteq X\subseteq V_i$ and $b_i^M(X)-p_i(X)=\beta_i^R$ (such a set exists by the definition of $\beta_i^R$). 
Then, we have the following: $b_{i+1}(X) = b_i^R(X)=b_i^M(X)-\alpha_i^R=p_i(X)=p_i^R(X) = p_{i+1}(X)$, where the first equality is because $(H', w') = (H^R, w^R)$, the second equality is because $\zeros_i = \emptyset$, the third equality is because $\alpha_i = \beta_i$ by our choice of $i \in [\ell - 1]$, and the final equality is because $p_{i+1} = p_{i}^R$. Thus, $X \in \tightSetFamily{i+1}$. Consequently, 
there exists a set $Y \in \maximalTightSetFamily{i+1}$ such that $e_i\cup f_i \subseteq X \subseteq Y$. We note that $Y \not \in \tightSetFamily{i}$ due to the following: suppose (for sake of contradiction) that $Y \in \tightSetFamily{i}$, then 
$b_i(Y)=p_i(Y)$ and $b_i^M(Y)=b_i(Y)-\alpha_i^M=p_i(Y)-\alpha_i^M<p_i(Y)$, a contradiction to \Cref{lem:WeakToStrong:properties-of-alpha^M-and-alpha^R}(3) -- the first equality is by definition of $b_i^M$ and $e_i\cup f_i\subseteq Y$ and the inequality is because $\alpha_i^M\ge 1$ by \Cref{lem:WeakToStrong:properties-of-alpha^M-and-alpha^R}(1). Thus, $Y \in \maximalTightSetFamily{i+1} - \tightSetFamily{i} \subseteq \maximalTightSetFamily{i+1} - \cumulativeMaximalTightSetFamily{i}$ by \Cref{lem:WeakToStrong:properties-of-cumulative-tight-set-families}(1).
Consequently, we have that
$\cumulativeMaximalTightSetFamily{i} \subsetneq \cumulativeMaximalTightSetFamily{i} \cup \maximalTightSetFamily{i+1} = \cumulativeMaximalTightSetFamily{i+1}.$

We now show the claim. By the previous two properties of our potential function $\phi$, we have the following:
$$\left|\left\{i \in [\ell - 1] : \alpha_i^R = \beta_i^R\right\}\right| \leq \phi(\ell - 1) - \phi(1) \leq |\cumulativeMaximalTightSetFamily{\ell - 1}| + 3|\cumulativeZeros{\ell - 1}| \leq 5|V| - 1,$$
where the final inequality is because the family $\cumulativeMaximalTightSetFamily{\ell - 1}$ is laminar by \Cref{lem:WeakToStrong:properties-of-cumulative-tight-set-families}(2).
\end{proof}


\begin{claim}\label{claim:WeakToStrong:hyperedge-support-size:alpha^R<weight-of-added-hyperedge}
    $\left|\left\{i \in [\ell - 1] : \alpha_i^R = w_i^M(e_i \cup f_i)\right\}\right| \leq |E| + 5|V| - 1$. 
\end{claim}
\begin{proof}
    We define a potential function $\phi:[\ell - 1] \rightarrow \Z_{\geq 0}$ as follows: for every $i \in [\ell]$,
    $$\phi(i) \coloneqq  -|E_i| - \left|\left\{j \in [i,\ell - 1] : \alpha_j^R < w_j^M(e_j\cup f_j)\right\}\right|.$$
    
    First, consider a recursive call $i \in [\ell - 1]$ such that $\alpha_i^R = \beta_i^R < w_i^M(e_i \cup f_i)$. We will show that the potential function $\phi$ does not decrease. By Lemma \ref{lem:WeakToStrong:properties-of-alpha^M-and-alpha^R}(2) and \ref{lem:WeakToStrong:properties-of-alpha^M-and-alpha^R}(5), we have that 
    \begin{align*}
       E_{i+1} & = E_i - \{e_i\}\cdot \indicator_{\alpha_i^M = w_i(e_i)} - \{f_i\}\cdot \indicator_{\alpha_i^M = w_i(f_i)} + \{e_i\cup f_i\}\cdot \indicator_{\alpha_i^R = \beta_i^R < w_i^M(e_i \cup f_i)}&\\
        &= E_i - \{e_i\}\cdot \indicator_{\alpha_i^M = w_i(e_i)} - \{f_i\}\cdot \indicator_{\alpha_i^M = w_i(f_i)} + \{e_i\cup f_i\}.&
    \end{align*}
Consequently, 
    we have that $|E_{i+1}| \leq |E_i| + 1$. Furthermore, by the choice of recursive call $i$, we have that $\left|\{j \in [i+1, \ell - 1] : \alpha_j^R < w_j^M(e_j\cup f_j)\}\right| = \left|\{j \in [i, \ell - 1] : \alpha_j^R < w_j^M(e_j\cup f_j)\}\right| - 1$. Thus, $\phi(i+1) \geq \phi(i)$.
    
     Next, consider a recursive call $i \in [\ell - 1]$ such that $\alpha_i^R = w_i^M(e_i \cup f_i)$. We will show that the potential function $\phi$ strictly increases. We note that the hypergraph $(H_i, w_i)$ is non-empty since $i < \ell$.  
     By \Cref{lem:WeakToStrong:properties-of-alpha^M-and-alpha^R}(1), we have that $\alpha_i^R = w_i^M(e_i \cup f_i) = \alpha_i^M \geq 1$. 
     Consequently, by  \Cref{lem:WeakToStrong:alpha^R=0-iff-alpha^M=beta^M}, we have that $\alpha_i^M  < \beta_i^M$ and hence, $\alpha_i^M \in \{w_i(e_i), w_i(f_i)\}$ by definition. 
     Suppose that $\alpha_i^M = w_i(e_i)$. 
     Then, we have that $w^R_i(e_i) = w^M_i(e_i) = w_i(e_i) - \alpha_{i}^M = 0$. 
     Thus, we have that $e_i \not \in E_{i+1}$. 
     Then, by our choice of $i$ and by Lemma \ref{lem:WeakToStrong:properties-of-alpha^M-and-alpha^R}(2) and \ref{lem:WeakToStrong:properties-of-alpha^M-and-alpha^R}(5), we have that $E_{i+1} \subsetneq E_i$. 
     We note that this strict inclusion can also be obtained in the case where $\alpha_i^M = w_i(f_i)$ using the same argument. Furthermore, by the choice of recursive call $i$, we have that $\left|\{j \in [i+1, \ell - 1] : \alpha_j^R < w_j^M(e_j\cup f_j)\}\right| = \left|\{j \in [i, \ell - 1] : \alpha_j^R < w_j^M(e_j\cup f_j)\}\right|$.
     Thus, $\phi(i+1) > \phi(i)$.

      The previous two properties of our potential function give us the following:
      $$\left|\left\{i \in [\ell - 1] : \alpha_i^R = w_i^M(e_i \cup f_i)\right\}\right| \leq \phi(\ell - 1) - \phi(1) \leq |E_1| + \left|\left\{j \in [\ell - 1] : \alpha_j^R < w_j^M(e_j\cup f_j)\right\}\right| \leq |E| + 5|V| - 1,$$ 
      where the final inequality is by \Cref{claim:WeakToStrong:hyperedge-support-size:alpha^R=weight-of-added-hyperedge} -- we note that if $\alpha_j^R < w_j^M(e_j\cup f_j)$, then $\alpha_j^R = \beta_j^R$ by definition.
\end{proof}

%% file: run-time-analysis-3.tex
\subsection{Run-time of the Algorithm}
\label{sec:WeakToStrongCover:strongly-poly-runtime:main}
In this section, we show that \Cref{alg:WeakToStrongCover} can be implemented to run in strongly polynomial time. 
\Cref{lem:WeakToStrong:strongly-polytime:main} is the main lemma of this section. 
In order to show that each recursive call of \Cref{alg:WeakToStrongCover} can be implemented in polynomial time, we will require the ability to compute the family \maximalTightSetFamily{p, H, w}.  
In the following lemma, we show that this family can indeed be computed in polynomial time. We defer the proof of the lemma to \Cref{appendix:sec:Function-Maximization-Oracles:computing-maximal-pHw-tight-sets-family}.

\begin{restatable}{lemma}{lemComputingMaximalTightSetFamily}\label{lem:FunctionMaximizetionOracles:computing-maximal-pHw-tight-sets-with-p-max-oracle}
 Let $p:2^V\rightarrow\Z$ be a skew-supermodular function and let $\left(H = (V, E), w:E \rightarrow\Z_+\right)$ be a hypergraph such that every hyperedge $e \in E$ is contained in some $(p, H, w)$-tight set, $b_{(H,w)}(u) \not = 0$ for every $u \in V$, and $b_{(H,w)}(X) \geq p(X)$ for every $X \subseteq V$.
Then, the family $\maximalTightSetFamily{p, H, w}$ can be computed in $O(|V|^2(|V|+|E|))$ time using $O(|V|^2)$ queries to \functionMaximizationOracle{p}. The run-time includes the time to construct the hypergraphs used as input to the queries to \functionMaximizationOracle{p}. Moreover, for each query to \functionMaximizationOracle{p}, the hypergraph $(G_0, c_0)$ used as input to the query has at most $|V|$ vertices and at most $|E| + 1$ hyperedges.  \end{restatable}

We now bound the run-time of \Cref{alg:WeakToStrongCover}. 

\begin{lemma}\label{lem:WeakToStrong:strongly-polytime:main}
Suppose that the input to \Cref{alg:WeakToStrongCover} is a non-empty hypergraph $(H = (V, E), w)$ and a symmetric skew-supermodular function $p:2^V\rightarrow \Z$ such that every hyperedge $\Tilde{e} \in E$ is contained in some $(p, H, w)$-tight set, $b_{(H,w)}(u) \not = 0$ for every $u \in V$, and $b_{(H,w)}(X) \geq p(X)$ for every $X \subseteq V$.
 Then, \Cref{alg:WeakToStrongCover}  can be implemented to run in $O(|V|^3(|V| + |E|)^2)$ time using $O(|V|^3(|V|+|E|))$ queries to \functionMaximizationOracleStrongCover{p}. The run-time includes the time to construct the hypergraphs used as input to the queries to \functionMaximizationOracleStrongCover{p}.
Moreover, for each query to \functionMaximizationOracleStrongCover{p}, the hypergraph $(G_0, c_0)$ used as input to the query has $O(|V|)$ vertices and $O(|V|+|E|)$ hyperedges.
\end{lemma}
\begin{proof}
Let $\ell$ denote the recursion depth of \Cref{alg:WeakToStrongCover}. By \Cref{lem:WeakToStrongCover:NumberOfHyperedges}, we have that $\ell=O(|E|+|V|)$. 
Let  $i \in [\ell]$ be a recursive call.
 We recall that $p_i:2^{V_i}\rightarrow\Z$ and $(H_i, w_i)$ are inputs to the $i^{th}$ recursive call, where $p_1 := p$ and $(H_1, w_1) := (H, w)$.
By Lemmas \ref{lem:weak-to-strong-cover:next-recursive-call-hypothesis-satification:symm-skew-supermod-strong-cover} and \ref{lem:weak-to-strong-cover:next-recursive-call-hypothesis-satification:hyperedge-in-tight-set-positive-coverage} and induction on $i$, we have that $p_i$ is a symmetric skew-supermodular function and $(H_i = (V_i, E_i), w_i)$ is a (non-empty if $i \leq \ell - 1$) hypergraph such that every hyperedge $\Tilde{e} \in E_i$ is contained in some $(p_i, H_i, w_i)$-tight set, $b_{(H_i, w_i)}(u) \not = 0$ for every $u \in V_i$, and $b_{(H_i, w_i)}(X) \geq p_i(X)$ for every $X \subseteq V_i$. Let $e_i, f_i, \beta_i^M, \alpha_i^M, (H_i^M, w_i^M), \beta_i^R, \alpha_i^R, (H_i^R, w_i^R)$, $p_i^R$, $(H_0^{i}, w_0^{i})$, $E_0^i, \zeros_i$ denote the quantities  as defined by \Cref{alg:WeakToStrongCover} during the $i^{th}$ recursive call. We recall that $V_i \subseteq V$. By Lemma \ref{lem:WeakToStrong:properties-of-alpha^M-and-alpha^R}(2), (5) and induction on $i$, we have that $|E_i| \leq |E| + i \leq |E| + \ell = O(|E| + |V|)$, where the final equality is by \Cref{lem:WeakToStrongCover:NumberOfHyperedges}(1). 
Furthermore, the hypergraph returned in Step \ref{alg:WeakToStrongCover:(14)} contains at most $|V|$ vertices and at most $|E_i| + O(|V_i|)= O(|V| + |E|)$ hyperedges by \Cref{lem:WeakToStrongCover:NumberOfHyperedges}(2) and induction on $i$. 
The following claim shows that all steps except for Steps \ref{alg:WeakToStrongCover:(10)},  \ref{alg:WeakToStrongCover:(13)}, and  \ref{alg:WeakToStrongCover:(14)} of \Cref{alg:WeakToStrongCover} in the $i^{th}$ recursive call can be implemented to run in $O(|V_i|^2(|V_i| + |E_i|))$ time
using $O(|V_i|^2)$ queries \functionMaximizationOracle{p_i}. This includes the run-time to construct the hypergraphs used in the queries to \functionMaximizationOracle{p_i}. We will address the run-time for Steps \ref{alg:WeakToStrongCover:(10)},  \ref{alg:WeakToStrongCover:(13)}, and  \ref{alg:WeakToStrongCover:(14)} later. 

\begin{claim}\label{claim:weak-to-strong-runtime-of-one-recursive-call}
For every $i\in [\ell]$, all steps except for Steps \ref{alg:WeakToStrongCover:(10)},  \ref{alg:WeakToStrongCover:(13)}, and  \ref{alg:WeakToStrongCover:(14)} of \Cref{alg:WeakToStrongCover} in the $i^{th}$ recursive call can be implemented to run in 
$O(|V_i|^2(|V_i| + |E_i|))$ time
using 
$O(|V_i|^2)$ queries to \functionMaximizationOracle{p_i}. The run-time includes the time to construct the hypergraphs used as input to the queries to \functionMaximizationOracle{p_i}. Moreover, for each query to \functionMaximizationOracle{p_i}, the hypergraph $(G_0, c_0)$ used as input to the query has  $|V_i|$ vertices and at most $|E_i| + 1$ hyperedges.
\end{claim}
\begin{proof}
We first show how to compute hyperedges $e_i$ and $f_i$ (Step \ref{alg:WeakToStrongCover:(4)}). We recall that by \Cref{cor:weak-to-strong:uncrossing:exist-two-hyperedges-contained-in-disjoint-maximal-tight-sets}, there exist two hyperedges $e_i, f_i \in E_i$ such that $\maximalTightSetContainingHyperedge{p_i, H_i, w_i}(e_i) \not = \maximalTightSetContainingHyperedge{p_i, H_i, w_i}(f_i)$ and consequently Step \ref{alg:WeakToStrongCover:(4)} is well defined. By \Cref{lem:weak-to-strong:uncrossing:maximal-tight-set-family-disjoint}, the family $\maximalTightSetFamily{p_i, H_i, w_i}$ is a disjoint family. 
Furthermore,  by \Cref{lem:FunctionMaximizetionOracles:computing-maximal-pHw-tight-sets-with-p-max-oracle}, we can compute the family $\maximalTightSetFamily{p_i, H_i, w_i}$ in $O(|V_i|^2(|V_i|+|E_i|))$ time using $O(|V_i|^2)$ queries to \functionMaximizationOracle{p_i}. 
Moreover, for each query to \functionMaximizationOracle{p_i}, the hypergraph $(G_0, c_0)$ used as input to the query has at most $|V_i|$ vertices and at most $|E_i|+1$ hyperedges. 
We note that $\maximalTightSetContainingHyperedge{p_i, H_i, w_i}(\Tilde{e}) \in \maximalTightSetFamily{p_i, H_i, w_i}$ for every $\Tilde{e} \in E$.
Consequently, we can find the required hyperedges for Step \ref{alg:WeakToStrongCover:(4)} by finding a pair of distinct sets $X, Y \in \maximalTightSetFamily{p_i, H_i, w_i}$ that each contains a hyperedge, and define these hyperedges to be $e_i$ and $f_i$ respectively. Thus, Step  \ref{alg:WeakToStrongCover:(4)} can be implemented to run in $O(|V_i|^2(|V_i|+|E_i|))$ time using $O(|V_i|^2)$ queries to \functionMaximizationOracle{p_i}. Moreover, for each query to \functionMaximizationOracle{p_i}, the hypergraph $(G_0, c_0)$ used as input to the query has at most $|V_i|$ vertices and at most $|E_i|+1$ hyperedges. 

Next, we show how to compute the values $\alpha_i^M$ (Step \ref{alg:WeakToStrongCover:(5)}) and $\alpha^R_i$ (Step \ref{alg:WeakToStrongCover:(7)}).  We first show how to compute the value $\beta_i^M$.
We have the following:
\begin{align*}
    \beta_i^M& = \min\left\{b_{(H_i,w_i)}(X) - p_i(X) : X \subseteq V_i \text{ and } e_i \cap X, f_i \cap X \not = \emptyset\right\} &\\
    & = \min_{u \in e_i, v \in f_i}\left\{\min\left\{b_{(H_i,w_i)}(X) - p_i(X) : \{u,v\} \subseteq X \subseteq V_i\right\}\right\}.&
\end{align*}
For $u \in e_i$ and $v \in f_i$, let $\beta_i^{uv} := \min\{b_{(H_i,w_i)}(X) - p_i(X) : \{u,v\} \subseteq X \subseteq V_i\}.$ Thus, $\beta_i^M = \min_{u \in e_i, v\in f_i}\beta_i^{uv}$.
Then, $\beta_i^M$ (and consequently the value $\alpha_i^M$) can be computed in $O(|V_i|^2(|V_i| + |E_i|))$ time using 
$|V_i|^2$ queries to \functionMaximizationOracle{p_i}, where each query is on a input hypergraph $(G_0, c_0)$ that has at most $|V_i|$ vertices and at most $|E_i|+1$ hyperedges. This is because for every $u \in e_i$ and $v \in f_i$ (there are $O(|V_i|^2)$ choices for $u$ and $v$), the value $\beta_i^{uv}$ can be computed using a single oracle call to \functionMaximizationOracle{p_i} on a hypergraph $(G_0, c_0)$ that has at most $|V_i|$ vertices and at most $|E_i|+1$ hyperedges.
Next, the value $\beta^R_i$, (and consequently the value $\alpha^R_i$)  can be computed 
in $O(|V_i| + |E_i|)$ time using a single query to \functionMaximizationOracle{p_i} on the input hypergraph $(G_0=H_i^M, c_0=w_i^M)$ that has  $|V_i|$ vertices and at most $|E_i|+1$ hyperedges. Here, the hypergraph $(H_i^M, w_i^M)$ has  $|V_i|$ vertices and $|E_i|+1$ hyperedges because of  \Cref{lem:WeakToStrong:properties-of-alpha^M-and-alpha^R}(2).
Next, we show that the hypergraphs $(H_i^M, w_i^M)$ (Step \ref{alg:WeakToStrongCover:(6)}) and $(H_i^R, w_i^R)$ (Step \ref{alg:WeakToStrongCover:(8)}) can be computed in $O(|V_i| + |E_i|)$ time. We recall that the \merge operation decreases the weights of hyperedges $e_i$ and $f_i$ by $\alpha_i^M$,  increases the weight of the hyperedge $e_i\cup f_i$ by $\alpha_i^M$ and discards the resulting zero-weighted edges (if any). This can be implemented to run in $O(|E_i|)$ time. Furthermore, the \reduce operation decreases the weight of the hyperedge $e_i\cup f_i$ by $\alpha^R_i$, and so can also be implemented to run in $O(|E^R_i|) = O(|E_i|)$ time by \Cref{lem:WeakToStrong:properties-of-alpha^M-and-alpha^R}(2) and (5). Finally, the computations of the hypergraph $(H^{i}_0, w^{i}_0)$ (Step \ref{alg:WeakToStrongCover:(9)}), the set $\zeros_i$ (\Cref{alg:WeakToStrongCover:(11)}) and the set $V'_i$  (\Cref{alg:WeakToStrongCover:(12)}) can each be done in $O(|V_i|)$ time. Furthermore, the hypergraph returned in Step \ref{alg:WeakToStrongCover:(14)} contains at most $|V|$ vertices and at most $|E_i| + O(|V_i|)$ hyperedges. Consequently, Steps \ref{alg:WeakToStrongCover:(15)} and \ref{alg:WeakToStrongCover:(16)}  can be implemented to run in $O(|V_i| + |E_i|)$ time. 
\end{proof}

\Cref{claim:weak-to-strong-runtime-of-one-recursive-call} implies that each query to \functionMaximizationOracle{p_i} takes as input a hypergraph 
with at most $|V_i|$ vertices and $O(|E_i|)$ hyperedges.  
Next, we focus on the time to implement a query to \functionMaximizationOracle{p_i} on such hypergraphs. 

\begin{claim}\label{claim:WeakToStrong:p_i-oracle-using-p-oracle}
For every $i\in [\ell]$ and for every given hypergraph $(G = (V_i, E_G), c:E_G\rightarrow\Z_+)$ such that $|E_G| = O(|E_i|)$ AND disjoint subsets $S, T\subseteq V_i$, $\functionMaximizationOracle{p_i}((G, c), S, T, y)$ can be computed in $O(|V|(|V| + |E|))$ time using at most $|V_i|+1$ queries to \functionMaximizationOracleStrongCover{p}. The run-time includes the time to construct the hypergraphs used as input to the queries to \functionMaximizationOracleStrongCover{p}. Moreover, for each query to \functionMaximizationOracleStrongCover{p}, the hypergraph $(G_0, c_0)$ used as input to the query has $O(|V|)$ vertices and $O(|V|+|E|)$ hyperedges.
\end{claim}
\begin{proof}
Let the input to \functionMaximizationOracle{p_i} be a hypergraph $(G = (V_i, E_G), c:E_G\rightarrow\Z_+)$ such that $|E_G| = O(|E_i|)$ and disjoint subsets $S, T\subseteq V_i$. By \Cref{lem:Preliminaries:wc-oracle-from-sc-oracle}, we can implement $\functionMaximizationOracle{p_i}((G, c), S, T, y_0)$ in $O(|V_i|(|V_i| + |E_i|))$ time using at most $|V_i|+1$ queries to $\functionMaximizationOracleStrongCover{p_i}$, where each query to $\functionMaximizationOracleStrongCover{p_i}$ is on a input hypergraph with $|V_i|+1$ vertices and $O(|E_i|)$ hyperedges. Next, we show that the query $\functionMaximizationOracleStrongCover{p_i}((G, c), S, T)$ where $(G = (V_i, E_G), c:E_G\rightarrow\Z_+)$ is a hypergraph with $|E_G| = O(|E_i|)$ and $S, T\subseteq V_i$ are disjoint subsets, can be implemented in $O(|V| + |E|)$ time using one query to  $\functionMaximizationOracleStrongCover{p}$ on an input hypergraph $(G_0, c_0)$ that has $O(|V|)$ vertices and $O(|V|+|E|)$ hyperedges. 

We first recall some notation and define some quantities for convenience. For $i\in [\ell]$, we recall that $(H_0^i, w_0^i)$ is the hypergraph $(H_0, w_0)$, $p^R_i, p_i'$ are the functions $p^R, p'$ respectively, and $\zeros_i$ is the set $\zeros$ defined in the $i^{th}$ recursive call of the algorithm. We note that $p_i' =p_{i+1}$ for every $i\in [\ell-1]$. Furthermore, the hypergraph $(H_0^i, w_0^i)$ contains exactly one hyperedge for every $i\in [\ell-1]$. 
For $i\in [\ell-1]$, let $(H_1^i, w_1^i)$ be the hypergraph obtained from $(H_0^i, w_0^i)$ by adding the vertices $\cup_{j=1}^{i-1}\zeros_{j}$.

Let $i\in [\ell]$. Let the input to \functionMaximizationOracleStrongCover{p_i} be a hypergraph $(G = (V_i, E_G), c:E_G\rightarrow\Z_+)$ such that $|E_G| = O(|E_i|)$ and disjoint subsets $S_0, T_0\subseteq V_i$. 
We note that $p_i=p_{i-1}'$ and hence, it suffices to answer $\functionMaximizationOracleStrongCover{p_{i-1}'}((G, c), S_0, T_0)$. 
Construct the hypergraph $(G_1, c_1)$ to be the hypergraph obtained from $(G,c)$ by adding vertices $\cup_{j=1}^{i-1}\zeros_{j}$. 
Construct the hypergraph $(\Tilde{H}, \Tilde{w})=(\sum_{j=1}^{i-1}H_1^j, \sum_{j=1}^i w_1^j)$. The hypergraph $(\Tilde{H}, \Tilde{w})$ contains at most $i - 1 \leq \ell = O(|E| + |V|)$ hyperedges, and can be constructed in time $O(|V|+|E|)$. 
Next, construct the hypergraph $(G, c)=(G_1+\Tilde{H}, c_1+\Tilde{w})$. The hypergraph $(G, c)$ contains $O(|E|+|V|)$ hyperedges and can be constructed in time $O(|V|+|E|)$. 
Obtain $(Z, p(Z))$ by querying $\functionMaximizationOracleStrongCover{p}((G, c), S_0, T_0)$ and return $(Z-\cup_{j=1}^{i-1}\zeros_{j}, p(Z))$. 
\end{proof}

By \Cref{claim:weak-to-strong-runtime-of-one-recursive-call}, each query to \functionMaximizationOracle{p_i} takes as input a hypergraph with $|V_i|\le |V|$ vertices and $O(|E_i|) = O(|E|+|V|)$ hyperedges. 
By \Cref{claim:WeakToStrong:p_i-oracle-using-p-oracle}, a query to \functionMaximizationOracle{p_i} on a input hypergraph with at most $|V_i|$ vertices and at most $O(|E_i|)$ hyperedges 
can be implemented in $O(|V|(|V| + |E|))$  time using $O(|V|)$ queries to \functionMaximizationOracleStrongCover{p}, where each query to \functionMaximizationOracleStrongCover{p} is on a input hypergraph $(G_0, c_0)$ that has $O(|V|)$ vertices and $O(|E|+|V|)$ hyperedges. 
Thus, all steps except for Steps \ref{alg:WeakToStrongCover:(10)}, \ref{alg:WeakToStrongCover:(13)}, and \ref{alg:WeakToStrongCover:(14)} of \Cref{alg:WeakToStrongCover} in the $i^{th}$ recursive call can be implemented to run in $O(|V|^3(|V| + |E|))$ time using $O(|V|^3)$ queries to \functionMaximizationOracleStrongCover{p}, where each query to \functionMaximizationOracleStrongCover{p} is on a input hypergraph $(G_0, c_0)$ that has $O(|V|)$ vertices and $O(|E|+|V|)$ hyperedges. 

We note that Steps \ref{alg:WeakToStrongCover:(10)} and \ref{alg:WeakToStrongCover:(13)} need not be implemented explicitly for the purposes of the algorithm. Instead, \functionMaximizationOracle{p_{i+1}} on input hypergraphs with $|V_{i+1}|\le |V|$ vertices and at most $|E_{i+1}|+1$ hyperedges can be used to execute all steps except for Steps \ref{alg:WeakToStrongCover:(10)}, \ref{alg:WeakToStrongCover:(13)}, and \ref{alg:WeakToStrongCover:(14)} in the $(i+1)^{th}$ recursive call. 
We have already seen in \Cref{claim:WeakToStrong:p_i-oracle-using-p-oracle} that a query to \functionMaximizationOracle{p_{i+1}} on input hypergraphs with at most $|V|$ vertices and $O(|E_{i+1}|) = O(|V| + |E|)$ hyperedges can be implemented in  $O(|V|(|V|+|E|))$ time using $O(|V|)$ queries to \functionMaximizationOracleStrongCover{p}. 
Step \ref{alg:WeakToStrongCover:(14)} is the recursive call. We note that the hypergraph returned in Step \ref{alg:WeakToStrongCover:(14)} contains at most $|V|$ vertices and at most $|E_i| + O(|V_i|)= O(|V| + |E|)$ hyperedges. 

Thus, we have shown that each recursive call can be implemented to run in time $O(|V|^3(|V| + |E|))$ using $O(|V|^3)$ queries to \functionMaximizationOracleStrongCover{p}, where each query to \functionMaximizationOracleStrongCover{p} is on an input hypergraph $(G_0, c_0)$ that has $O(|V|)$ vertices and $O(|E|+|V|)$ hyperedges. The number of recursive calls is $O(|V|+|E|)$. 
Hence, the algorithm can be implemented to run in $O(|V|^3(|V| + |E|)^2)$ time using $O(|V|^3(|V|+|E|))$ queries to \functionMaximizationOracleStrongCover{p}, where each query to \functionMaximizationOracleStrongCover{p} is on a input hypergraph $(G_0, c_0)$ that has $O(|V|)$ vertices and $O(|E|+|V|)$ hyperedges. 
\end{proof}

%% file: conclusion.tex
\section{Conclusion}\label{sec:conclusion}
We introduced a splitting-off operation in (weighted) hypergraphs. Our contribution on this front is conceptualizing an appropriate notion of splitting-off in hypergraphs. 
Next, we proved that for every hypergraph there exists a local-connectivity preserving complete h-splitting-off at an arbitrary vertex from the hypergraph.  
Although our proof of existence follows from previously known existential results of Bern\'{a}th and Kir\'{a}ly \cite{Bernath-Kiraly} for an abstract function cover problem, our main contribution is identifying that our newly introduced notion of hypergraph splitting-off falls within the framework of their abstract function cover problem. 
Next, we designed a strongly polynomial-time algorithm to find a local-connectivity preserving complete splitting-off at a vertex. This involved substantial technical challenges to overcome. In particular, our main contribution towards the strongly polynomial-time algorithm is a strengthening of the above-mentioned existential result: we showed that there exists a local connectivity preserving complete h-splitting-off at an arbitrary vertex from the hypergraph 
\emph{in which the number of additional hyperedges is polynomial in the number of vertices}. 
The strongly polynomial-time algorithm follows from our techniques to achieve this stronger existential result via standard algorithmic tools in submodularity. 
Finally, we illustrated the usefulness of the existence of local connectivity preserving complete h-splitting-off at an arbitrary vertex from a hypergraph by presenting two applications. 
Our first application is to give a constructive characterization of $k$-hyperedge-connected hypergraphs. Our second application is to give an alternative proof of an approximate min-max relation for 
the max Steiner rooted-connected orientation problem in graphs and hypergraphs. 
Two notable features of our proof of the approximate min-max relation for \emph{graphs} are that (1) it avoids the strong orientation theorem of Nash-Williams and (2) it proves a result for graphs using tools developed for hypergraphs. 

Local and global connectivity preserving complete splitting-off at a vertex from a graph is a powerful operation for graphs. It finds applications in a variety of graph problems including graph orientation \cite{Mad78, Frank93:Applications-of-submod-functions, FK02, Kiraly-Lau}, connectivity augmentation \cite{Fra92, Fra94, JFJ95, Frank-book, nagamochi_ibaraki_2008_book}, 
minimum cuts enumeration \cite{NNI97, HW96, Goe01}, network design \cite{Lov76, Jor03, CS08, CFLY11}, tree packing \cite{BHTP03, Kri03, JMS03, Lau07}, congruency-constrained cuts \cite{NZ20}, and approximation algorithms for TSP \cite{GB93, BN23}. 
We believe that our local connectivity preserving complete splitting-off results for hypergraphs is likely to find future applications akin to its counterpart in graphs. 
We mention some of the open questions raised by our work:
\begin{enumerate}
    \item Our work focused on local connectivity preserving complete \emph{h-splitting-off} at a vertex from a hypergraph. We gave an example showing that local/global connectivity preserving complete \emph{g-splitting-off} at a vertex from a hypergraph may not exist (Figure \ref{figure:hypergraph-splitting-off-example-2}). 
    Are there sufficient conditions to guarantee local/global connectivity preserving complete \emph{g-splitting-off} at a vertex from a hypergraph? We recall that Lov\'{a}sz's \cite{Lov74, Lovasz-problems-book} and Mader's \cite{Mad78} results give sufficient conditions to guarantee local and global connectivity preserving complete g-splitting-off at a vertex from a \emph{graph}. 
    \item As one of the applications of our splitting-off result, we presented an alternative proof of an approximate min-max relation for the max Steiner \emph{rooted}-connected orientation problem in hypergraphs. The computational complexity of a closely related hypergraph orientation problem is open: In max Steiner connected orientation problem in hypergraphs, the input is a hypergraph $G=(V, E)$ and a subset $T$ of terminals. The goal is to find the maximum $k$ and an orientation $\overrightarrow{G}$ of $G$ such that $\overrightarrow{G}$ contains $k$ hyperarc-disjoint paths from $u$ to $v$ for every pair of distinct terminals $u, v\in T$. Max Steiner connected orientation problem in \emph{graphs} is solvable in polynomial time via the Nash-Williams' strong orientation theorem. 
    Is max Steiner connected orientation problem in \emph{hypergraphs} solvable in polynomial time? 
\end{enumerate} 

%% file: acknowledgement.tex
\medskip
\paragraph{Acknowledgements.} 
Karthekeyan thanks Eklavya Sharma for engaging in preliminary discussions on hypergraph splitting-off. 
Karthekeyan and Shubhang were supported in part by NSF grants CCF-1814613 and CCF-1907937. 
Karthekeyan was supported in part by the Distinguished Guest Scientist Fellowship of the Hungarian Academy of Sciences -- grant number VK-6/1/2022. 
Krist\'{o}f and Tam\'{a}s were supported in part 
by the Lend\"ulet Programme of the Hungarian Academy of Sciences -- grant number LP2021-1/2021, by the Ministry of Innovation and Technology of Hungary from the National Research, Development and Innovation Fund -- grant number ELTE TKP 2021-NKTA-62 funding scheme, and by the Dynasnet European Research Council Synergy project -- grant number ERC-2018-SYG 810115.

%% file: appendix.tex
\input{appendix-weak-to-strong-quadratic-example}
\input{appendix-uncrossing-properties}
\input{appendix-function-maximization-oracles-2}

%% file: appendix-weak-to-strong-quadratic-example.tex
\section{An Example for Weak To Strong Cover with Quadratic New Hyperedges}\label{sec:appendix:tight-example-weak-to-strong}

\begin{lemma}
There exists a symmetric skew-supermodular function  $p:2^V\rightarrow\Z$ and a hypergraph\linebreak $\left(H = (V, E), w:E \rightarrow\Z_+\right)$ satisfying $b_{(H, w)}(X) \geq p(X)$ for every $X \subseteq V$ such that for every hypergraph $\left(H^* = (V, E^*), w^*:E^* \rightarrow\Z_+\right)$ obtained by merging hyperedges of the hypergraph $(H, w)$ such that $d_{(H^*, w^*)}(X) \geq p(X)$ for every $X\subseteq V$, we have that 
\[
|E^*-E|=\Omega(|V|^2). 
\]
\end{lemma}
\begin{proof}
Let $u\in V$ be a fixed vertex. Let $p:2^V\rightarrow\Z\cup\{-\infty\}$ be a symmetric skew-supermodular function defined as follows,
$$p(X) \coloneqq  \begin{cases}
    {\binom{|V| - 1}{2}}& \text{ if } X \in \left\{\{u\}, V - u\right\},\\
    - \infty &\text{otherwise.}
\end{cases}$$
Furthermore, let $(H = (V, E), w:E\rightarrow\Z_+)$ be a hypergraph such that $E \coloneqq  {\binom{V - u}{2}}\cup \{\{u\}\}$, and the weight function $w$ is defined as
$$w(e) \coloneqq  \begin{cases}
    1& \text{ if } e \in {\binom{V - u}{2}}\\
    {\binom{|V| - 1}{2}}& \text{otherwise (i.e., if $e = \{u\}$)}.
\end{cases}$$

We note that $b_{H, w}(X) \geq p(X)$ for every $X \subseteq V$. Next, consider the hypergraph $(H^* = (V, E^*), w^*)$, where the hyperedge set $E^* \coloneqq  \left\{\{u, a, b\} : ab \in {\binom{V-u}{2}}\right\}$, and the weight function $w^*(e) \coloneqq  1$ for every $e \in E^*$. Then, the hypergraph $(H^*, w^*)$ strongly covers the function $p$ because $d_{H^*, w^*}(X) \geq p(X)$ for every $X \subseteq V$. Furthermore, the hypergraph $(H^*, w^*)$ is the \emph{unique} hypergraph that can be obtained by merging hyperedges of $(H, w)$ to strongly cover the function $p$. We observe that $|E^*-E| = \Omega(|V|^2)$. 
\end{proof}

%% file: appendix-uncrossing-properties.tex
\section{Projection of Laminar Family}\label{appendix:sec:uncrossing-properties:projection-of-laminar-families}
We recall that 
a family $\calL \subseteq 2^V$ is \emph{laminar} if for every pair of sets $X, Y\calL$, we have that either $X\subseteq Y$ or $Y\subseteq X$ or $X\cap Y=\emptyset$  --  we use \emph{family} to refer to a set (and not a multiset) of subsets of the ground set $V$. 
For a set $\calZ\subseteq V$, and a family $\calL \subseteq 2^V$,
the \emph{projection} of the family $\calL$ to the ground set $V - \calZ$ is the family formed by restricting each set of $\calL$ to the ground set $V - \calZ$, i.e. $\{X - \calZ : X \in \calL\}$.
The next lemma says that if the family $\calL$ is laminar, then the size of the projection 
is comparable to that of the original family $\calL$.

\projectionLaminarFamily*
\begin{proof}
    We first show that the family $\calL'$ is a laminar family. By way of contradiction, let $X,Y \in \calL'$ be two distinct sets such that $X-Y, Y-X, X\cap Y \not = \emptyset$. Then, there exist distinct sets $A, B \in \calL$ such that $A - \calZ = X$ and $B - \calZ = Y$. Since the family $\calL$ is laminar, we have that either $A\cap B = \emptyset$, $A \subseteq B$ or $B\subseteq A$. We note that if $A\subseteq B$, then we have that $X = A - \calZ \subseteq B - \calZ = Y$, which contradicts our choice of sets $X$ and $Y$. Similarly, if $B\subseteq A$, then we obtain that $Y \subseteq X$, contradicting our choice of sets $X$ and $Y$. Thus, we have that $A\cap B = \emptyset$, and so $X\cap Y = (A - \calZ)\cap(B - \calZ) = \emptyset$, once again contradicting our choice of sets $X$ and $Y$. 

    Next, we show the second part of the claim, i.e. $|\calL| \leq |\calL'| + 3|\calZ|$. For this, we first  partition the family $\calL$ into two subfamilies $\calL_1 : = \{A \in \calL : A \subseteq \calZ\}$ and $\calL_2 \coloneqq  \calL - \calL_1$. Furthermore, for every $X \in \calL'$, we let $S_X \coloneqq  \{A \in \calL_2 : A - \calZ = X\}$ denote the family of sets in $\calL_2$ that map to the (non-empty) set $X$ under projection with the set $\calZ$. 
    We note that for every $A \in \calL_2$, we have that $\emptyset \not = A - \calZ \in \calL'$ and hence $A \in S_{A - \calZ}$. Consequently, the subfamilies $S_X$ for $X \in \calL'$ partition the family $\calL_2$ 
    . In particular, we have that $\calL = \calL_1 \uplus \left(\biguplus_{X \in \calL'} S_X \right)$. \Cref{claim:Laminar-family-projection:S_X-is-a-chain} below shows that the family $S_X$ is a chain family for every $X \in \calL'$. Next, for every $X \in \calL'$, we define three sets: we let $A_X \coloneqq  \arg\max\{|A| : A \in S_X\}$, $B_{X} \coloneqq  \arg\min\{|B| : B \in S_X\}$, and $\calZ_X\coloneqq A_X-B_X$. 
    We note that $X \subseteq B_X \subseteq A_X$, where the second containment is because the family $S_X$ is a chain family. 
    Since $A_X - B_X \subseteq A_X - X \subseteq \calZ$, we have that $\calZ_X \subseteq \calZ$. \Cref{claim:Laminar-family-projection:|S_X|-atmost-|Z_X|+1} below shows that $|S_X| \leq |\calZ_X| + 1$. Furthermore, \Cref{claim:Laminar-family-projection:Z_X-Z_Y-disjoint} below shows that $\calZ_X$ and $\calZ_Y$ are disjoint for distinct $X,Y \in \calL'$. Then, we have the following:
\begin{align*}
        |\calL| & = |\calL_1| + \sum_{X \in \calL'} |S_X|&\\
        & \leq 2|\calZ| - 1 + \sum_{X \in \calL'}\left(|\calZ_X| + 1\right)&\\
        &\le 3|\calZ| + |\calL'|-1.&
    \end{align*}
    Here, the first equality is because $\calL = \calL_1 \uplus \left(\biguplus_{X \in \calL'} S_X \right)$. The first inequality is because the family $\calL_1$ is a laminar family on the ground set $\calZ$, and by \Cref{claim:Laminar-family-projection:|S_X|-atmost-|Z_X|+1}. The final inequality is because the sets $\calZ_X$ for $X \in\calL'$ form a subpartition of the set $\calZ$ by \Cref{claim:Laminar-family-projection:Z_X-Z_Y-disjoint}.
\end{proof}

\begin{claim}\label{claim:Laminar-family-projection:S_X-is-a-chain}
        For every $X \in\calL'$, the family $S_X$ is a chain family.
    \end{claim}
    \begin{proof}
        By way of contradiction, let $X \in \calL'$ be such that the family $S_X$ is not a chain. Then, there exist sets $A, B \in S_X$ such that $A - B, B - A \not = \emptyset$. Since $A$ and $B$ are sets in the laminar family $\calL$, we have that the sets $A$ and $B$ are disjoint. However, both the sets $A$ and $B$ contain the set $X$, and so $A\cap B \not = \emptyset$, a contradiction to the fact that $A$ and $B$ are sets in the laminar family $\calL$.
    \end{proof}

\begin{claim}\label{claim:Laminar-family-projection:|S_X|-atmost-|Z_X|+1}
    For every $X \in \calL'$, we have that $|S_X| \leq |\calZ_{X}| + 1$.
\end{claim}
\begin{proof}
    We may assume that $|S_X| \geq 2$, i.e.,  $A_X \not = B_X$, as otherwise the claim holds. Then, we have the following:
    $|S_X| \leq |A_X| - |B_X| + 1 \leq |\calZ_X| +1,$ where the first inequality is by \Cref{claim:Laminar-family-projection:S_X-is-a-chain}.
\end{proof}

\begin{claim}\label{claim:Laminar-family-projection:Z_X-Z_Y-disjoint}
    For distinct non-empty sets $X, Y \in \calL'$, we have that $\calZ_{X}\cap\calZ_{Y} = \emptyset$.
\end{claim}
\begin{proof}
     Let $X, Y \in\calL'$ be distinct sets. Since the family $\calL'$ is laminar, we have that either $X \subsetneq Y$ or $Y\subsetneq X$ or $X\cap Y = \emptyset$.

    First, suppose that $X \subsetneq Y$ (we note that the case of $Y\subsetneq X$ can be show by a similar argument). Here, we note that $X \subseteq A_X \cap B_Y $ since $X \subseteq A_X$ and $X \subseteq Y \subseteq B_Y$. Consequently, $A_X\cap B_Y\neq \emptyset$. 
    Since the family $\calL$ is laminar and $A_X, B_Y \in \calL$, we have that either $A_X \subseteq B_Y$ or $B_Y \subseteq A_X$. We note that if $A_X \subseteq B_Y$, then the claim holds since $\calZ_X \subseteq A_X \subseteq B_{Y} \subseteq V - \calZ_Y$.
    We now show that $B_Y - A_X \not = \emptyset$, i.e. the scenario $B_Y \subseteq A_X$ cannot happen. 
    Consider a vertex $v \in Y - X$. We note that 
    $v \not \in \calZ$
    since $v \in Y$. Furthermore, $v \not \in A_X$ since $v \not \in X$ and $A_X - X \subseteq \calZ$.
    Thus, we have that $v \in B_Y - A_X \not = \emptyset$.

    Next, suppose that the sets $X$ and $Y$ are disjoint. We also suppose that $A_X \cap A_Y \not = \emptyset$ as otherwise, we have that $\calZ_X \cap \calZ_Y \subseteq A_X \cap A_Y = \emptyset$ and the claim holds. Since the family $\calL$ is laminar and $A_X, A_Y \in \calL$, either $A_X \subseteq A_Y$ or $A_Y \subseteq A_X$. Suppose that  $A_X \subseteq A_Y$ (we note that the case of $A_Y \subseteq A_X$ can be show by a similar argument).
    We also suppose that $A_X - B_Y \not = \emptyset$, as otherwise $\calZ_X = A_X - B_X \subseteq B_Y \subseteq V - (A_Y - B_Y) = V - \calZ_Y$, and the claim holds. 
    The rest of the proof is devoted to showing that this case is impossible. In particular, we show that
    the sets $A_X$ and $B_Y$ are such that $A_X-B_Y, B_Y-A_X, A_X\cap B_Y\neq \emptyset$, contradicting the laminarity of the family $\calL$. First, we show that $B_Y - A_X \not = \emptyset$. By way of contradiction, suppose that $B_Y \subseteq A_X$. Then, we have that $Y \subseteq A_X$ since $Y \subseteq B_Y$. Furthermore,  we recall that $Y\cap \calZ = \emptyset$. Thus, we have that $\emptyset = Y \cap X = Y \cap (A_X - \calZ) = Y\cap A_X$. Thus, $Y\subseteq A_X$ and $Y\cap A_X=\emptyset$ together imply that $Y = \emptyset$, contradicting $Y \in \calL'$. Second, we show that $B_Y\cap A_X \not = \emptyset$. By way of contradiction, suppose that $B_Y\cap A_X = \emptyset$. We recall that $A_X\subseteq A_Y$. Consequently, we have that $X \subseteq A_X \subseteq A_Y - B_Y \subseteq \calZ$, contradicting $X \in \calL'$. Thus the sets $A_X$ and $B_Y$ are such that $A_X-B_Y, B_Y-A_X, A_X\cap B_Y\neq \emptyset$, giving us the required contradiction.
\end{proof}

%% file: appendix-function-maximization-oracles-2.tex
\section{Optimization Problems using Function Maximization Oracle}\label{appendix:sec:Function-Maximization-Oracles}
In this section, we prove Lemmas 
\ref{lem:Preliminaries:wc-oracle-from-sc-oracle}, \ref{lemma:helper-for-p-max-oracle}, and \ref{lem:FunctionMaximizetionOracles:computing-maximal-pHw-tight-sets-with-p-max-oracle}. The proofs of these lemmas follow from known tools in the area of submodular functions. We present their proofs for the sake of completeness.

\subsection{\functionMaximizationOracle{p} using \functionMaximizationOracleStrongCover{p}}\label{sec:wc-oracle-from-sc-oracle}
In this section, we prove \Cref{lem:Preliminaries:wc-oracle-from-sc-oracle}. 
\lemmaWCoraclefromSCoracle*
\begin{proof}
Let $u \in V$ be a vertex. We define the hypergraph $G^u := (V, E^u := \{e \cup \{u\} : e \in E\})$ and the weight function $w^u : E^u \rightarrow\Z_+$ as $w^u(e \cup \{u\}) := w(e)$ for every $e \in E$. We observe that $b_{(G, w)}(X) = d_{(G^u, w^u)}(X)$ for every set $X \subseteq V - \{u\}$. We also let $(Z^u, p(Z^u))$ denote the tuple obtained by querying \functionMaximizationOracleStrongCover{p}$((G_0, c_0) := (G^u, w^u), S_0 := S, T_0 := T\cup \{u\})$.

We now describe a procedure to compute an answer $(Z, p(Z))$ to the query $\functionMaximizationOracle{p}((G,w), S, T)$. We consider two cases to define our set $Z$. If $T \not= \emptyset$, then we choose an arbitrary vertex $t \in T$ and define $Z := Z^t$. Otherwise, $T = \emptyset$ and we define the set $Z$ to be an arbitrary optimum solution to $\max\{p(Z) - b_{(G, w)}(Z) : Z \in \calZ \}$, where $\calZ := \{Z^u: u \in V\} \cup \{ V\}$.  We return the tuple $(Z, p(Z))$ as the answer to the query $\functionMaximizationOracle{p}((G, w), S, T)$.

We now prove correctness of the above procedure. Let $Z$ be the set returned by the above procedure. Then, by definition $S\subseteq V\subseteq V-T$. Our goal is to show that $Z$ is an optimum solution to the following problem:
\begin{align}
    \max\left\{p(X)-b_{(G, w)}(X): S\subseteq X\subseteq V-T\right\} \label{eq:wc-oracle-from-sc-oracle}
\end{align}
Let $Y$ be an optimum solution to \eqref{eq:wc-oracle-from-sc-oracle}. It suffices to show that $p(Z)-b_{(G, w)}(Z) \ge p(Y)-b_{(G, w)}(Y) $.   
We consider two cases: suppose that $T\neq \emptyset$. Then, we have that 
\begin{align*}
p(Z)-b_{(G, w)}(Z)
&=p(Z^t)-d_{(G^t, w^t)}(Z^t)\\
&\ge p(Y)-d_{(G^t, w^t)}(Y) \\
&= p(Y) - b_{(G, w)}(Y) .
\end{align*} 
In the above, the first equality is because $Z=Z^t$, the inequality is because $Y\subseteq V-\{v\}$ and by definition of $Z^t$, and the last equality is because $b_{(G, w)}(Y) = d_{(G^v, w^v)}(Y)$ since $Y\subseteq V-\{t\}$. 
Next, suppose that $T=\emptyset$. We consider two subcases. Firstly, suppose that $V$ is the unique optimum solution to \eqref{eq:wc-oracle-from-sc-oracle}. Then, $Z=Y=V$ by definition and we are done. 
Secondly, suppose that $Y\neq V$. Let $v\in V-Y$.  
Then, we have the following:
\begin{align*}
p(Z)-b_{(G, w)}(Z)
&\ge p(Z^v) - b_{(G, w)}(Z^v) \\
&=p(Z^v)-d_{(G^v, w^v)}(Z^v) \\
&\ge p(Y)-d_{(G^v, w^v)}(Y) \\
&=p(Y)-b_{(G, w)}(Y).
\end{align*} 
In the above sequence, the first inequality is because of the choice of the set $Z$ (i.e., $Z$ is an optimum solution to $\max\{p(Z)-b_{(G, w)}(Z): Z\in \calZ\}$ and $Z^v \in \calZ$), the first equality is because  $b_{(G, w)}(Z^v) = d_{(G^v, w^v)}(Z^v)$ since $Z^v\subseteq V-\{v\}$, the second inequality is by definition of $Z^v$, and the second equality is because $b_{(G, w)}(Y) = d_{(G^v, w^v)}(Y)$ since $Y\subseteq V-\{v\}$. Thus, all inequalities in the above sequence are equalities, and we have that $p(Z)-b_{(G, w)}(Z) = p(Y)-b_{(G, w)}(Y)$. Thus, in both cases, the set $Z$ returned by our procedure is an optimum solution to \eqref{eq:wc-oracle-from-sc-oracle}.

We now analyze the runtime of the procedure and the number of \functionMaximizationOracleStrongCover{p} queries. Let $u \in V$ be a vertex. Constructing the hypergraph $(G^{u}, w^{u})$ takes $O(|V| + |E|)$. Thus, a tuple $(Z^{u}, p(Z^u))$ takes $O(|V|+|E|)$ time and one query to \functionMaximizationOracleStrongCover{p}. The runtime includes the time to construct the inputs to the query. Moreover, the input hypergraph to the query has $|V|$ vertices and at most $|E|$ hyperedges. Thus, if $T \not = \emptyset$, the procedure makes a single \functionMaximizationOracleStrongCover{p} query and runs in $O(|V| + |E|)$ time. Alternatively, if  $T  = \emptyset$, then the procedure makes $|V| + 1$ \functionMaximizationOracleStrongCover{p} queries and runs in $O(|V|(|V| + |E|))$ time. In both cases, the runtime includes the time to construct the inputs for the \functionMaximizationOracleStrongCover{p} queries. Moreover, the input hypergraph to each query has $|V|$ vertices and at most $|E|$ hyperedges.
\end{proof}

\subsection{Skew-supermodularity of a function \(p\)}\label{sec:helper-for-p-max-oracle}
In this section, we prove \Cref{lemma:helper-for-p-max-oracle}. 
\lempmaxoracle*
\begin{proof}
    For ease of notation, let $p=p_{(G, c, r)}$. 
    The function $R:2^V\rightarrow \Z_{\ge 0}$ is symmetric skew-supermodular (see Lemma 8.1.9 of \cite{Frank-book}) and the function $d_{(G,c)}:2^V\rightarrow \Z_{\ge 0}$ is symmetric submodular. The function $p$ is obtained by subtracting a symmetric submodular function from a symmetric skew-supermodular and is hence, symmetric skew-supermodular. 

    Next, we show that \functionMaximizationOracleStrongCover{p} can be implemented in strongly polynomial-time. 
    Let the input to the oracle be the hypergraph $(G_0=(V, E_0), c_0: E_0\rightarrow \Z_+)$ and disjoint subsets $S_0, T_0\subseteq V$. 
    We construct the hypergraph $(G+G_0, c+c_0)$. Moreover, for distinct vertices $u, v\in V$ with $u\in V-T_0$ and $v\in V-S_0$, we compute  
    \begin{align*}
    \lambda_{(G+G_0, c+c_0)}(u, v)&\coloneqq \min\left\{d_{(G+G_0, c+c_0)}(Z): S_0\cup \{u\}\subseteq Z\subseteq V-(T_0\cup\{v\})\right\} 
    \end{align*}
    and the set $Z(u, v)$ that achieves the minimum for $\lambda_{(G+G_0, c+c_0)}(u, v)$. We note that $\lambda_{(G+G_0, c+c_0)}(u, v)$ and $Z(u,v)$ can be computed using a single call to a min $(s,t)$-cut oracle on the hypergraph $(G+G_0, c+c_0)$. This hypergraph has $|V|$ vertices and at most $|E_0|+|E|$ hyperedges.  

    To implement the \functionMaximizationOracleStrongCover{p}, we observe that 
    \begin{align*}
    \max&\left\{p(Z)-d_{(G_0, c_0)}(Z): S_0\subseteq Z\subseteq V-T_0\right\}\\
    &=\max\left\{R(Z)-d_{(G, c)}(Z)-d_{(G_0, c_0)}(Z): S_0\subseteq Z\subseteq V-T_0\right\}\\
    &=\max\left\{r(u, v)-\lambda_{((G, c), (G_0, c_0))}(u, v): u\in V-T_0, v\in V-S_0, u\neq v\right\}. 
    \end{align*}
    The RHS problem can be solved by iterating over all distinct vertices $u\in V-T_0$ and $v\in V-S_0$ and returning a pair for which the objective function $r(u, v)-\lambda_{((G, c), (G_0, c_0))}(u, v)$ is maximum. Let $u\in V-T_0$ and $v\in V-S_0$ be a pair of vertices that achieves the maximum in the RHS problem. Then, we observe that the set $Z(u, v)$ is a maximizer for the LHS problem. Thus, it suffices to return $Z(u, v)$ and the value $p(Z(u, v))$. We observe that the value $p(Z(u, v))$ can be computed by subtracting $d_{(G_0, c_0)}(Z(u,v))$ from the objective value of the RHS problem.

    The above procedure can be implemented in $O(|V|^2(|V| + |E| + |E_0|))$ time and $O(|V|^2)$ calls to a min $(s,t)$-cut oracle on hypergraphs with $|V|$ vertices and $|E| + |E_0|$ hyperedges. Then, the claimed runtime is obtained by using a crude upper bound on the run-time of standard deterministic min $(s,t)$-cut algorithms for hypergraphs---we recall that min $(s,t)$-cut in a $n$-vertex $m$-hyperedge hypergraph can be solved in $O((m+n)mn)$.
\end{proof}

\subsection{Computing the Family of Maximal \((p, H, w)\)-Tight Sets}\label{appendix:sec:Function-Maximization-Oracles:computing-maximal-pHw-tight-sets-family}
In this section, we prove \Cref{lem:FunctionMaximizetionOracles:computing-maximal-pHw-tight-sets-with-p-max-oracle}. 
We recall that for a set function $p:2^V\rightarrow\Z$ and hypergraph $(H = (V, E), w)$, a set $X\subseteq V$ is $(p, H, w)$-tight if $p(X) = b_{(H,w)}(X)$. The family $\maximalTightSetFamily{p, H, w}$ denotes the family of all inclusion-wise maximal $(p, H, w)$-tight sets. In this section, we show that the family $\maximalTightSetFamily{p, H, w}$ can be computed in $O(|V|^2(|V| + |E|))$ time using $O(|V|^2)$ queries to \functionMaximizationEmptyOracle{p} if the function $p$ and hypergraph $(H, w)$ satisfy certain properties.

\lemComputingMaximalTightSetFamily*
\begin{proof}
We describe a procedure that recursively constructs \maximalTightSetFamily{p, H, w}. The procedure finds a $(p, H, w)$-tight set $X\subseteq V$, and then searches for a maximal $(p, H, w)$-tight set $Y$ containing the set $X$ --- we describe this computation in the paragraph below. 
The procedure is recursive and it reaches the base case if 
the computed set $X$ is not a $(p, H, w)$-tight set or $X,Y = \emptyset$. If $X$ is not a $(p, H, w)$-tight set, then the procedure returns the empty set, and if $X,Y = \emptyset$, then the procedure returns $\{\emptyset\}$.
If both of the base case conditions are not satisfied, the procedure recurses on the function $p':2^{V'}\rightarrow\Z$ defined on the ground set $V':=V - Y$ as $p'(Z) = p(Z)$ for every $Z \subseteq V'$, and the hypergraph $(H' = (V', E'), w')$ where $E' := \{e \cap V' : e \in E\}$ and $w'(e'):= \sum_{e \in E}w(e)\cdot\indicator_{e\cap V' = e'}$ for every $e' \in E'$, to obtain the family $\maximalTightSetFamily{p', H', w'}$.
It then returns the family $\{Y\} \cup (\maximalTightSetFamily{p', H', w'} - \{\emptyset\})$. 
The procedure terminates within $|V|$ recursive calls because the size of the ground set decreases during every recursive call (i.e. $|V'| < |V|$). Furthermore, the correctness of the procedure follows from disjointness of the family $\maximalTightSetFamily{p, H, w}$ by \Cref{lem:weak-to-strong:uncrossing:maximal-tight-set-family-disjoint} and because $b_{(H,w)}(X) = b_{(H', w')}(X)$ for each $X \subseteq V'$, where $(H', w')$ is the hypergraph constructed as input to the subsequent recursive call.
    
We now bound the run-time and the number of queries required by the procedure described above. Consider a recursive call of the procedure.
We note that a $(p, H, w)$-tight set exists if and only if the set returned by the query $\functionMaximizationOracle{p}((G_0, c_0):= (H, w), S_0:=\emptyset, T_0:=\emptyset)$ is a $(p, H, w)$-tight set.
Furthermore, verifying whether the returned set $X$ is $(p, H, w)$-tight can be done in $O(|V| + |E|)$ time by checking whether $p(X) = b_{(H,w)}(X)$.
Thus, a $(p, H, w)$-tight set $X\subseteq V$ can be found in  $O(|V| + |E|)$ time using one query to \functionMaximizationOracle{p}. Moreover, the input hypergraph to \functionMaximizationOracle{p} query has $|V|$ vertices and $|E|$ hyperedges.
Next, \Cref{claim:FunctionMaximizetionOracles:maximal-tight-set-containing-given-tight-set} below shows that we can find a maximal $(p, H, w)$-tight set $Y$ containing the set $X$ in $O(|V|(|V|+|E|))$ time using $O(|V|)$ queries to \functionMaximizationOracle{p} --- here, the runtime includes the time required to construct the inputs to \functionMaximizationOracle{p}; moreover, each query to \functionMaximizationOracle{p} takes the hypergraph $(H, w)$ as input.
Let $p':2^{V'}\rightarrow\Z$ be the input function to the subsequent recursive call. Let $(G = (V', E_G), c)$ be a hypergraph such that $|E_G| \leq |E|$ and $S, T \subseteq V'$ be disjoints sets. 
Then, \Cref{claim:FunctionMaximizetionOracles:computing-maximal-tight-sets:maximization-oracle-for-restriction} below shows that a query to $\functionMaximizationOracle{p'}((G, c), S, T)$ can be implemented in $O(|V| + |E|)$ time using one query to \functionMaximizationOracle{p} --- here, the runtime includes the time required to construct the input to \functionMaximizationOracle{p} query; moreover the hypergraph input to \functionMaximizationOracle{p} query has $|V|$ vertices and $|E_G|$ hyperedges. We note that all hypergraphs $(G = (V', E_G), c)$ that are used as inputs to \functionMaximizationOracle{p'} during subsequent recursive calls satisfy $|E_G| \leq |E|$.
Thus, every recursive call of the procedure can be implemented to run in $O(|V|(|V|+|E|))$ time and make $O(|V|)$ calls to \functionMaximizationOracle{p}, where the runtime includes the time to construct the inputs to each query to \functionMaximizationOracle{p}. Consequently, the procedure runs in time $O(|V|^2(|V|+|E|))$ and makes $O(|V|^2)$ queries to $\functionMaximizationOracle{p}$, where the runtime includes the time to construct the inputs to the \functionMaximizationOracle{p} queries.
\end{proof}

\begin{claim}\label{claim:FunctionMaximizetionOracles:maximal-tight-set-containing-given-tight-set}
         Let $X \subseteq V$ be a $(p, H, w)$-tight set. Then, a maximal $(p, H, w)$-tight set containing the set $X$ can be computed in $O(|V|(|V|+|E|))$ time using $O(|V|)$ queries to \functionMaximizationOracle{p}. The runtime includes the time required to construct the inputs to \functionMaximizationOracle{p}. Moreover, each query to \functionMaximizationOracle{p} takes the hypergraph $(H, w)$ as input.
\end{claim}
\begin{proof}
We describe a recursive procedure to compute a maximal $(p, H, w)$-tight set containing the set $X$. For every $u \in V - X$, let $Z^u$ be the set returned by querying $\functionMaximizationOracle{p}((H_0, w_0):= (H, w), S_0 := X\cup\{u\}, T_0 := \emptyset)$. If there exists a vertex $u \in X$ such that the set $Z^u$ is a $(p, H, w)$-tight set, then recursively compute a $(p, H, w)$-tight set containing the set $X':= Z^u$. Otherwise, return the set $X$. 

The correctness of the above procedure is because of the following. The above procedure terminates (in at most $|V|$ recursive calls) because the cardinality of the input set $X$ increases during each recursive call (i.e., $|X'| > |X|$). Let $Y \subseteq V$ be a maximal $(p, H, w)$-tight set such that $X \subseteq Y$. If $Y = X$, then there are no $(p, H, w)$-tight sets that strictly contain the set $X$ and the procedure returns the set $X$. Suppose that $X \subsetneq Y$ and let $u \in Y - X$. Then, there exists a $(p, H, w)$-tight set $Z$ such that $X\cup\{u\} \subsetneq Z \subseteq Y$. Thus, the set $Z^u$ computed by the procedure will be a $(p, H, w)$-tight set that strictly contains the set $X$. Using this fact, the procedure can be shown to return a maximal $(p, H, w)$-tight set containing the set $X$ by induction on $|X|$.   
    
We now describe a slight modification of the above procedure that leads to an improved runtime and a smaller number of queries to \functionMaximizationOracle{p}. Let $u \in V - X$ be a vertex. If the set $Z^{u}$ is $(p, H, w)$-tight, then we include $u$ in the input set $S_0$ to all $\functionMaximizationOracle{p}$ queries in subsequent recursive calls of the procedure. Alternatively, if the set $Z^{u}$ is not $(p, H, w)$-tight, then we include $u$ in the input set $T_0$ to all $\functionMaximizationOracle{p}$ queries in subsequent recursive calls procedure. Consequently, the number of queries to $\functionMaximizationOracle{p}$ (and number of recursive calls of the procedure) is at most $|V|$. We note that constructing the input to each query and verifying whether the set returned by the query is $(p, H, w)$-tight can be done in $O(|V| + |E|)$ time. Thus, the procedure terminates in $O(|V|(|V|+|E|))$ time. Here, the runtime includes the time to construct the inputs to the \functionMaximizationOracle{p} queries. Moreover, each query takes the hypergraph $(H,w)$ as input. 
\end{proof}
    
\begin{claim}\label{claim:FunctionMaximizetionOracles:computing-maximal-tight-sets:maximization-oracle-for-restriction}
    Let $p:2^V\rightarrow\Z$ be a function. Let $Y \subseteq V$ and let $p'$ denote the function $p$ restricted to the ground set $V':=V - Y$. Then, for a given hypergraph $(G = (V', E_G), c:E_G\rightarrow\Z_+)$ such that $|E_G| \leq |E|$ and disjoint subsets $S, T \subseteq V'$, $\functionMaximizationOracle{p'}((G,c), S, T, y)$ can be computed in $O(|V| + |E|)$ time using one query to $\functionMaximizationOracle{p}$.
    The runtime includes the time required to construct the input to \functionMaximizationOracle{p} query. Moreover, the hypergraph input to \functionMaximizationOracle{p} query has $|V|$ vertices and $|E_G|$ hyperedges.
\end{claim}
\begin{proof}
    Let $(G_1, c_1)$ be the hypergraph obtained from $(G,c)$ by adding the vertices $Y$. Then, the answer to the query $\functionMaximizationOracle{p'}((G, c), S, T)$ is the same as the answer to the query $\functionMaximizationOracle{p}((G_1, c_1), S, T\cup Y)$.
\end{proof}